\documentclass[a4paper,onecolumn,10pt,accepted=2025-03-10]{quantumarticle}
\pdfoutput=1
\usepackage{algorithm}
\usepackage[utf8]{inputenc}
\usepackage[english]{babel}
\usepackage[T1]{fontenc}
\usepackage{amsmath}
\usepackage{listings}
\usepackage{physics}
\usepackage{amsfonts}
\usepackage{amssymb}
\usepackage{mathrsfs}
\usepackage{bbm}
\usepackage{amsthm}
\usepackage{mathtools}
\usepackage{algpseudocode}
\usepackage{graphicx}  % Para manejar imágenes
\usepackage{subfig}   
\usepackage{mathdots}

\usepackage{tikz}
\usepackage{lipsum}
\usepackage[resetlabels]{multibib}

\usepackage{hyperref}
\usepackage{cleveref}

\newtheorem{theorem}{Theorem}
\newtheorem{lemma}[theorem]{Lemma}

\newtheorem{proposition}[theorem]{Proposition}

\newtheorem{remark}{Remark}

\newcommand{\Si}{s_{\mathrm{i}}}
\newcommand{\Sf}{s_{\mathrm{f}}}
\newcommand{\Sb}{\boldsymbol{s}}
\newcommand{\bs}{\Sb}

\newcommand{\zb}{\boldsymbol{z}}

\newcommand{\ii}{\mathrm{i}}
\newcommand{\Id}{\mathrm{Id}}
\newcommand{\ff}{\mathrm{f}}
\newcommand{\e}{\mathrm{e}}
\newcommand{\sgn}{\operatorname{sgn}}
\newcommand{\FGA}{\mathrm{FGA}}
\newcommand{\G}[2]{G_{k_1k_2}^{#1,#2}}

\begin{document}

\title[Inchworm method with Frozen Gaussian Approximation]{Solving Caldeira-Leggett Model by Inchworm Method with Frozen Gaussian Approximation}

\author{Geshuo Wang}
\email{geshuo@uw.edu}
\affiliation{Department of Applied Mathematics, University of Washington, Seattle, WA 98195, USA}
\orcid{0000-0001-7890-2758}
\thanks{The research was conducted while G. Wang was affiliated with the National University of Singapore.}

\author{Siyao Yang}
\email{siyaoyang@uchicago.edu}
\affiliation{Committee on Computational and Applied Mathematics, Department of Statistics, University of Chicago, Chicago, IL 60637 USA}
\orcid{0000-0002-6651-6224}
\thanks{Corresponding author.}

\author{Zhenning Cai}
\email{matcz@nus.edu.sg}
\affiliation{Department of Mathematics, National University of Singapore, Singapore 119076}
\orcid{0000-0002-7086-7983}
\thanks{The work of Zhenning Cai was supported by the Academic Research Fund of the Ministry of Education of Singapore under Grant
No. A-8000965-00-00.}

\maketitle

\begin{abstract}
We propose an algorithm that combines the inchworm method and the frozen Gaussian approximation to simulate the Caldeira-Leggett model in which a quantum particle is coupled with thermal harmonic baths. In particular, we are interested in the real-time dynamics of the reduced density operator. In our algorithm, we use frozen Gaussian approximation to approximate the wave function as a wave packet in integral form. The desired reduced density operator is then written as a Dyson series, which is the series expression of path integrals in quantum mechanics of interacting systems. To compute the Dyson series, we further approximate each term in the series using Gaussian wave packets, and then employ the idea of the inchworm method to accelerate the convergence of the series. The inchworm method formulates the series as an integro-differential equation of ``full propagators'', and rewrites the infinite series on the right-hand side using these full propagators, so that the number of terms in the sum can be significantly reduced, and faster convergence can be achieved. The performance of our algorithm is verified numerically by various experiments.
\end{abstract}

\section{Introduction}
The time evolution of a quantum system follows the time-dependent Schrödinger equation.
The numerical simulation of Schrödinger equation is an important topic in scientific computing.
Many effective and accurate methods \cite{leforestier1991comparison,iitaka1994solving,bao2002time,bao2003numerical,vandijk2007accurate} have been developed to solve the time-dependent Schrödinger equation numerically.
Among these methods, the Gaussian beam (GB) method \cite{heller1975time,leung2007eulerian,leung2009eulerian} decomposes the wave function to Gaussian beams and evolves the parameters of Gaussian beams instead of the wave function itself.
The frozen Gaussian approximation (FGA) \cite{heller1981frozen,lu2011frozen,lu2012convergence,lu2012frozen,lu2018frozen,delgadillo2018frozen}{, which is highly related to the Herman-Kluk propagator \cite{herman1984semiclassical,swart2009mathematical}}, modifies the GB method to maintain the width of the beams.
By preventing wave spreading, the FGA achieve higher accuracy than the GB method.

Although numerical methods have been successful for closed quantum systems, real quantum systems are more or less affected by the environment so that these numerical methods cannot apply directly.
Therefore, the theory of open quantum systems, in recent years, has attracted increasingly more attention.
Due to its universality, the theory of open quantum systems has been applied in many fields including but not limited to quantum computing \cite{knill1997theory}, quantum communication \cite{nielsen2002quantum}, and quantum optical systems \cite{breuer2002theory}.
In the context of an open quantum system, a system of interest is immersed in some uninteresting environment or bath.
The coupling between the system and the environment leads to quantum dissipation and quantum decoherence \cite{grigorescu1998decoherence,schlosshauer2019quantum}.
The dynamics of the system part therefore becomes a non-Markovian process, which is the main difficulty for the simulation of open quantum systems.
Time non-locality typically entails significant memory or computational expenses, in particular for long time simulation.

The Nakajima-Zwanzig equation \cite{nakajima1958quantum,zwanzig1960ensemble} depicts the temporal nonlocality by a memory kernel.
In the case of weak system-bath coupling, the process can be approximated by a Markovian process with Lindblad equation \cite{lindblad1976generators}.
In more general cases, however, it would be difficult to bypass the non-Markovian nature during the numerical simulation.
Some methods therefore introduce a memory length to prevent unlimited growth of storage or computational cost.
Based on the path integral technique \cite{feynman1948space}, the iterative quasi-adiabatic propagator path integral (i-QuAPI) has been proved to be an efficient method \cite{makri1992improved,makri1995numerical}.
In recent years, different techniques are introduced to improve the method \cite{makri2014blip,wang2022differential,makri2024kink}.
Due to the non-Markovian intrinsicality,
 path integral methods typically face high memory costs to store the contributions of various paths within the memory length.
Recently, the small matrix path integral (SMatPI) \cite{makri2020smallMatrixPath,makri2021smallMatrixPathIntegralExtended,wang2024tree} has successfully overcome this problem by summing up the contribution of all paths into small matrices.
Although the starting points are different, SMatPI turns out to be in a similar form of the Nakajima-Zwanzig equation.
The transfer tensor method (TTM) \cite{cerrillo2014nonmarkovian} derives a discretized form of Nakajima-Zwanzig equation and also introduces a memory length to keep memory cost low.
% \delet{
% SMatPI can therefore be considered as a detailed implementation of TTM.}

If the interacting system is regarded as a perturbation of the uncorrelated system, one can derive the Dyson series expansion \cite{dyson1949radiation,xu2018convergence} with a series of high-dimensional integrals.
Monte Carlo sampling is generally used for the integrals in the series.
Due to the high {oscillation of} the integrand, direct application of the Monte Carlo method suffers from numerical sign problems \cite{loh1990sign,cai2023numerical}.
Recently, the inchworm algorithm was developed to tame the sign problem \cite{chen2017inchwormITheory,chen2017inchwormIIBenchmarks,Antipov2017currents,cai2020inchworm,erpenbeck2023quantum} and has been proved to be an efficient method.
Inchworm algorithm utilizes {bold lines} to compute the partial sum in the series so that the series converges faster \cite{prokof2007bold,prokof2008bold}.
Recently, some other methods utilize the intrinsic invariant properties \cite{yang2021inclusion,cai2022fast,cai2023bold} to accelerate the computation.

All these methods mainly focus on the simulation of open quantum systems where {the dimension of the system space is finite.}
A typical example is the spin-boson model \cite{chakravarty1984dynamics,thorwart2004dynamics,xu2023performance}, {which describes a single spin} living in a two-dimensional Hilbert space.
Spin-boson model can be regarded as an implementation of Caldeira-Leggett \cite{caldeira1981influence,caldeira1983path,caldeira1985influence} model where the particle locates in a double well potential.
For open systems with more possible states, current research mainly focuses on those with special structure like open spin chains \cite{makri2018modular,wang2023real,sun2024simulation}.
In such systems, {matrix product state representations are used to make the simulation feasible, but are difficult to extend beyond 1D} \cite{yan2021efficient,erpenbeck2023tensor,wang2023real,sun2024simulation}.

In this paper, we propose a method that combines the FGA and the inchworm algorithm for the simulation of quantum particles in the Caldeira-Leggett model. Different from models with finite possible states, the density operator for Caldeira-Leggett model is an operator on the infinite dimensional Hilbert space.
The density operator is given by a function with two variables but not in a simple matrix form.
In the numerical simulations, we focus on the diagonal of the density operator, i.e., the probability density of the position operator.

In \cref{sec_preliminaries}, we simply review the frozen Gaussian approximation and the Caldeira-Leggett model. 
In \cref{sec_FGA_inchworm}, a brand new method for the simulation of Caldeira-Leggett model is proposed.
Specifically, \cref{sec_dyson} gives the Dyson expansion for the Caldeira-Leggett model.
\Cref{sec_inchworm} integrates the FGA and the inchworm algorithm in the Caldeira-Leggett model, providing an algorithm for the simulation.
\Cref{sec_summary} gives a summary of our algorithm.
In \cref{sec_numerical_examples}, some results for numerical tests are shown to validate our method.
A quick summary and some discussions are given in \cref{sec_conclusion}.

\section{Preliminaries}
\label{sec_preliminaries}
\subsection{Frozen Gaussian approximation}
\label{sec_FGA}
Before introducing our work on open quantum systems,
 we first present the frozen Gaussian approximation for the one-dimensional Schrödinger equation in the context of closed quantum systems.

Consider a one-dimensional quantum particle in a potential $V(x)$, 
the time-dependent Schr\"odinger equation can be written as
\begin{equation}
\label{eq_Schrödinger}
    \ii \epsilon \pdv{\psi(t,x)}{t}
    = H \psi(t,x)
\end{equation}
with the Hamiltonian $H$ being
\begin{equation}
    H =  -\frac{\epsilon^2}{2} \pdv[2]{x}
    + V(x),
\end{equation}
and $\ii$ being the imaginary unit.
The constant $\epsilon$ is a parameter depending on the nondimensionalization,
 which is usually small when the scale of the system is not too small. The FGA \cite{lu2012frozen} utilizes the following ansatz
\begin{equation}
\label{eq_FGA_ansatz}
    \psi_{\FGA}(t,x)
    = \frac{1}{({2\pi\epsilon})^{3/2}}
    \int_{\mathbb{R}^2}
    \int_{-\infty}^\infty
    a(t,p,q) \e^{\ii \phi(t,x,y,p,q) / \epsilon}
    \psi_0(y)
    \dd y \dd p \dd q
\end{equation}
where
\begin{equation}\label{phi}
    \phi(t,x,y,p,q) = S(t,p,q) + \frac{\ii}{2}(x-Q)^2 + P(x-Q) + \frac{\ii}{2}(y-q)^2 - p(y-q)
\end{equation}
and $\psi_0$ is the initial wave function of the particle.
The variables $P, Q, S, a$ satisfy the following dynamics:
\begin{equation}
\label{eq_dynamics}
\begin{split}
    &\pdv{P(t,p,q)}{t} = - V'(Q(t,p,q)), \\
    &\pdv{Q(t,p,q)}{t} = P(t,p,q), \\
    &\pdv{S(t,p,q)}{t} = \frac{1}{2} \left(P(t,p,q)\right)^2 - V(Q(t,p,q)), \\
    &\pdv{a(t,p,q)}{t} = \frac{1}{2}a(t,p,q) \frac{\partial_q P - \partial_p Q V''(Q) - \ii\left(\partial_p P +\partial_q Q V''(Q)\right)}{\partial_q Q + \partial_p P + \ii (\partial_q P - \partial_p Q)}
\end{split}
\end{equation}
with initial conditions 
\begin{equation}
    P(0,p,q) = p, 
    \quad
    Q(0,p,q) = q,
    \quad
    S(0,p,q) = 0,
    \quad
    a(0,p,q) = \sqrt{2}.
\end{equation}
To solve these differential equation system, the following four equations should be also included
\begin{equation}
\begin{split}
    \pdv{t} \left(\partial_p P\right) =  -\partial_p QV''(Q) , \quad 
    \pdv{t} \left(\partial_p Q\right) =  \partial_p P, \\
    \pdv{t} \left(\partial_q P\right) =  -\partial_q QV''(Q), \quad
    \pdv{t} \left(\partial_q Q\right) = \partial_q P.
\end{split}
\end{equation}
This ansatz is derived based on some Fourier integral operators (refer to \cite{swart2009mathematical} for more details).
The FGA method gives an approximation solution to the Schr\"odinger equation, which can be concluded by the following proposition.
\begin{proposition}
% \ki{This proposition also needs modification.}
\label{prop_FGA}
    The FGA ansatz \cref{eq_FGA_ansatz} can be regarded as the solution to the Schr\"odinger equation \cref{eq_Schrödinger} with first order accuracy.
    In particular, we have
    \begin{equation}
    \label{eq_FGA_approx_sol_of_real}
        \Vert \psi_{\FGA}(t,x) - \e^{-\ii H t /\epsilon} \psi_0(x) \Vert_{L^2} \leqslant C_1\epsilon
    \end{equation}
    for some constant $C_1$ if the parameters in $\psi_{\FGA}(t,x)$ follow the dynamics \cref{eq_dynamics}.
\end{proposition}

% The ansatz $\psi_{\FGA}(t,x)$ can be regarded as an approximation of the solution, \emph{i.e.},
% \begin{equation}
% \label{eq_FGA_approx_sol_of_real}
%     \psi_{\FGA}(t,x) \approx \psi(t,x) = \e^{-\ii H t / \epsilon} \psi(0,x).
% \end{equation}
% This approximation has been proved to have accuracy $\mathcal{O}(\epsilon)$ 
% \cite{swart2009mathematical,robert2010herman,lu2012convergence,huang2023efficient}. \ki{Is it correct? To be double-checked...}
The proof of this lemma can be found in the literature \cite{swart2009mathematical} discussing the frozen Gaussian approximation.
With \cref{eq_FGA_approx_sol_of_real},
 solving the Schrödinger equation can be converted to solving ordinary differential equations \cref{eq_dynamics}.
Although the FGA can also be applied in high dimensional systems, we only consider the one dimensional case in this paper. With the FGA, the probability density $f(t,x)$ to observe the particle at position $x$ at time $t$ when measurement is performed can be computed by
\begin{displaymath}
f(t,x) \approx \vert \psi_{\FGA}(t,x)\vert^2.
\end{displaymath}
% \begin{equation}
%     \rho(t,x_1,x_2) \approx \psi_{\FGA}(t,x_1) {\psi_{\FGA}^*(t,x_2)}, 
%     \quad
%     f(t,x) \approx \vert\psi_{\FGA}(t,x)\vert^2.
% \end{equation}
The numerical simulation based on the FGA generally introduces a proper discretization for variables $p$ and $q$ in \cref{eq_FGA_ansatz}.
We will explain some more details later in \cref{sec_inchworm}.

% and the ansatz therefore becomes
% \begin{equation}\label{FGA}
%     \psi_{\FGA}(t,x)
%     = \frac{1}{\sqrt{2\pi\epsilon}}
%     \sum_{k=1}^{N^p}
%     \sum_{l=1}^{N^q}
%     \int_{-\infty}^\infty
%     a_{kl}(t) \e^{\ii\phi_{kl}(t,x,y)/\epsilon} \psi_0(y) \dd y
%     \Delta p \Delta q
% \end{equation}
% with
% \begin{equation}
%     \phi_{kl}(t,x,y) = S_{kl} + \frac{\ii}{2}(x-Q_{kl})^2 + P_{kl}(x-Q_{kl}) + \frac{\ii}{2} (y-q_l)^2 - p_k(y-q_l)
% \end{equation}

\subsection{Caldeira-Leggett model}
\label{sec_CL_model}
Before we introduce the model we study in this paper, we give a general picture for the open quantum system.
We consider the Sch\"odinger equation \cref{eq_Schrödinger} which describes a single particle coupled with a bath, where the wave function for the whole system takes the form $\psi(t,x,\zb)$ where $x$ is the position of the particle we are interested in and $\zb$ are the degrees of freedom in the bath space. The Hamiltonian is formulated as 
\begin{equation}
\label{eq_Hamiltonian_oqs}
    H = H_s \otimes \mathrm{Id}_b + \mathrm{Id}_s \otimes H_b + W
\end{equation}
where $H_s, H_b$ are respectively the Hamiltonians associated with the system and bath, and $\mathrm{Id}_s, \mathrm{Id}_b$ are respectively the identity operators for the system space and bath space. $W$ is the coupling term between system and the bath taking the form 
\begin{equation}\label{coupling}
W = W_s \otimes W_b.
\end{equation}

In the context of open quantum systems, instead of working with the wave function, it would be easier to work with the density operator {$\rho(t)$}. {We assume that initially the system is in a pure state, the bath is in thermal equilibrium, and the system is not entangled with the baht. We can then formulate the initial density operator as}
\begin{equation}
 \label{init_rho}
\rho_0 = \rho^{(0)}_s \otimes \rho^{(0)}_b :=  {\ket{\psi_{s}^{(0)}}\bra{{\psi_{s}^{(0)}}} } \otimes \frac{\e^{-\beta
H_b}}{\tr(\e^{-\beta
H_b})}
\end{equation}
where $\beta$ is the inverse temperature.

Furthermore, we introduce the reduced density operator 
\begin{equation}
\label{eq_RDM_definition}
    \begin{split}
      {\rho_s(t) := \tr_b\left(\rho(t)\right) }= \tr_b\left( \e^{-\ii H t/\epsilon} \rho_0 \e^{\ii H t/ \epsilon} \right)
    \end{split}
\end{equation}
where $\tr_b$ represents the partial trace operator with respect to the bath degree of freedom, and the second equality above is from the von Neumann equation that governs the evolution of the density operator. Our goal is to compute the marginal density $f(t,x)$ upon integrating out the uninteresting bath, which is now obtained by the ``diagonal" of the reduced density operator  
\begin{equation}
\label{eq_density_probability}
    f(t,x) = {\bra{x}\rho_s(t)\ket{x}}.
\end{equation}

This paper proposes an algorithm for a basic open quantum system model, the Caldeira-Leggett model \cite{caldeira1981influence,caldeira1983path,caldeira1985influence,ferialdi2017dissipation}.
% We propose an algorithm that combines FGA with sum of diagrams as discussed in Section \ref{sec_preliminaries} to estimate the reduced density operator for Caldeira-Leggett model. 
% To present our algorithm, we consider a fundamental model in open quantum system known as the Caldeira-Leggett model \cite{caldeira1981influence,caldeira1983path,caldeira1985influence,ferialdi2017dissipation}, 
It describes the dynamics of a quantum particle coupled with $L$ quantum harmonic oscillators as the bath. For simplicity, we consider a one-dimensional particle throughout this work. 
% In the Caldeira-Leggett model, the wave function for the whole system now takes the forms $\psi(t,x,\boldsymbol{z}): \mathbb{R}^+ \times \mathbb{R} \times \mathbb{R}^{L} \rightarrow \mathbb{C}$, and the Hamiltonian is formulated as 
% \begin{equation}
%     H = H_s \otimes \mathrm{Id}_b + \mathrm{Id}_s \otimes H_b + W
% \end{equation}
% where $ H_s, H_b$ are respectively the Hamiltonians associated with the particle and harmonic oscillators 
% while $W$ is the coupling term between system and the bath.
In the Caldeira-Leggett model, the Hamiltonians in \cref{eq_Hamiltonian_oqs} can be formulated as 
% In particular,
\begin{equation}\label{H}
    H_s = -\frac{\epsilon^2}{2} \pdv[2]{x}
    + V(x) + \sum_{j=1}^L\frac{c_j^2}{2\omega_j^2} \hat{x}^2 \quad \text{and} \quad     H_b = \sum_{j=1}^L -\frac{\epsilon^2}{2} \pdv[2]{z_j}
    + \frac{1}{2} \omega_j^2 \hat{z}_j^2,
\end{equation}
and $W = W_s \otimes W_b$ represents the coupling between the particle and harmonic oscillators where 
\begin{equation}\label{W}
     W_s = \hat{x}  \quad \text{and} \quad  W_b =\sum_{j=1}^L c_j \hat{z}_j.  
\end{equation}
The notations in \cref{H} and \cref{W} are defined as:
\begin{itemize}
\item $\hat{x}$: The position operator of the particle in system, {$ \bra{x,\boldsymbol{z}} \ket{\hat{x}\psi(t)} =  x  \bra{x,\boldsymbol{z}} \ket{\psi(t)}$}.
\item $\hat{z}_j$: The position operator of the $j$th harmonic oscillator, {$ \bra{x,\boldsymbol{z}} \ket{\hat{z}_j\psi(t)} = z_j  \bra{x,\boldsymbol{z}} \ket{\psi(t)}$.}
\item $\omega_j$: The frequency of the $j$th harmonic oscillator.
\item $c_j$: The coupling intensity between the particle and the $j$th harmonic oscillator.
\item $V$: The potential function which is real and smooth.
\end{itemize}

As a fundamental model, the Caldeira-Leggett model includes the major challenges in simulating open quantum systems, which is the large degree of freedom in the harmonic oscillators. 
Another well-known open quantum system model, the spin-boson model
% The Caldeira-Leggett model is also relevant to the spin-boson model 
\cite{chakravarty1984dynamics,fisher1985dissipative,leggett1987dynamics,hanggi1986escape} is also related to the Caldeira-Leggett model.
Specifically, it can be regarded as a special case of the Caldeira-Leggett model in double well potential with {an additional approximation neglecting high-energy system states}. 

\section{Inchworm-FGA algorithm for Caldeira-Leggett model}
\label{sec_FGA_inchworm}
\subsection{Dyson series for Caldeira-Leggett model}
\label{sec_dyson}
To obtain the reduced density operator, it is in practice unaffordable to directly evaluate the operator $\e^{\pm \ii H t/\epsilon}$ due to the large degrees of freedom in the bath. 
{
An existing method is to apply the forward-backward coherent state semiclassical approximation to the influence functional \cite{thompson1999influence}.
Alternatively, one can} express $\e^{\pm \ii H t / \epsilon}$ in terms of Dyson series \cite{dyson1949radiation,cai2020inchworm,cai2022fast} which is widely used in scattering theory. 
For the sake of simplicity, we denote $H_0 = H_s\otimes \Id + \Id\otimes H_b$ so the total Hamiltonian is written by $H = H_0 + W$.
In the Dyson series, the coupling term $W$ is regarded as a perturbation of $H_0$, the non-coupling Hamiltonian.
The exponential operator $\e^{\pm \ii H t/\epsilon}$ can be written as its Dyson series expansion
\begin{equation}
\label{eq_Dyson_expansion_exponential}
  \begin{split}
    \e^{-\ii H t / \epsilon}
    = \ &  \sum_{n=0}^{\infty}
            \int_{0\leqslant \boldsymbol{\tau} \leqslant t} \left(-\frac{\ii}{\epsilon}\right)^n
            \e^{-\ii H_0 (t-t_n)} W
            \e^{-\ii H_0 (t_n-t_{n-1})} W
            \cdots W
            % \e^{-\ii H_0 (t_2-t_1)} W
            \e^{-\ii H_0 t_1} \dd \boldsymbol{\tau}\\
    \e^{\ii H t / \epsilon}
    = \ & \sum_{n=0}^{\infty}
            \int_{0\leqslant \boldsymbol{\tau} \leqslant t} \left(\frac{\ii}{\epsilon}\right)^n
            \e^{\ii H_0 t_1} W
            \e^{\ii H_0 (t_2-t_{1})} W
            \cdots W
            % \e^{-\ii H_0 (t_2-t_1)} W
            \e^{\ii H_0 (t - t_n)} \dd \boldsymbol{\tau}        
            \end{split}
\end{equation}
with $\boldsymbol{\tau}=(t_1,\dots,t_n)$, where we have used the fact that the coupling $W$ in the Caldeira-Leggett model is a Hermitian operator. The integrals over simplices are formulated as 
\begin{equation}
    \int_{t_{\mathrm{i}} \leqslant \boldsymbol{\tau} \leqslant t_{\mathrm{f}} } \, \dd \boldsymbol{\tau} := \int^{t_{\mathrm{f}}}_{t_{\mathrm{i}}}  \int^{t_n}_{t_{\mathrm{i}}}  \cdots \int^{t_{2}}_{t_{\mathrm{i}}}    \, \dd t_1 \cdots  \dd t_{n-1} \, \dd t_n.  
\end{equation}
Inserting \cref{eq_Dyson_expansion_exponential} into \cref{eq_RDM_definition} with change of variables yields the following Dyson expansion of the reduced density operator:{
\begin{equation}
\begin{split}
    \rho_s(t) = \sum_{m=0}^\infty
    \int_{-t\leqslant \boldsymbol{s}\leqslant t}
    \prod_{j=1}^m \left(
        -\frac{\ii}{\epsilon}\sgn(s_j)
    \right)
    &\left(
    G^{(s)}(s_m,t) W_s 
    % G^{(s)}(s_{m-1},s_m) W_s 
    \dots 
    % W_s G^{(s)}(s_1,s_2) 
    W_s G^{(s)}(-t,s_1)
    \right) \times \\
    \times \tr &
    \left(
    G^{(b)}(s_m,t) W_b
    % G^{(b)}(s_{m-1},s_m) W_b
    \dots 
    % W_b G^{(b)}(s_1,s_2) 
    W_b G^{(b)}(-t,s_1)
    \right)
    \dd \boldsymbol{s}
\end{split}
\end{equation}}
where $\boldsymbol{s} = (s_1,\dots,s_m)$ is an ascending sequence of time points and $\sgn$ is the sign function. Under such formulation, the integrand is divided into system part and bath part where
\begin{equation}
\label{eq_bare_propagator}
    G^{(s)}(s_\ii,s_\ff)
    = \begin{cases}
        \e^{-\ii H_s (s_\ff-s_\ii)/\epsilon} ,&\text{if~} 0 \leqslant s_\ii \leqslant s_\ff, \\
        \e^{-\ii H_s s_\ff/\epsilon} {\rho_s^{(0)}}
        \e^{-\ii H_s s_\ii/\epsilon} ,&\text{if~} s_\ii < 0 \leqslant s_\ff, \\
        \e^{\ii H_s (s_\ff-s_\ii)/\epsilon} ,&\text{if~} s_\ii \leqslant s_\ff < 0.
    \end{cases}
\end{equation}
is a propagator associated with the system dynamics and
\begin{equation}
    G^{(b)}(s_\ii,s_\ff)
    = \begin{cases}
        \e^{-\ii H_b (s_\ff-s_\ii)/\epsilon} ,&\text{if~} 0 \leqslant s_\ii \leqslant s_\ff, \\
        \e^{-\ii H_b s_\ff/\epsilon} {\rho_b^{(0)}}
        \e^{-\ii H_b s_\ii/\epsilon} ,&\text{if~} s_\ii < 0 \leqslant s_\ff, \\
        \e^{\ii H_b (s_\ff-s_\ii)/\epsilon} ,&\text{if~} s_\ii \leqslant s_\ff < 0.
    \end{cases}
\end{equation}
is a propagator associated with the bath. Readers may refer to \cite[Section 2]{cai2020inchworm} for more details.

Consequently, the reduced density operator can be written as the following infinite sum over multiple integrals over simplices:
\begin{equation}
\label{eq_dyson}
    {\rho_s(t) 
    = \sum_{m=0}^{\infty}
    \int_{-t\leqslant \boldsymbol{s} \leqslant t}
    \prod_{j=1}^m 
    \left(-\ii \sgn (s_j)\right)
    \mathcal{U}(-t,\boldsymbol{s},t)
    \mathcal{L}_b(\boldsymbol{s})
    \dd \boldsymbol{s}.}
\end{equation}
 Inside the integrand of \cref{eq_dyson}, we have the Markovian system influence functional 
\begin{equation}
\label{eq_system_part_dyson}
   { \mathcal{U}(-t,\boldsymbol{s},t)
    = G^{(s)}(s_{m},t) W_s
    G^{(s)}(s_{m-1},s_m) W_s
    \dots
    W_s G^{(s)}(s_1,s_2) 
    W_s G^{(s)}(-t,s_1) }
\end{equation}
and a non-Markovian bath influence functional 
% \sy{in $\mathcal{L}_b$, will it be better to write down two more terms in the dots like \cref{eq_system_part_dyson}? People might think dots are all $W_b$. }
\begin{equation}
   { \mathcal{L}_b(\boldsymbol{s}) = \frac{1}{\epsilon^m} \tr
    \left(
    G^{(b)}(s_m,t) W_b
    G^{(b)}(s_{m-1},s_m) W_b
    \dots 
    W_b G^{(b)}(s_1,s_2) 
    W_b G^{(b)}(-t,s_1)
    \right).}
\end{equation}
% with 
% % \sy{can we change one of $\leqslant$ to $<$ for the second case below?} \ki{I just change it} \sy{I am thinking to have something like one side $<$ and the other $\leqslant$. In this way, readers will never ask why any point is missing.}\ki{Let's try.}
% \begin{equation}
% \label{eq_bare_propagator}
%     G^{(s)}(s_\ii,s_\ff)
%     = \begin{cases}
%         \e^{-\ii H_s (s_\ff-s_\ii)/\epsilon} &,\text{if~} 0 \leqslant s_\ii \leqslant s_\ff, \\
%         \e^{-\ii H_s s_\ff/\epsilon} \rho_s^{(0)}(x_1,x_2)
%         \e^{-\ii H_s s_\ii/\epsilon} &,\text{if~} s_\ii < 0 \leqslant s_\ff, \\
%         \e^{\ii H_s (s_\ff-s_\ii)/\epsilon} &,\text{if~} s_\ii \leqslant s_\ff < 0.
%     \end{cases}
% \end{equation}
% \begin{equation}
% \label{eq_system_part_dyson}
%     \mathcal{U}_{x_1x_2}(-t,\boldsymbol{s},t)
%     = \e^{-\ii H_s t }
%     \mathcal{T}\left(
%     W_I(s_m) W_I(s_{m-1}) 
%     \dots 
%     W_I(s_2) W_I(s_1)
%     \rho_s^{(0)}(x_1,x_2)
%     \right)
%     \e^{\ii H_s t }
% \end{equation}
% where $\mathcal{T}$ is the time ordering operator that sorts the operators descending with respect to time, and the initial system-associated density operator $\rho_s^{(0)}$ is considered to be at the time zero. 
% The operator $W_I(s)$ is the coupling operator in the interaction picture:
% \sy{in eq(18), $W$ should be $W_s$?}\ki{agree. I have changed}
% \begin{equation}
%     W_I(s) 
%     = 
%     \e^{\ii H_s \vert s \vert } 
%     W_s
%     \e^{-\ii H_s \vert s \vert } .
% \end{equation}
The bath influence functional $\mathcal{L}_b(\boldsymbol{s})$ is assumed to satisfy the Wick's theorem \cite{wick1950evaluation} so that it has the pairwise interaction form 
\begin{equation} \label{Lb}
    \mathcal{L}_b(\boldsymbol{s})
    = \begin{dcases}
        0, & \text{if $m$ is odd,} \\
        \sum_{\mathfrak{q} \in \mathscr{Q}(\boldsymbol{s})} 
        \prod_{(\tau_1,\tau_2) \in \mathfrak{q}} B(\tau_1, \tau_2), & \text{if $m$ is even.}
    \end{dcases}
\end{equation}
where $B(\tau_1,\tau_2)$ is the two-point correlation function $B(\tau_1,\tau_2)$. Under the setting of Caldeira-Leggett model, it takes the form \cite{sun2024simulation}:
\begin{equation}
    B(\tau_1,\tau_2) = \begin{cases}
        \tilde{B}^*(\vert \tau_1 \vert - \vert \tau_2 \vert), &\text{if~}\tau_1\tau_2>0 \\
        \tilde{B}(\vert \tau_1 \vert - \vert \tau_2 \vert), &\text{if~}\tau_1\tau_2\leqslant 0
    \end{cases}
\end{equation}
with 
% \ki{The parameter $\epsilon$ should go into the expression of $B$.}
\begin{equation}
 \begin{split}
    % B(\tau_1,\tau_2) =  
    \tilde{B}(\Delta \tau) 
    = & \   \frac{1}{2}\sum_j \frac{c_j^2}{\epsilon \omega_j} 
    \left(
    \coth\left(\frac{\beta\epsilon\omega_j}{2}\right)
    \cos(\omega_j \Delta \tau ) - \ii \sin(\omega_j \Delta \tau )
    \right)
     \end{split}
    \label{bath_function_B}
\end{equation}
where $\Delta \tau = \vert \tau_1 \vert - \vert \tau_2 \vert$.
Here we assume that the coupling intensity $c_j$ has order $\mathcal{O}(\epsilon)$, so that $\tilde{B}(\cdot) \sim \mathcal{O}(1)$.
Since $\mathcal{L}_b(\Sb)$ vanishes when the number of time points in $\Sb$ is odd, the summation in Dyson series \cref{eq_dyson} now only runs over even $m$. The set $\mathscr{Q}(\boldsymbol{s})$ denotes all possible pairings. For example, 
\begin{equation}
    \mathscr{Q}(\boldsymbol{s}) = \begin{cases}
        \{\{(s_1,s_2)\}\}, &\text{if~}\boldsymbol{s} = (s_1,s_2) \\
        \{\{(s_1,s_2),(s_3,s_4)\},\{(s_1,s_3),(s_2,s_4)\},\{(s_1,s_4),(s_2,s_3)\}\}, &\text{if~}\boldsymbol{s} = (s_1,s_2,s_3,s_4) \\
        \dots
    \end{cases}.
\end{equation}
and the corresponding \cref{Lb} is formulated as 
\begin{equation}
\begin{split}
    &\mathcal{L}_b(s_1,s_2) = B(s_1,s_2),\\
    & \mathcal{L}_b(s_1,s_2,s_3,s_4) = B(s_1,s_2)B(s_3,s_4)+ B(s_1,s_3)B(s_2,s_4) + B(s_1,s_4)B(s_2,s_3),\\
    & \cdots.
\end{split}
\end{equation}
Under such definition, we employ the idea of Feynman diagrams \cite{dauria2011special} to write the complicated Dyson series \cref{eq_dyson} into a simple diagrammatic equation
\begin{equation}
\label{eq_diag_eq}
    \begin{tikzpicture}[anchor=base, baseline,scale=0.6]
        \fill [black] (-1.4,-0.13) rectangle (1.4,0.13);
        \draw [thick] (-1.4,-0.18)--(-1.4,0.18); 
        \draw [thick] (1.4,-0.18)--(1.4,0.18);
        \node [below] at (-1.4,-0.18) {$-t$}; 
        \node [below] at (1.4,-0.18) {$t$}; 
    \end{tikzpicture} 
    \hspace{-5pt}=\hspace{-5pt}
    \begin{tikzpicture}[anchor=base, baseline,scale=0.6] 
        \draw [thick] (-1.4,0) -- (1.4,0);
        \draw [thick] (-1.4,-0.1)--(-1.4,0.1); 
        \draw [thick] (1.4,-0.1)--(1.4,0.1);
        \node [below] at (-1.4,-0.1) {$-t$}; 
        \node [below] at (1.4,-0.1) {$t$}; 
    \end{tikzpicture}
    \hspace{-2pt}+\hspace{-5pt}
    \begin{tikzpicture}[anchor=base, baseline,scale=0.6] 
        \draw [thick] (-1.4,0) -- (1.4,0);
        \draw [thick] (-1.4,-0.1)--(-1.4,0.1); 
        \draw [thick] (1.4,-0.1)--(1.4,0.1);
        \node [below] at (-1.4,-0.1) {$-t$}; 
        \node [below] at (1.4,-0.1) {$t$}; 
        \draw[-] (-1.4/3,0) to [bend left=60] (1.4/3,0);
        \draw plot[only marks,mark =*, mark options={color=black, scale=0.5}]coordinates {(-1.4/3,0)(1.4/3,0)};
        \node [below] at (-1.4/3,0) {$s_1$};
        \node [below] at (1.4/3,0) {$s_2$};
    \end{tikzpicture}
    \hspace{-2pt}+\hspace{-5pt}
    \begin{tikzpicture}[anchor=base, baseline,scale=0.6] 
        \draw [thick] (-1.4,0) -- (1.4,0);
        \draw [thick] (-1.4,-0.1)--(-1.4,0.1); 
        \draw [thick] (1.4,-0.1)--(1.4,0.1);
        \node [below] at (-1.6,-0.1) {$-t$}; 
        \node [below] at (1.6,-0.1) {$t$}; 
        \draw[-] (-0.84,0) to[bend left=75] (-0.28,0);
        \draw[-] (0.28,0) to[bend left=75] (0.84,0);
        \draw plot[only marks,mark =*, mark options={color=black, scale=0.5}]coordinates {(-0.84,0) (-0.28,0) (0.28,0) (0.84,0)};
        \node [below] at (-0.84,0) {$s_1$};  
        \node [below] at (-0.28,0) {$s_2$};   
        \node [below] at (0.28,0) {$s_3$};
        \node [below] at (0.84,0) {$s_4$};
    \end{tikzpicture}
    \hspace{-2pt}+\hspace{-5pt}
    \begin{tikzpicture}[anchor=base, baseline,scale=0.6] 
        \draw [thick] (-1.4,0) -- (1.4,0);
        \draw [thick] (-1.4,-0.1)--(-1.4,0.1); 
        \draw [thick] (1.4,-0.1)--(1.4,0.1);
        \node [below] at (-1.6,-0.1) {$-t$}; 
        \node [below] at (1.6,-0.1) {$t$}; 
        \draw[-] (-0.84,0) to[bend left=75] (0.28,0);
        \draw[-] (-0.28,0) to[bend left=75] (0.84,0);
        \draw plot[only marks,mark =*, mark options={color=black, scale=0.5}]coordinates {(-0.84,0) (-0.28,0) (0.28,0) (0.84,0)};
        \node [below] at (-0.84,0) {$s_1$};  
        \node [below] at (-0.28,0) {$s_2$};   
        \node [below] at (0.28,0) {$s_3$};
        \node [below] at (0.84,0) {$s_4$};
    \end{tikzpicture}
    \hspace{-2pt}+\hspace{-5pt}
    \begin{tikzpicture}[anchor=base, baseline,scale=0.6] 
        \draw [thick] (-1.4,0) -- (1.4,0);
        \draw [thick] (-1.4,-0.1)--(-1.4,0.1); 
        \draw [thick] (1.4,-0.1)--(1.4,0.1);
        \node [below] at (-1.6,-0.1) {$-t$}; 
        \node [below] at (1.6,-0.1) {$t$}; 
        \draw[-] (-0.84,0) to[bend left=75] (0.84,0);
        \draw[-] (-0.28,0) to[bend left=75] (0.28,0);
        \draw plot[only marks,mark =*, mark options={color=black, scale=0.5}]coordinates {(-0.84,0) (-0.28,0) (0.28,0) (0.84,0)};
        \node [below] at (-0.84,0) {$s_1$};  
        \node [below] at (-0.28,0) {$s_2$};   
        \node [below] at (0.28,0) {$s_3$};
        \node [below] at (0.84,0) {$s_4$};
    \end{tikzpicture}
    \hspace{-2pt}+\dots.
\end{equation}
On the left side of the equation above, the bold line with ends labeled with $\pm t$ denotes the reduced density operator $\rho_s(t,x_1,x_2)$ in \cref{eq_dyson}. On the right side, we use each diagram to represent one integral in \cref{eq_dyson}. For example, the third diagram on the right side denotes the integral  {
\begin{multline*}
    \begin{tikzpicture}[anchor=base, baseline,scale=0.6] 
        \draw [thick] (-1.4,0) -- (1.4,0);
        \draw [thick] (-1.4,-0.1)--(-1.4,0.1); 
        \draw [thick] (1.4,-0.1)--(1.4,0.1);
        \node [below] at (-1.6,-0.1) {$-t$}; 
        \node [below] at (1.6,-0.1) {$t$}; 
        \draw[-] (-0.84,0) to[bend left=75] (-0.28,0);
        \draw[-] (0.28,0) to[bend left=75] (0.84,0);
        \draw plot[only marks,mark =*, mark options={color=black, scale=0.5}]coordinates {(-0.84,0) (-0.28,0) (0.28,0) (0.84,0)};
        \node [below] at (-0.84,0) {$s_1$};  
        \node [below] at (-0.28,0) {$s_2$};   
        \node [below] at (0.28,0) {$s_3$};
        \node [below] at (0.84,0) {$s_4$};
    \end{tikzpicture} =   \int^{t}_{-t}  \int^{s_4}_{-t} \int^{s_3}_{-t} \int^{s_{2}}_{-t}  \int^{s_{1}}_{-t}  
      G^{(s)}(s_{4},t) \left(-\ii \sgn (s_4) W_s \right)
    G^{(s)}(s_{3},s_4)  \left(-\ii \sgn (s_3) W_s \right) \times \\
     \times 
    G^{(s)}(s_2,s_3) 
     \left(-\ii \sgn (s_2) W_s \right) G^{(s)}(s_1,s_2) 
     \left(-\ii \sgn (s_1) W_s \right) G^{(s)}(-t,s_1) \times \\ 
     \times  B(s_1,s_2)B(s_3,s_4)
    \, \dd s_1  \, \dd s_2  \dd s_{3} \, \dd s_4   .
\end{multline*} }
Each operator {$G^{(s)}(s_j,s_{j+1})$} is denoted by a line segment from $s_j$ to $s_{j+1}$; the coupling term $-\ii \sgn (s_j) W_s$ is represented by the dot labeled by $s_j$; each two point correlation function $B(s_\ii,s_\ff)$ is denoted by the arc connecting dots $s_\ii$ and $s_\ff$.

In practice, to evaluate the Dyson series numerically, one can truncate the series by a sufficiently large integer $\bar{M}$ and the{n} compute the integrals/thin diagrams using Monte Carlo integration. Such implementation results in the bare dQMC method \cite{prokof1998polaron,werner2009diagrammatic}. While bare dQMC avoids {the} curse of dimensionality with the advantage of Monet Carlo, it is well known that this method often suffers from the sign problem \cite{loh1990sign,cai2023numerical} caused by a large variance from sampling. To mitigate the sign problem, bold line methods \cite{prokof2008bold,gull2010bold,kulagin2013boldDiagrammaticMonteCarloMethod,kulagin2013boldMonteCarloTechnique,cohen2015taming,li2019bold,cai2023bold} have been developed in recent years. These approaches regroup some of the thin diagrams into bold diagrams and can achieve more stable results with a smaller truncation number of the series $\bar{M}$. Among these methods, the inchworm method has shown successful performance applied to models including quantum impurity problems and exact non-adiabatic dynamics \cite{chen2017inchwormITheory,chen2017inchwormIIBenchmarks,cai2020inchworm,cai2022fast}. In the next section, we will develop our proposed algorithm which integrates the central idea of inchworm method into FGA to simulate open quantum systems.

\subsection{Inchworm-FGA algorithm}
\label{sec_inchworm}
\subsubsection{FGA for Dyson series}
We now apply FGA to approximate the Dyson series \cref{eq_dyson}. We begin with the system influence functional $\mathcal{U}_{x_1x_2}$ defined in \cref{eq_system_part_dyson}. Under the setting of Caldeira-Leggett model, we can decompose the system influence functional as {$\bra{x_1}\mathcal{U} \ket{x_2} = \mathcal{U}_{x_1}^ -\mathcal{U}_{x_2}^+$} where {
\begin{equation}
\label{eq_U_half}
    \begin{split}
        &\mathcal{U}_{x_1}^-(\boldsymbol{s},t)=\langle x_1 | \e^{-\ii H_s (t - s_m)/\epsilon}   \hat{x}  \e^{-\ii H_s (s_m - s_{m-1})/\epsilon}   \hat{x} \cdots \hat{x} \e^{-\ii H_s s_{i+1}/\epsilon} |\psi_{s}^{(0)}\rangle, \\
        &\mathcal{U}_{x_2}^+(\boldsymbol{s},t)=\langle \psi_{s}^{(0)} |
        \e^{-\ii H_s s_{i}/\epsilon} \hat{x} \cdots  \hat{x}  \e^{\ii H_s (s_2 - s_1)/\epsilon} \hat{x}  \e^{\ii H_s (s_1 + t)/\epsilon} | x_2 \rangle
    \end{split}
\end{equation} }
according to the definition of {$G^{(s)}$} in \cref{eq_bare_propagator}. Here we have assumed that $s_i < 0 < s_{i+1}$ for convenience. 
% One can easily notice that $\mathcal{U}_{x_1}^+(\boldsymbol{s},t)
%     = \left(\mathcal{U}_{x_1}^{-}(\boldsymbol{s},t)\right)^\dagger$ and $\mathcal{U}_{x_2}^+(\boldsymbol{s},t)
%     = \left(\mathcal{U}_{x_2}^{-}(\boldsymbol{s},t)\right)^\dagger$. \ki{I still think there is a negative sign before $\boldsymbol{s}$ and also reverse the sequence:
%     \begin{align}
%         \mathcal{U}_{x_1}^{-}([-1,2],3) = \e^{-\ii H_s (3-2) / \epsilon} \hat{x} \e^{-\ii H_s (2)/\epsilon} \psi_s^{(0)}(x_1) \\
%         \mathcal{U}_{x_1}^{+}([-1,2],3) = \psi_{s}^{(0)^*} (x_1) \e^{\ii H_s [-(-1)]/\epsilon }\hat{x}\e^{\ii H_s (-1+3) /\epsilon} \\
%         \mathcal{U}_{x_1}^{+}([-2,1],3) = \psi_s^{(0)^*}(x_1)
%         \e^{\ii H_s [-(-2)]/\epsilon }\hat{x}\e^{\ii H_s (-2+3) /\epsilon}
%     \end{align}
%     } \sy{Then I think we don't need to write down any relation. This is just for us to say the simplification of $\mathcal{U}^+_{x_2}$ is similar, which is quite trivial. It seems that this relation is not used anywhere else.}\ki{Yes. We may delete this part directly.} 
To approximate \cref{eq_U_half}, we start from approximating the initial wave function {$\psi_{s}^{(0)}(x) = \langle x | \psi_{s}^{(0)}\rangle$} using FGA, and then apply the semi-group $\exp(\pm \ii H_s (s_{j+1} - s_{j})/\epsilon)$ as well as position operator $\hat{x}$ repeatedly while ensuring the evolved wave function remain to take the form of FGA ansatz during the whole process.    
Suppose we have approximated {$\langle x_1| \hat{x}  \e^{-\ii H_s (s_k - s_{k-1})/\epsilon}   \hat{x} \cdots \hat{x} \e^{-\ii H_s s_{i+1}/\epsilon} |\psi_{s}^{(0)}\rangle$} by $\psi_{\mathrm{FGA}}(s_k,x_1)$, the following suggests the approximation of $\hat{x}\psi_{\mathrm{FGA}}(s_k,x_1)$ in the next iteration:
% To proceed, we treat the pure state $\psi_{s}^{(0)}$ in \cref{eq_U_half} as the initial condition of the closed Schr\"odinger equation governed by system Hamiltonian $H_s$ and approximate the dynamics of the resulting wave function by FGA.
% We need the following two lemmas for the approximation of \cref{eq_U_half}.
\begin{proposition}
\label{prop_position_operator_beam_center}
Given an FGA ansatz
\begin{equation*}
    \psi_{\FGA}(t,x)
    = \frac{1}{({2\pi\epsilon})^{3/2}}
    \int_{\mathbb{R}^2}
    \int_{-\infty}^\infty
    a(t,p,q) \e^{\ii \phi(t,x,y,p,q) / \epsilon}
    \psi_{0}(y)
    \dd y \dd p \dd q,
\end{equation*}
suppose the function $c(t,y,p,q) = a(t,p,q)\psi_{0}(y)$ is in Schwartz class. For any $t\in[0,T]$, we have 
\begin{equation}
 \left\|  \hat{x}\psi_{\FGA}(t,x) - \frac{1}{({2\pi\epsilon})^{3/2}}
    \int_{\mathbb{R}^2}
    \int_{-\infty}^\infty
    Q(t,p,q)a(t,p,q) \e^{\ii \phi(t,x,y,p,q) / \epsilon}
    \psi_0(y)
    \dd y \dd p \dd q \right\|_{L^{2}} \leqslant C_2 \epsilon
\end{equation}
for some constant $C_2$.
\end{proposition}
This result can be derived from \cite[Lemma 3.2]{lu2011frozen}. The assumption that $c(t,y,p,q)$ is in Schwartz class can be achieved by using Gaussian-like initial values $\psi^{(0)}_s$ whose derivatives decay sufficiently rapidly. To understand the intuition of Proposition \ref{prop_position_operator_beam_center}, we may consider $\psi_{\mathrm{FGA}}$ as an integral of ``beams"
\begin{equation}\label{eq_beam}
    \psi_{\mathrm{beam}}(t,x,p,q) =  \frac{1}{({2\pi\epsilon})^{3/2}}
    \int_{-\infty}^{\infty} a(t,p,q)
    \e^{\ii \phi(t,x,y,p,q)/\epsilon} \psi^{(0)}_s(y) \dd y
\end{equation}
over the $(p,q)$-domain. When approximating $\hat{x}\psi_{\mathrm{FGA}}(s_k,x_1)$, each beam in $\psi_{\mathrm{FGA}}(s_k,x_1)$ is centered at $Q(s_k,p,q)$ along the $x$-dimension, and decays very fast as $x$ leaves $Q(s_k,p,q)$ due to the narrow width $\epsilon$. As a result, one can replace each position operator $\hat{x}$ by the beam center $Q(s_k,p,q)$ with a $\mathcal{O}(\epsilon)$ error. 

Afterwards, we will fix $\hat{x}$ as the time-independent function $Q(s_k,p,q)$ in the next evolution under the semi-group $\e^{-\ii H_s (s_{k+1} - s_{k})/\epsilon}$.  The following proposition suggests that the wave function can keep the FGA form with each $\psi_{\mathrm{beam}}$ with fixed beam index $(p,q)$ rescaled by $Q(s_k,p,q)$ in the next episode $[s_{k},s_{k+1}]$:
\begin{proposition}
\label{prop_scale_coefficient}
For two FGA ans\"atze 
% \ki{how can I write the plural form of ``ansatz''?}
\begin{equation}
    {\psi}_{\FGA}^{(a)}(t,x) = \frac{1}{(2\pi\epsilon)^{3/2}}
    \int_{\mathbb{R}^2}\int_{-\infty}^\infty
    a(t,p,q) {\e^{\ii \phi_a(t,x,y,p,q)/\epsilon}} \psi_0(y) \dd y \dd p \dd q,
\end{equation}
and
\begin{equation}
    {\psi}_{\FGA}^{(b)}(t,x) = \frac{1}{(2\pi\epsilon)^{3/2}}
    \int_{\mathbb{R}^2}\int_{-\infty}^\infty
    b(t,p,q) {\e^{\ii \phi_b(t,x,y,p,q)/\epsilon}} \psi_0(y) \dd y \dd p \dd q
\end{equation}
following the dynamics \cref{eq_dynamics} with the same potential function $V(x)$,
if at some time $t_0$, we have
\begin{itemize}
    \item $b(t_0,p,q) = c(p,q) a(t_0,p,q)$ for some $c(p,q)$,
    \item $S_a(t_0,p,q) = S_b(t_0,p,q)$, $P_a(t_0,p,q) = P_b(t_0,p,q)$ and $Q_a(t_0,p,q) = Q_b(t_0,p,q)$,
\end{itemize}
then for all $t$, we have
\begin{itemize}
    \item $b(t,p,q) = c(p,q) a(t,p,q)$,
    \item $S_a(t,p,q) = S_b(t,p,q)$, $P_a(t,p,q) = P_b(t,p,q)$ and $Q_a(t,p,q) = Q_b(t,p,q)$.
\end{itemize}
Note that we simply use subscriptions $a,b$ to denote the corresponding quantities in $\psi_{\FGA}^{(a)},\psi_{\FGA}^{(b)}$, respectively.
\begin{proof}
    Based on \cref{eq_dynamics}, the evolution of $S,P,Q$ is independent from the coefficients $a(t,p,q)$ and $b(t,p,q)$.
    By the uniqueness of the solution to an ODE system,
    we have the second statement $S_a(t,p,q) = S_b(t,p,q)$, $P_a(t,p,q) = P_b(t,p,q)$ and $Q_a(t,p,q) = Q_b(t,p,q)$.
    For fixed $p,q$, the dynamics of $a(t,p,q),b(t,p,q)$ satisfies
    \begin{equation}
    \begin{split}
        &\pdv{a(t,p,q)}{t} = \beta(t,p,q) a(t,p,q), \\
        &\pdv{b(t,p,q)}{t} = \beta(t,p,q) b(t,p,q).
    \end{split}
    \end{equation}
    In these two differential equations,
    \begin{equation}
        \beta(t,p,q) = \frac{1}{2} \frac{\partial_q P - \partial_p Q V''(Q) - \ii\left(\partial_p P +\partial_q Q V''(Q)\right)}{\partial_q Q + \partial_p P + \ii (\partial_q P - \partial_p Q)}
    \end{equation}
    only depends on $P,Q$.
    Hence, the dynamics of $a(t,p,q)$ and $b(t,p,q)$ share the same quantity $\beta(t,p,q)$.
    Therefore, we have
    \begin{equation}
    \begin{split}
        &a(t,p,q) = a(t_0,p,q) \exp\left({\int_{t_0}^t\beta(\tau,p,q)\dd \tau}\right), \\
        &b(t,p,q) = b(t_0,p,q) \exp\left({\int_{t_0}^t\beta(\tau,p,q)\dd \tau}\right),
    \end{split}
    \end{equation}
    and the first statement $b(t,p,q) = c(p,q) a(t,p,q)$ holds naturally.
\end{proof}
\end{proposition}
With these two propositions, we are ready to give an approximation of the system influence functional {$\bra{x_1}\mathcal{U} \ket{x_2}$}.
    Below we present the approximation of $\mathcal{U}_{x_1}^-(\boldsymbol{s},t)$ only, the approximation for $\mathcal{U}_{x_2}^+(\boldsymbol{s},t)$ is similar. The FGA for $\mathcal{U}_{x_1}^-(\boldsymbol{s},t)$ is done from the right to left. 
To begin with, we have 
\begin{equation}\label{semigroup_FGA}
\begin{split}
    \e^{-\ii H_s s_{i+1}/\epsilon} \psi_{s}^{(0)}(x_1)
    \approx\  &\psi_{\FGA}(s_{i+1},x_1) \\
    =\ &\frac{1}{(2\pi\epsilon)^{3/2}}
    \int_{\mathbb{R}^2}\int_{-\infty}^\infty
    a(s_{i+1},p,q) {\e^{\ii \phi(s_{i+1},x_1,y,p,q)/\epsilon}} \psi^{(0)}_s(y) \dd y \dd p \dd q
\end{split}
\end{equation}
based on the FGA.
Next, we use \cref{prop_position_operator_beam_center} to replace the position operator $\hat{x}$ by beam centers $Q(s_{i+1},p,q)$,
 and thus we have
\begin{equation}
\begin{split}
    &\hat{x} \e^{-\ii H_s s_{i+1}/\epsilon} \psi_{s}^{(0)}(x_1)  \\
    \approx &
    \frac{1}{(2\pi\epsilon)^{3/2}}
    \int_{\mathbb{R}^2}\int_{-\infty}^\infty
    Q(s_{i+1},p,q)a(s_{i+1},p,q) {\e^{\ii \phi(s_{i+1},x_1,y,p,q)/\epsilon}} \psi^{(0)}_s(y) \dd y \dd p \dd q
\end{split}
\end{equation}
    % \begin{align*}
    %      \e^{-\ii H_s s_{i+1}/\epsilon} \psi_{s}^{(0)}(x_1)  \approx \ & \psi_{\mathrm{FGA}}[a(s_{i+1}),\phi(s_{i+1}),\psi_{s}^{(0)}](x_1) \\
    %      := \ & \frac{1}{({2\pi\epsilon})^{3/2}}
    % \int_{\mathbb{R}^2}
    % \int_{-\infty}^\infty
    % a(s_{i+1},p,q) \e^{\ii \phi(s_{i+1},x,y,p,q) / \epsilon}
    % \psi_{s}^{(0)}(y)
    % \dd y \dd p \dd q
    % \end{align*}
% where $\phi$ is defined in \cref{phi} which depends on the variables $P,Q,S$. Together with $a$, these variables satisfy the dynamics \cref{eq_dynamics}. Next, we use the fact that $\phi$ is centered at $Q$ and decays very fast when $x$ leaves the center due to the small width $\epsilon$, and thus we have 
% \begin{displaymath}
%     \hat{x} \e^{-\ii H_s s_{i+1}/\epsilon} \psi_{s}^{(0)}(x_1)  \approx \psi_{\mathrm{FGA}}[Q(s_{i+1})a(s_{i+1}),\phi(s_{i+1}),\psi_{s}^{(0)}](x_1).
% \end{displaymath}
% By considering another set of variables $\tilde{P},\tilde{Q},\tilde{S},\tilde{a}$ satisfying the dynamics \cref{eq_dynamics} with the initial conditions
% \begin{gather*}
%     \tilde{P}(0,p,q) = P(s_{i+1},p,q), 
%     \quad
%     \tilde{Q}(0,p,q) = Q(s_{i+1},p,q),
%     \\
%     \tilde{S}(0,p,q) = S(s_{i+1},p,q),
%     \quad
%     \tilde{a}(0,p,q) = Q(s_{i+1},p,q) a(s_{i+1},p,q),
% \end{gather*}
By \cref{prop_scale_coefficient},
the approximation continues as 
\begin{equation}
\begin{split}
   &\e^{-\ii H_s (s_{i+2} - s_{i+1})/\epsilon}  \hat{x} \e^{-\ii H_s s_{i+1}/\epsilon} \psi_{s}^{(0)}(x_1)   \\
   \approx\ &  \frac{1}{(2\pi\epsilon)^{3/2}}
    \int_{\mathbb{R}^2}\int_{-\infty}^\infty
    Q(s_{i+1},p,q)
    a(s_{i+2},p,q) {\e^{\ii \phi(s_{i+2},x_1,y,p,q)/\epsilon}} \psi^{(0)}_s(y) \dd y \dd p \dd q.
\end{split}
\end{equation}
Repeating such a process, we obtain the approximation of $\mathcal{U}_{x_1}^-(\boldsymbol{s},t)$ as 
\begin{equation}\label{U_minus_FGA}
    \mathcal{U}_{x_1}^-(\boldsymbol{s},t)
    \approx \frac{1}{(2\pi\epsilon)^{3/2}}
    \int_{\mathbb{R}^2}
    \int_{-\infty}^{\infty}
    Q(s_m)\cdots Q(s_{i+1}) a(t) \e^{\ii \phi(t,x_1,y,p,q)/\epsilon} \psi^{(0)}_s(y) \dd y \dd p \dd q
\end{equation}
where we omit the parameters $p,q$ for simplicity.
Based on \cref{prop_FGA,prop_position_operator_beam_center,prop_scale_coefficient}, we expect that the approximation of $\mathcal{U}_{x_1}^-(\bs,t)$ has an error at $\mathcal{O}(\epsilon)$. {In Section \ref{sec:error analysis}, we conduct a more rigorous numerical analysis of this approximation error.}
% We can have a similar approximation for $\mathcal{U}_{x_2}^+(\boldsymbol{s},t)$.
% \begin{displaymath}
%     \mathcal{U}_{x_1}^-(\boldsymbol{s},t)
%     \approx   \psi_{\mathrm{FGA}}[Q(s_{m})\cdots Q(s_{i+2})Q(s_{i+1})a(t),\phi(t),\psi_{s}^{(0)}](x_1).
% \end{displaymath}

In practice, the integral with respect to variables $p$ and $q$ is computed by numerical quadrature. 
To be specific, we choose some quadrature points $\{(p_k,q_k)\}_{k=1}^K\subset \mathbb{R}^2$ and compute their corresponding integral weights $w_k$. 
A simple strategy is to apply uniform mesh with grid points 
\begin{equation}
\label{eq_grids}
\begin{split}
    &p_k = p_{\min} + k_1\Delta p \text{ for }k_1=0,\cdots,N_p \\
    &q_k = q_{\min} + k_2 \Delta q,
    \text{ for }k_2=0,\cdots,N_q
\end{split}
\end{equation}
where $k=k_1N_q +k_2 + 1$
for some proper choices of $p_{\min},q_{\min}$ and $N_p, N_q$. 
With this form, the quadrature weights can be simply set as $w_k = \Delta p \Delta q$ for all $k$.
Physically, the parameter $p$ and $q$ denotes the momentum and position of the particle,
so the range of $p$ and $q$ in the numerical discretization relies on the momentum and position of the initial state.
The choice of $\Delta p$ and $\Delta q$ should be consistent with the parameter $\epsilon$ so that oscillation in the states can be better captured. 
With the discretization of $p$ and $q$, the initial wave function is then approximated by
\begin{equation}
\label{eq_initial_decomposition}
    \psi_s^{(0)}(x) \approx  \sum_{k=1}^K w_k \psi_{k}(0,x)
\end{equation}
where 
\begin{equation}
    \psi_{k}(t,x)
    = \frac{1}{({2\pi\epsilon})^{3/2}}
    \int_{-\infty}^{\infty} a_{k}(t)
    \e^{\ii \phi_{k}(t,x,y)/\epsilon} \psi^{(0)}_s(y) \dd y
\end{equation}
 is a discrete version of \cref{eq_beam} with $(p,q)$-pair chosen to be grid point $(p_k,q_k)$, and each $\phi_k$ takes the form
\begin{equation}
    \phi_{k}(t,x,y)
    = S_{k}(t)
    + \frac{\ii}{2}(x-Q_{k}(t))^2 + P_{k}(x-Q_{k}(t))
    + \frac{\ii}{2}(y-q_k)^2 - p_k(y-q_k).
\end{equation}
The time evolution of parameters $S_k,Q_k,P_k,a_k$ for each $k$ can be parallelly obtained by solving 
\begin{equation}
\label{eq_dynamics_discrete}
\begin{split}
    &\pdv{P_k(t)}{t} = -\tilde{V}'(Q_k(t)), \\
    &\pdv{Q_k(t)}{t} = P_k(t), \\
    &\pdv{S_k(t)}{t} = \frac{1}{2}(P_k(t))^2 - \tilde{V}(Q_k(t)),\\
    &\pdv{a_k(t)}{t} = \frac{1}{2} a_k(t)
    \frac{\partial_q P_k - \partial_p Q_k \tilde{V}''(Q_k) - \ii\left(\partial_p P_k +\partial_q Q_k \tilde{V}''(Q_k)\right)}{\partial_q Q_k + \partial_p P_k + \ii (\partial_q P_k - \partial_p Q_k)}.\\
    & \pdv{t} \left(\partial_p P_k\right) =  -\partial_p Q_k{\tilde{V}}''(Q_k) , \quad 
    \pdv{t} \left(\partial_p Q_k\right) =  \partial_p P_k, \\
    & \pdv{t} \left(\partial_q P_k\right) =  -\partial_q Q_k {\tilde{V}}''(Q_k), \quad
    \pdv{t} \left(\partial_q Q_k\right) = \partial_q P_k.
    \end{split}
\end{equation}
where $\tilde{V}(x) = V(x) + \sum_{j=1}^L \frac{c_j^2}{2\omega_j^2} {x}^2$ with the initial conditions
\begin{equation}
    P_k(0) = p_{k}, 
    \quad
    Q_k(0) = q_{k},
    \quad
    S_k(0) = 0,
    \quad
    a_k(0) = \sqrt{2}.
\end{equation}

Under such discretization, $\mathcal{U}_{x_1}^-(\boldsymbol{s},t)$ is eventually approximate by 
\begin{equation}
    \mathcal{U}_{x_1}^-(\boldsymbol{s},t)
    \approx \sum_{k=1}^K w_k Q_{k}(s_m) \cdots Q_{k}(s_{i+2})
    Q_{k}(s_{i+1})
    \psi_{k}(t,x_1).
\end{equation}
The approximation of $\mathcal{U}_{x_2}^+(\boldsymbol{s},t)$ can be similarly obtained by
\begin{equation}
    \mathcal{U}_{x_2}^+
    = \sum_{k=1}^K w_k
    \psi_{k_2}^*(t,x_2) Q_k(-s_i)Q_k(-s_{i-1})\cdots Q_k(-s_1),
\end{equation} 
and one arrives at the overall approximation for the system influence functional 
\begin{equation}
\label{eq_U_approximation}
\begin{split}
    &{\bra{x_1} \mathcal{U}(-t,\boldsymbol{s},t)\ket{x_2}} \\
    % &\approx \sum_{k_1l_1}\sum_{k_2l_2}
    % Q_{k_1l_1}(s_m) \cdots Q_{k_1l_1}(s_{r+1})
    % Q_{k_2l_2}(s_r) \cdots Q_{k_2l_2}(s_1)
    % \e^{-\ii H_s t / \epsilon}
    % \psi_{k_1l_1}(0,x_1) \psi_{k_2l_2}^*(0,x_2)
    % \e^{\ii H_s t / \epsilon} \\
    \approx& \sum_{k_1=1}^K\sum_{k_2=1}^K
    w_{k_1}w_{k_2}Q_{k_1}(s_m) \cdots Q_{k_1}(s_{i+1})
    Q_{k_2}(-s_i) \cdots Q_{k_2}(-s_1)
    \psi_{k_1}(t,x_1) \psi_{k_2}^*(t,x_2)
\end{split}
\end{equation}
under the assumption that $s_i < 0 < s_{i+1}$. 
Here, we move the term $\psi_{k_1}(t,x_1)\psi_{k_2}^*(t,x_2)$ to the end of the expression because all $Q$'s are scalars.
In terms of the bath influence functional, we compute it exactly according to the original definition \cref{Lb} without approximation. 
We therefore compute the reduced density operator \cref{eq_dyson} by
\begin{equation}
\begin{split}
    &{\bra{x_1}\rho_s(t)\ket{x_2} } \\
    =& \sum_{m=0}^{\infty}
    \int_{-t\leqslant \boldsymbol{s} \leqslant t}
    \prod_{j=1}^m 
    \left(-\ii \sgn (s_j)\right)
    \sum_{k_1=1}^K\sum_{k_2=1}^K
    w_{k_1}w_{k_2}
    \prod_{j=1}^m \left(\tilde{Q}_{k_1k_2}(s_j)\right)
    \psi_{k_1}(t,x_1) \psi_{k_2}^*(t,x_2)
    \mathcal{L}_b(\boldsymbol{s})
    \dd \boldsymbol{s} \\
    =&\sum_{k_1=1}^K \sum_{k_2=1}^K
    w_{k_1}w_{k_2} 
    \left(
    \sum_{m=0}^\infty
    \int_{-t\leqslant \boldsymbol{s} \leqslant t}
    \prod_{j=1}^m \left(-\ii \sgn(s_j)\tilde{Q}_{k_1k_2}(s_j)\right)
    % Q_{k_1}(s_m) \cdots Q_{k_1}(s_{i+1})
    % Q_{k_2}(-s_i) \cdots Q_{k_2}(-s_1)
    \mathcal{L}_b(\boldsymbol{s}) \dd \boldsymbol{s} \right)
    \psi_{k_1}(t,x_1) \psi_{k_2}^*(t,x_2)
\end{split}
\end{equation}
with 
\begin{equation}\label{Q_symmetry}
    \tilde{Q}_{k_1 k_2}(s)
    = \begin{cases}
        Q_{k_1}(s), & s>0,\\
        Q_{k_2}(-s), & s<0.
    \end{cases}
\end{equation}
Here, we exchange the summations over $k_1$ and $k_2$ with the path integral,
and the term $\psi_{k_1}(t,x_1) \psi_{k_2}^*(t,x_2)$, which is independent of $\bs$, can be moved out of the path integral.
At this point, the Dyson series for the reduced density operator for Caldeira-Leggett model is simplified as  
\begin{equation}
\label{eq_rho_varrho}
    {\bra{x_1}\rho_s(t)\ket{x_2} }
    = 
    \sum_{k_1=1}^K\sum_{k_2=1}^K 
    % \varrho_{k_1l_1}^{k_2l_2}(t,x_1,x_2)
    w_{k_1}w_{k_2}\varrho_{k_1 k_2}(t)
    \psi_{k_1}(t,x_1) \psi_{k_2}^*(t,x_2)
\end{equation}
where each coefficient $\varrho_{k_1k_2}(t)$ takes the form of a Dyson series
\begin{equation}
\label{eq_varrho}
    \varrho_{k_1 k_2}(t)
    = \sum_{m=0}^\infty 
    \int_{-t\leqslant\boldsymbol{s}\leqslant t} 
    \prod_{j=1}^m 
    \left(-\ii \sgn (s_j)
        \tilde{Q}_{k_1 k_2}(s_j)
    \right)
    % Q_{k_1l_1}(s_m) \cdots Q_{k_1l_1}(s_{r+1})
    % Q_{k_2l_2}(s_r) \cdots Q_{k_2l_2}(s_1)
    \mathcal{L}_b(\boldsymbol{s})
    \dd \boldsymbol{s}
\end{equation}
with the initial condition $\varrho_{k_1 k_2}(0) = 1$ for all fixed $k_1,k_2$.
% The beam centers $Q_{k}(t)$ in \cref{eq_varrho} can be retrieved from solving the dynamics \cref{eq_dynamics_discrete}.
% Note that the values of $\tilde{Q}_{k_1k_2}(s)$ when $s=0$ is theoretically not important in the integrals since it forms a zero-measure set.
% In the numerical experiments, instead, we should split the integral region into subregions according to the values of $\boldsymbol{s}$
% and the values of $\tilde{Q}_{k_1k_2}(s)$ should be the limit value within the subregion.

Similar to \cref{eq_diag_eq}, one may also use a diagrammatic equation to represent \cref{eq_varrho}:
\begin{equation}
\label{eq_diag_eq_CL}
    \begin{tikzpicture}[anchor=base, baseline,scale=0.6]
        \fill [black] (-1.4,-0.13) rectangle (1.4,0.13);
        \draw [thick] (-1.4,-0.18)--(-1.4,0.18); 
        \draw [thick] (1.4,-0.18)--(1.4,0.18);
        \node [below] at (-1.4,-0.18) {$-t$}; 
        \node [below] at (1.4,-0.18) {$t$}; 
    \end{tikzpicture} 
    \hspace{-5pt}=\hspace{-5pt}
    \begin{tikzpicture}[anchor=base, baseline,scale=0.6] 
        \draw [thick,dashed] (-1.4,0) -- (1.4,0);
        \draw [thick] (-1.4,-0.1)--(-1.4,0.1); 
        \draw [thick] (1.4,-0.1)--(1.4,0.1);
        \node [below] at (-1.4,-0.1) {$-t$}; 
        \node [below] at (1.4,-0.1) {$t$}; 
    \end{tikzpicture}
    \hspace{-2pt}+\hspace{-5pt}
    \begin{tikzpicture}[anchor=base, baseline,scale=0.6] 
        \draw [thick,dashed] (-1.4,0) -- (1.4,0);
        \draw [thick] (-1.4,-0.1)--(-1.4,0.1); 
        \draw [thick] (1.4,-0.1)--(1.4,0.1);
        \node [below] at (-1.4,-0.1) {$-t$}; 
        \node [below] at (1.4,-0.1) {$t$}; 
        \draw[-] (-1.4/3,0) to [bend left=60] (1.4/3,0);
        \draw plot[only marks,mark =*, mark options={color=black, scale=0.5}]coordinates {(-1.4/3,0)(1.4/3,0)};
        \node [below] at (-1.4/3,0) {$s_1$};
        \node [below] at (1.4/3,0) {$s_2$};
    \end{tikzpicture}
    \hspace{-2pt}+\hspace{-5pt}
    \begin{tikzpicture}[anchor=base, baseline,scale=0.6] 
        \draw [thick,dashed] (-1.4,0) -- (1.4,0);
        \draw [thick] (-1.4,-0.1)--(-1.4,0.1); 
        \draw [thick] (1.4,-0.1)--(1.4,0.1);
        \node [below] at (-1.6,-0.1) {$-t$}; 
        \node [below] at (1.6,-0.1) {$t$}; 
        \draw[-] (-0.84,0) to[bend left=75] (-0.28,0);
        \draw[-] (0.28,0) to[bend left=75] (0.84,0);
        \draw plot[only marks,mark =*, mark options={color=black, scale=0.5}]coordinates {(-0.84,0) (-0.28,0) (0.28,0) (0.84,0)};
        \node [below] at (-0.84,0) {$s_1$};  
        \node [below] at (-0.28,0) {$s_2$};   
        \node [below] at (0.28,0) {$s_3$};
        \node [below] at (0.84,0) {$s_4$};
    \end{tikzpicture}
    \hspace{-2pt}+\hspace{-5pt}
    \begin{tikzpicture}[anchor=base, baseline,scale=0.6] 
        \draw [thick,dashed] (-1.4,0) -- (1.4,0);
        \draw [thick] (-1.4,-0.1)--(-1.4,0.1); 
        \draw [thick] (1.4,-0.1)--(1.4,0.1);
        \node [below] at (-1.6,-0.1) {$-t$}; 
        \node [below] at (1.6,-0.1) {$t$}; 
        \draw[-] (-0.84,0) to[bend left=75] (0.28,0);
        \draw[-] (-0.28,0) to[bend left=75] (0.84,0);
        \draw plot[only marks,mark =*, mark options={color=black, scale=0.5}]coordinates {(-0.84,0) (-0.28,0) (0.28,0) (0.84,0)};
        \node [below] at (-0.84,0) {$s_1$};  
        \node [below] at (-0.28,0) {$s_2$};   
        \node [below] at (0.28,0) {$s_3$};
        \node [below] at (0.84,0) {$s_4$};
    \end{tikzpicture}
    \hspace{-2pt}+\hspace{-5pt}
    \begin{tikzpicture}[anchor=base, baseline,scale=0.6] 
        \draw [thick,dashed] (-1.4,0) -- (1.4,0);
        \draw [thick] (-1.4,-0.1)--(-1.4,0.1); 
        \draw [thick] (1.4,-0.1)--(1.4,0.1);
        \node [below] at (-1.6,-0.1) {$-t$}; 
        \node [below] at (1.6,-0.1) {$t$}; 
        \draw[-] (-0.84,0) to[bend left=75] (0.84,0);
        \draw[-] (-0.28,0) to[bend left=75] (0.28,0);
        \draw plot[only marks,mark =*, mark options={color=black, scale=0.5}]coordinates {(-0.84,0) (-0.28,0) (0.28,0) (0.84,0)};
        \node [below] at (-0.84,0) {$s_1$};  
        \node [below] at (-0.28,0) {$s_2$};   
        \node [below] at (0.28,0) {$s_3$};
        \node [below] at (0.84,0) {$s_4$};
    \end{tikzpicture}
    \hspace{-2pt}+\cdots.
\end{equation}
The meaning of diagrams above is slightly different from \cref{eq_diag_eq}: the bold line above now denotes $\varrho_{k_1 k_2}(t)$ in Dyson series formulation as defined in \cref{eq_varrho}. On the right side of the equation, each diagram again represents one integral in \cref{eq_varrho}. The arcs still stand for the two-point correlations. The dots now denotes the beam centers $\tilde{Q}_{k_1k_2}(s_j)$. Since the bare propagators does not exist in $\varrho_{k_1 k_2}(t)$, we change the solid segments in \cref{eq_diag_eq} to dashed segments which are identity operators in system space.

\begin{remark}
We would like to emphasize here that \Cref{prop_scale_coefficient} is the key proposition that allows us to apply FGA to open quantum systems.
It essentially states that for a single Gaussian wave packet, the effect of the harmonic bath is merely a constant factor.
The interaction with the bath does not change the position $Q$, the momentum $P$ or the phase $S$.
As a result, although the evolution of a wave packet is a path integral, the function $\phi(t,x,y,p,q)$ inside all terms in the integral is identical, so that we can apply the path integral to the amplitude only.
\end{remark}
% \ki{To be continued here...}
% The discussion above provides a basic idea to deal with the system-bath interaction in Caldeira-Leggett model.
% The interaction operator $\hat{x}$ is approximated by multiplying the beam center at specific time for each beam in the frozen Gaussian approximation.
% However, the convergent rate of \cref{eq_varrho} is in general slow.
% A possible idea is to apply the inchworm algorithm introduced in \cref{sec_dyson_inchworm} to accelerate the convergence and also relief the numerical sign problem. 
\subsubsection{Inchworm method}
As we have discussed at the end of \cref{sec_preliminaries}, instead of evaluating $\varrho_{k_1k_2}(t)$ by a direct sum over the dashed diagrams in \cref{eq_diag_eq_CL} which will lead to a severe sign problem, we employ the key idea of inchworm method to compute $\varrho_{k_1k_2}(t)$ in this section. Specifically, we aim to reduce the number of diagrams by regrouping the dashed diagrams to bold lines. To start, we generalize the definition of $\varrho_{k_1 k_2}(t)$ to the following full propagator 
\begin{equation} \label{eq:G_Dyson}
    G_{k_1 k_2}(s_\ii,s_\ff)
    = \sum_{m=0}^\infty 
    \int_{s_\ii \leqslant\boldsymbol{s}\leqslant s_\ff} 
    \prod_{j=1}^m 
    \left(-\ii \sgn (s_j)
        \tilde{Q}_{k_1 k_2}(s_j)
    \right)
    % Q_{k_1}(s_m) \cdots Q_{k_1}(s_{r+1})
    % Q_{k_2}(s_r) \cdots Q_{k_2}(s_1)
    \mathcal{L}_b(\boldsymbol{s})
    \dd \boldsymbol{s}
\end{equation}
defined in non-symmetric time interval $[s_\ii,s_\ff] \subset [-t,t]$, which also takes the form of Dyson series and thus can again be represented as a sum of dashed diagrams 
\begin{equation}
\label{eq_diag_eq_CL_G}
    \begin{tikzpicture}[anchor=base, baseline,scale=0.6]
        \fill [black] (-1.4,-0.13) rectangle (1.4,0.13);
        \draw [thick] (-1.4,-0.18)--(-1.4,0.18); 
        \draw [thick] (1.4,-0.18)--(1.4,0.18);
        \node [below] at (-1.4,-0.18) {$s_\ii$}; 
        \node [below] at (1.4,-0.18) {$s_\ff$}; 
    \end{tikzpicture} 
    \hspace{-5pt}=\hspace{-5pt}
    \begin{tikzpicture}[anchor=base, baseline,scale=0.6] 
        \draw [thick,dashed] (-1.4,0) -- (1.4,0);
        \draw [thick] (-1.4,-0.1)--(-1.4,0.1); 
        \draw [thick] (1.4,-0.1)--(1.4,0.1);
        \node [below] at (-1.4,-0.1) {$s_\ii$}; 
        \node [below] at (1.4,-0.1) {$s_\ff$}; 
    \end{tikzpicture}
    \hspace{-2pt}+\hspace{-5pt}
    \begin{tikzpicture}[anchor=base, baseline,scale=0.6] 
        \draw [thick,dashed] (-1.4,0) -- (1.4,0);
        \draw [thick] (-1.4,-0.1)--(-1.4,0.1); 
        \draw [thick] (1.4,-0.1)--(1.4,0.1);
        \node [below] at (-1.4,-0.1) {$s_\ii$}; 
        \node [below] at (1.4,-0.1) {$s_\ff$}; 
        \draw[-] (-1.4/3,0) to [bend left=60] (1.4/3,0);
        \draw plot[only marks,mark =*, mark options={color=black, scale=0.5}]coordinates {(-1.4/3,0)(1.4/3,0)};
        \node [below] at (-1.4/3,0) {$s_1$};
        \node [below] at (1.4/3,0) {$s_2$};
    \end{tikzpicture}
    \hspace{-2pt}+\hspace{-5pt}
    \begin{tikzpicture}[anchor=base, baseline,scale=0.6] 
        \draw [thick,dashed] (-1.4,0) -- (1.4,0);
        \draw [thick] (-1.4,-0.1)--(-1.4,0.1); 
        \draw [thick] (1.4,-0.1)--(1.4,0.1);
        \node [below] at (-1.6,-0.1) {$s_\ii$}; 
        \node [below] at (1.6,-0.1) {$s_\ff$}; 
        \draw[-] (-0.84,0) to[bend left=75] (-0.28,0);
        \draw[-] (0.28,0) to[bend left=75] (0.84,0);
        \draw plot[only marks,mark =*, mark options={color=black, scale=0.5}]coordinates {(-0.84,0) (-0.28,0) (0.28,0) (0.84,0)};
        \node [below] at (-0.84,0) {$s_1$};  
        \node [below] at (-0.28,0) {$s_2$};   
        \node [below] at (0.28,0) {$s_3$};
        \node [below] at (0.84,0) {$s_4$};
    \end{tikzpicture}
    \hspace{-2pt}+\hspace{-5pt}
    \begin{tikzpicture}[anchor=base, baseline,scale=0.6] 
        \draw [thick,dashed] (-1.4,0) -- (1.4,0);
        \draw [thick] (-1.4,-0.1)--(-1.4,0.1); 
        \draw [thick] (1.4,-0.1)--(1.4,0.1);
        \node [below] at (-1.6,-0.1) {$s_\ii$}; 
        \node [below] at (1.6,-0.1) {$s_\ff$}; 
        \draw[-] (-0.84,0) to[bend left=75] (0.28,0);
        \draw[-] (-0.28,0) to[bend left=75] (0.84,0);
        \draw plot[only marks,mark =*, mark options={color=black, scale=0.5}]coordinates {(-0.84,0) (-0.28,0) (0.28,0) (0.84,0)};
        \node [below] at (-0.84,0) {$s_1$};  
        \node [below] at (-0.28,0) {$s_2$};   
        \node [below] at (0.28,0) {$s_3$};
        \node [below] at (0.84,0) {$s_4$};
    \end{tikzpicture}
    \hspace{-2pt}+\hspace{-5pt}
    \begin{tikzpicture}[anchor=base, baseline,scale=0.6] 
        \draw [thick,dashed] (-1.4,0) -- (1.4,0);
        \draw [thick] (-1.4,-0.1)--(-1.4,0.1); 
        \draw [thick] (1.4,-0.1)--(1.4,0.1);
        \node [below] at (-1.6,-0.1) {$s_\ii$}; 
        \node [below] at (1.6,-0.1) {$s_\ff$}; 
        \draw[-] (-0.84,0) to[bend left=75] (0.84,0);
        \draw[-] (-0.28,0) to[bend left=75] (0.28,0);
        \draw plot[only marks,mark =*, mark options={color=black, scale=0.5}]coordinates {(-0.84,0) (-0.28,0) (0.28,0) (0.84,0)};
        \node [below] at (-0.84,0) {$s_1$};  
        \node [below] at (-0.28,0) {$s_2$};   
        \node [below] at (0.28,0) {$s_3$};
        \node [below] at (0.84,0) {$s_4$};
    \end{tikzpicture}
    \hspace{-2pt}+\cdots.
\end{equation}
In the equation above, the full propagator $G_{k_1 k_2}(s_\ii,s_\ff)$ is represented as bold line segment with two ends $s_\ii$ and $s_\ff$, which can be considered as a partial sum of total Dyson series \cref{eq_varrho}. According to the definition, one can recover the desired $\varrho_{k_1 k_2}(t)$ by 
\begin{equation}\label{recover_relation}
    \varrho_{k_1 k_2}(t) = G_{k_1 k_2}(-t,t).
\end{equation}
The central idea of inchworm method is to compute longer bold segments $G_{k_1 k_2}(s_{\ii},s_{\ff})$ from shorter ones iteratively until arriving at the full bold line defined in \cref{eq_diag_eq_CL}. 
To realize such iterative process, we will need an evolution of the full propagator. One option is to derive the governing equation for the full propagator \cite{cai2020inchworm,cai2022fast,cai2023bold}. To do this, we begin with the derivative 
 \begin{equation}
\label{eq_differential_equation}
    \pdv{G_{k_1k_2}(s_\ii,s_\ff)}{s_\ff}
    = \sum_{m=1}^\infty \int_{s_\ii \leqslant \boldsymbol{s}\leqslant s_\ff} 
        \prod_{j=1}^{m+1} 
        \left(-\ii \sgn(s_j) \tilde{Q}_{k_1k_2}(s_j)
        \right)
        \mathcal{L}_b([\boldsymbol{s},s_\ff])
    \dd \boldsymbol{s}
\end{equation}
where $\boldsymbol{s}=(s_1,\cdots,s_m)$ is the integral variable, $[\boldsymbol{s},s_\ff]$ denotes the vector $(s_1,\cdots,s_m,s_\ff)$ and $s_{m+1}=s_\ff$ is used in the sake of simplicity of the notation. 
Note that \cref{eq_differential_equation} does not hold for $s_\ff=0$ because the integrand is discontinuous as a result of \cref{Q_symmetry}.
However, the quantities $G_{k_1k_2}(s_\ii,s_\ff)$ are continuous for all $s_\ii \leqslant s_\ff$.
Such formula can again be written diagrammatically as 
\begin{equation}
\label{eq_diag_eq_CL_G_div}
    \pdv{s_\ff}
    \begin{tikzpicture}[anchor=base, baseline,scale=0.6]
        \fill [black] (-1.4,-0.13) rectangle (1.4,0.13);
        \draw [thick] (-1.4,-0.18)--(-1.4,0.18); 
        \draw [thick] (1.4,-0.18)--(1.4,0.18);
        \node [below] at (-1.4,-0.18) {$s_\ii$}; 
        \node [below] at (1.4,-0.18) {$s_\ff$}; 
    \end{tikzpicture} 
    \hspace{-5pt}=\hspace{-5pt}
    \begin{tikzpicture}[anchor=base, baseline,scale=0.6] 
        \draw [thick,dashed] (-1.4,0) -- (1.4,0);
        \draw [thick] (-1.4,-0.1)--(-1.4,0.1); 
        \draw [thick] (1.4,-0.1)--(1.4,0.1);
        \node [below] at (-1.4,-0.1) {$s_\ii$}; 
        \node [below] at (1.4,-0.1) {$s_\ff$}; 
        \draw[-] (0,0) to [bend left=60] (1.4,0);
        \draw plot[only marks,mark =*, mark options={color=black, scale=0.5}]coordinates {(0,0)(1.4,0)};
        \node [below] at (0,0) {$s_1$};
    \end{tikzpicture}
    \hspace{-2pt}+\hspace{-5pt}
    \begin{tikzpicture}[anchor=base, baseline,scale=0.6] 
        \draw [thick,dashed] (-1.4,0) -- (1.4,0);
        \draw [thick] (-1.4,-0.1)--(-1.4,0.1); 
        \draw [thick] (1.4,-0.1)--(1.4,0.1);
        \node [below] at (-1.6,-0.1) {$s_\ii$}; 
        \node [below] at (1.6,-0.1) {$s_\ff$}; 
        \draw[-] (-0.7,0) to[bend left=75] (-0,0);
        \draw[-] (0.7,0) to[bend left=75] (1.4,0);
        \draw plot[only marks,mark =*, mark options={color=black, scale=0.5}]coordinates {(-0.7,0) (-0,0) (0.7,0)};
        \node [below] at (-0.7,0) {$s_1$};  
        \node [below] at (-0,0) {$s_2$};   
        \node [below] at (0.7,0) {$s_3$};
    \end{tikzpicture}
    \hspace{-2pt}+\hspace{-5pt}
    \begin{tikzpicture}[anchor=base, baseline,scale=0.6] 
        \draw [thick,dashed] (-1.4,0) -- (1.4,0);
        \draw [thick] (-1.4,-0.1)--(-1.4,0.1); 
        \draw [thick] (1.4,-0.1)--(1.4,0.1);
        \node [below] at (-1.6,-0.1) {$s_\ii$}; 
        \node [below] at (1.6,-0.1) {$s_\ff$}; 
        \draw[-] (-0.7,0) to[bend left=75] (0.7,0);
        \draw[-] (0,0) to[bend left=75] (1.4,0);
        \draw plot[only marks,mark =*, mark options={color=black, scale=0.5}]coordinates {(-0.7,0) (-0,0) (0.7,0)};
        \node [below] at (-0.7,0) {$s_1$};  
        \node [below] at (-0,0) {$s_2$};   
        \node [below] at (0.7,0) {$s_3$};
    \end{tikzpicture}
    \hspace{-2pt}+\hspace{-5pt}
    \begin{tikzpicture}[anchor=base, baseline,scale=0.6] 
        \draw [thick,dashed] (-1.4,0) -- (1.4,0);
        \draw [thick] (-1.4,-0.1)--(-1.4,0.1); 
        \draw [thick] (1.4,-0.1)--(1.4,0.1);
        \node [below] at (-1.6,-0.1) {$s_\ii$}; 
        \node [below] at (1.6,-0.1) {$s_\ff$}; 
        \draw[-] (-0.7,0) to[bend left=75] (1.4,0);
        \draw[-] (0,0) to[bend left=75] (0.7,0);
        \draw plot[only marks,mark =*, mark options={color=black, scale=0.5}]coordinates {(-0.7,0) (-0,0) (0.7,0)};
        \node [below] at (-0.7,0) {$s_1$};  
        \node [below] at (-0,0) {$s_2$};   
        \node [below] at (0.7,0) {$s_3$};
    \end{tikzpicture}
    \hspace{-2pt}+\cdots.
\end{equation}
From this diagrammatic equation, one can easily observe that the last time point is now fixed at $s_\ff$ in each diagram on the right side and thus the degrees of freedom decrease by 1 due to taking derivative. 

At this point, we want to replace every diagram on the right side of \cref{eq_diag_eq_CL_G_div} by ``bo{l}difying" all the dashed segments into the corresponding bold line segments. Such operation mathematically is to insert the full propagator $G_{k_1k_2}(s_{j-1},s_{j})$ between every $\tilde{Q}_{k_1k_2}(s_j)$ and $\tilde{Q}_{k_1k_2}(s_{j-1})$. Each of the bo{l}dified segments $G_{k_1k_2}(s_{j-1},s_{j})$ is again a summation of infinite dashed diagrams according to the expansion \cref{eq_diag_eq_CL_G}. During this replacement, we need to remove the over-counted dashed diagrams to maintain the equality. For example, we replace the first diagram in expansion \cref{eq_diag_eq_CL_G_div} by 
\begin{equation}
\scalebox{.75}{$\displaystyle%
      \begin{tikzpicture}[anchor=base, baseline]
        \draw [thick] (-1.4,0)--(1.4,0);
        \draw[-] (-0.8,0.1) to[bend left=60] (1.4,0.1);
        \fill [black] (-1.4,-0.1) rectangle (1.4,0.1);
        \node [below] at (-1.4,-0.18) {$s_{\ii}$};
        \node [below] at (1.4,-0.18) {$s_{\ff}$};
        \draw [thick] (-1.4,-0.12)--(-1.4,0.12);
        \draw [thick] (1.4,-0.12)--(1.4,0.12);
        \draw [thick,white] (-0.8,0.1) -- (-0.8,-0.1);
        \node [below] at (-0.8,-0.1) {$s_1$};
      \end{tikzpicture}  =  \begin{tikzpicture}[anchor=base, baseline]  
        \draw [thick,dashed] (-1.4,0) -- (1.4,0);
        \draw [thick] (-1.4,-0.1)--(-1.4,0.1); \draw [thick] (1.4,-0.1)--(1.4,0.1);
        \node [below] at (-1.4,-0.1) {$s_{\ii}$}; 
        \node [below] at (1.4,-0.1) {$s_{\ff}$}; 
        \draw[-] (0,0) to[bend left=75] (1.4,0);
        \draw plot[only marks,mark =*, mark options={ scale=0.5}]coordinates {(0,0)};
        \node [below] at (0,0) {$s_1$};
      \end{tikzpicture}  + \begin{tikzpicture}[anchor=base, baseline] 
        \draw [thick,dashed] (-1.4,0) -- (1.4,0);
        \draw [thick] (-1.4,-0.1)--(-1.4,0.1); \draw [thick] (1.4,-0.1)--(1.4,0.1);
        \node [below] at (-1.4,-0.1) {$s_{\ii}$}; 
        \node [below] at (1.4,-0.1) {$s_{\ff}$}; 
        \draw[-] (0,0) to[bend left=75] (1.4,0);
        \draw[-] (-1,0) to[bend left=75] (-0.4,0);
        \draw plot[only marks,mark =*, mark options={ scale=0.5}]coordinates {(-1,0)(-0.4,0)(0,0)};
        \node [below] at (-1,0) {$s_1$};
        \node [below] at (-0.4,0) {$s_2$};
        \node [below] at (0,0) {$s_3$};
      \end{tikzpicture} 
         +  \begin{tikzpicture}[anchor=base, baseline] 
        \draw [thick,dashed] (-1.4,0) -- (1.4,0);
        \draw [thick] (-1.4,-0.1)--(-1.4,0.1); \draw [thick] (1.4,-0.1)--(1.4,0.1);
        \node [below] at (-1.4,-0.1) {$s_{\ii}$}; 
        \node [below] at (1.4,-0.1) {$s_{\ff}$}; 
        \draw[-] (-0.9,0) to[bend left=75] (1.4,0);
        \draw[-] (-0.3,0) to[bend left=75] (0.8,0);
        \draw plot[only marks,mark =*, mark options={ scale=0.5}]coordinates {(-0.9,0) (-0.3,0) (0.8,0)};
        \node [below] at (-0.9,0) {$s_1$};  \node [below] at (-0.3,0) {$s_2$};   \node [below] at (0.8,0) {$s_3$};
      \end{tikzpicture} + \begin{tikzpicture}[anchor=base, baseline] 
        \draw [thick,dashed] (-1.4,0) -- (1.4,0);
        \draw [thick] (-1.4,-0.1)--(-1.4,0.1); \draw [thick] (1.4,-0.1)--(1.4,0.1);
        \node [below] at (-1.4,-0.1) {$s_{\ii}$}; 
        \node [below] at (1.4,-0.1) {$s_{\ff}$}; 
        \draw[-] (0,0) to[bend left=75] (1.4,0);
        \draw[-] (-1,0) to[bend left=75] (-0.4,0);
        \draw[-] (0.4,0) to[bend left=75] (1,0);
        \draw plot[only marks,mark =*, mark options={ scale=0.5}]coordinates
          {(-1,0)(-0.4,0)(0,0)(0.4,0)(1,0)};
        \node [below] at (-1,0) {$s_1$};
        \node [below] at (-0.4,0) {$s_2$};
        \node [below] at (0,0) {$s_3$};
        \node [below] at (0.4,0) {$s_4$};
        \node [below] at (1,0) {$s_5$};
      \end{tikzpicture}  + \cdots  $}
      \label{1st_diagram_update}
\end{equation}
Notice that in \cref{1st_diagram_update}, the first three diagrams are identical to the first, second and fourth diagrams in the expansion \cref{eq_diag_eq_CL_G_div} and thus these three diagrams (as well as infinitely many others that appear in both \cref{eq_diag_eq_CL_G_div} and \cref{1st_diagram_update} but are not written down explicitly) should be removed from the original expansion \cref{eq_diag_eq_CL_G_div} once this update is done. For those diagrams which are not removed during this update (e.g the third diagram in the expansion \cref{eq_diag_eq_CL_G_div}), we will update them by again ``bo{l}difying" their dashed segments. We repeat such process and after replacing all the dashed diagrams, we reach to a new reformulation of \cref{eq_diag_eq_CL_G_div}:
\begin{multline}\label{eq_diag_integro_differential_eq}
 \pdv{s_\ff} 
                \begin{tikzpicture}[anchor=base, baseline=0,scale=0.6]
                    \draw [thick] (-1.4,0)--(1.4,0);
                    \fill [black] (-1.4,-0.1) rectangle (1.4,0.1);
                    \draw [thick] (-1.4,-0.12)--(-1.4,0.12);
                    \draw [thick] (1.4,-0.12)--(1.4,0.12);
                    \node [below] at (-1.4,-0.12) {$s_\ii$};
                    \node [below] at (1.4,-0.12) {$s_\ff$};
                \end{tikzpicture} = 
          \begin{tikzpicture}[anchor=base, baseline,scale=0.6]
\draw [thick] (-1.4,0)--(1.4,0);
\draw[-] (-0.7,0.1) to[bend left=75] (1.4,0.1);
\fill [black] (-1.4,-0.1) rectangle (1.4,0.1);
\node [below] at (-1.4,-0.18) {$\Si$};
\node [below] at (1.4,-0.18) {$\Sf$};
\draw [thick] (-1.4,-0.12)--(-1.4,0.12);
\draw [thick] (1.4,-0.12)--(1.4,0.12);
\draw [thick,white] (-0.7,0.1) -- (-0.7,-0.1);
\node [below] at (-0.7,-0.1) {$s_1$};
\end{tikzpicture}
+ 
\begin{tikzpicture}[anchor=base, baseline,scale=0.6]
\draw [thick] (-1.4,0)--(1.4,0);
\draw[-] (-0.7,0.1) to[bend left=75] (0.7,0.1);
\draw[-] (0,0.1) to[bend left=75] (1.4,0.1);
\fill [black] (-1.4,-0.1) rectangle (1.4,0.1);
\node [below] at (-1.4,-0.18) {$\Si$};
\node [below] at (1.4,-0.18) {$\Sf$};
\draw [thick] (-1.4,-0.12)--(-1.4,0.12);
\draw [thick] (1.4,-0.12)--(1.4,0.12);
\draw [thick,white] (-0.7,0.1) -- (-0.7,-0.1);
\draw [thick,white] (0.7,0.1) -- (0.7,-0.1);
\draw [thick,white] (0,0.1) -- (0,-0.1);
\node [below] at (-0.7,-0.1) {$s_1$};\node [below] at (0,-0.1) {$s_2$};
\node [below] at (0.7,-0.1) {$s_3$};
\end{tikzpicture} + \\
  + 
\begin{tikzpicture}[anchor=base, baseline,scale=0.6]
\draw [thick] (-1.5,0)--(1.5,0);
\fill [black] (-1.5,-0.1) rectangle (1.5,0.1);
\draw[-] (-1,0.1) to[bend left=75] (0.5,0.1);
\draw[-] (-0.5,0.1) to[bend left=75] (1,0.1);
\draw[-] (0,0.1) to[bend left=75] (1.5,0.1);
\node [below] at (-1.5,-0.18) {$\Si$};
\node [below] at (1.5,-0.18) {$\Sf$};
\draw [thick] (-1.5,-0.12)--(-1.5,0.12);
\draw [thick] (1.5,-0.12)--(1.5,0.12);
\draw [thick,white] (-1,0.1) -- (-1,-0.1);
\draw [thick,white] (-0.5,0.1) -- (-0.5,-0.1);
\draw [thick,white] (-0,0.1) -- (-0,-0.1);
\draw [thick,white] (0.5,0.1) -- (0.5,-0.1);
\draw [thick,white] (1,0.1) -- (1,-0.1);
\node [below] at (-1,-0.1) {$s_1$};\node [below] at (-0.5,-0.1) {$s_2$};
\node [below] at (0,-0.1) {$s_3$};\node [below] at (0.5,-0.1) {$s_4$};
\node [below] at (1,-0.1) {$s_5$};
\end{tikzpicture} 
+ 
\begin{tikzpicture}[anchor=base, baseline,scale=0.6]
\draw [thick] (-1.5,0)--(1.5,0);
\fill [black] (-1.5,-0.1) rectangle (1.5,0.1);
\draw[-] (-1,0.1) to[bend left=75] (0.5,0.1);
\draw[-] (-0.5,0.1) to[bend left=75] (1.5,0.1);
\draw[-] (0,0.1) to[bend left=75] (1,0.1);
\node [below] at (-1.5,-0.18) {$\Si$};
\node [below] at (1.5,-0.18) {$\Sf$};
\draw [thick] (-1.5,-0.12)--(-1.5,0.12);
\draw [thick] (1.5,-0.12)--(1.5,0.12);
\draw [thick,white] (-1,0.1) -- (-1,-0.1);
\draw [thick,white] (-0.5,0.1) -- (-0.5,-0.1);
\draw [thick,white] (-0,0.1) -- (-0,-0.1);
\draw [thick,white] (0.5,0.1) -- (0.5,-0.1);
\draw [thick,white] (1,0.1) -- (1,-0.1);
\node [below] at (-1,-0.1) {$s_1$};\node [below] at (-0.5,-0.1) {$s_2$};
\node [below] at (0,-0.1) {$s_3$};\node [below] at (0.5,-0.1) {$s_4$};
\node [below] at (1,-0.1) {$s_5$};
\end{tikzpicture} 
+ 
\begin{tikzpicture}[anchor=base, baseline,scale=0.6]
\draw [thick] (-1.5,0)--(1.5,0);
\fill [black] (-1.5,-0.1) rectangle (1.5,0.1);
\draw[-] (-1,0.1) to[bend left=75] (0,0.1);
\draw[-] (-0.5,0.1) to[bend left=75] (1,0.1);
\draw[-] (0.5,0.1) to[bend left=75] (1.5,0.1);
\node [below] at (-1.5,-0.18) {$\Si$};
\node [below] at (1.5,-0.18) {$\Sf$};
\draw [thick] (-1.5,-0.12)--(-1.5,0.12);
\draw [thick] (1.5,-0.12)--(1.5,0.12);
\draw [thick,white] (-1,0.1) -- (-1,-0.1);
\draw [thick,white] (-0.5,0.1) -- (-0.5,-0.1);
\draw [thick,white] (-0,0.1) -- (-0,-0.1);
\draw [thick,white] (0.5,0.1) -- (0.5,-0.1);
\draw [thick,white] (1,0.1) -- (1,-0.1);
\node [below] at (-1,-0.1) {$s_1$};\node [below] at (-0.5,-0.1) {$s_2$};
\node [below] at (0,-0.1) {$s_3$};\node [below] at (0.5,-0.1) {$s_4$};
\node [below] at (1,-0.1) {$s_5$};
\end{tikzpicture} 
 + \begin{tikzpicture}[anchor=base, baseline,scale=0.6]
\draw [thick] (-1.5,0)--(1.5,0);
\fill [black] (-1.5,-0.1) rectangle (1.5,0.1);
\draw[-] (-1,0.1) to[bend left=75] (1,0.1);
\draw[-] (-0.5,0.1) to[bend left=75] (0.5,0.1);
\draw[-] (0,0.1) to[bend left=75] (1.5,0.1);
\node [below] at (-1.5,-0.18) {$\Si$};
\node [below] at (1.5,-0.18) {$\Sf$};
\draw [thick] (-1.5,-0.12)--(-1.5,0.12);
\draw [thick] (1.5,-0.12)--(1.5,0.12);
\draw [thick,white] (-1,0.1) -- (-1,-0.1);
\draw [thick,white] (-0.5,0.1) -- (-0.5,-0.1);
\draw [thick,white] (-0,0.1) -- (-0,-0.1);
\draw [thick,white] (0.5,0.1) -- (0.5,-0.1);
\draw [thick,white] (1,0.1) -- (1,-0.1);
\node [below] at (-1,-0.1) {$s_1$};\node [below] at (-0.5,-0.1) {$s_2$};
\node [below] at (0,-0.1) {$s_3$};\node [below] at (0.5,-0.1) {$s_4$};
\node [below] at (1,-0.1) {$s_5$};
\end{tikzpicture} 
+ \cdots
\end{multline}
We then immediately reach to the governing equation for the full propagator by writing \cref{eq_diag_integro_differential_eq} into the mathematical representation
\begin{equation}
\label{eq_integro_differential_equation}
    \pdv{G_{k_1k_2}(s_\ii,s_\ff)}{s_\ff}
    = \sum_{m=1}^\infty \int_{s_\ii \leqslant \boldsymbol{s}\leqslant s_\ff} 
        \prod_{j=1}^{m+1} 
        \left(-\ii \sgn(s_j) \tilde{Q}_{k_1k_2}(s_j)
            G_{k_1k_2}(s_{j-1},s_j)
        \right)
        \mathcal{L}_b^c([\boldsymbol{s},s_\ff])
    \dd \boldsymbol{s} 
    % \text{~for~} -t \leqslant s_{\ii} \leqslant s_{\ff} \leqslant t
\end{equation}
for $-t \leqslant s_{\ii} \leqslant s_{\ff} \leqslant t$
with the condition $G_{k_1k_2}(s',s') = 1$ for all $s' \in [-t,t]$. Above $\boldsymbol{s} = (s_1,\cdots,s_m)$ is the integral variable with $s_0=s_\ii$, $s_{m+1}=s_\ff$ defined for convenient purposes. The bath influence functional $\mathcal{L}_b^c(\boldsymbol{s})$ is now defined by
\begin{equation}
    \mathcal{L}_b^c(\boldsymbol{s})
    = \begin{dcases}
        0, & \text{if $m$ is odd,} \\
        \sum_{\mathfrak{q} \in \mathscr{Q}^c(\boldsymbol{s})} 
        \prod_{(\tau_1,\tau_2) \in \mathfrak{q}} B(\tau_1, \tau_2), & \text{if $m$ is even.}
    \end{dcases}
\end{equation}
where $\mathscr{Q}^c(\boldsymbol{s})$ is a subset of $\mathscr{Q}(\boldsymbol{s})$ that contains only ``linked" diagrams, meaning that any two time points can be connected with each other using the arcs as ``bridges'' in such diagrams. Rigorous mathematical proof can be found in \cite[Section 3.3]{cai2020inchworm}, showing that the expansion \cref{eq_diag_integro_differential_eq} includes all dashed diagrams without over-counting. Compared with direct sum of dashed diagrams \cref{eq_diag_eq_CL}, working with the resummation \cref{eq_diag_integro_differential_eq} is more efficient as it has faster convergence with respect to $m$ since each diagram in the new expansion \cref{eq_diag_integro_differential_eq} includes infinite dashed diagrams. Furthermore, the reduction in the number of diagrams in \cref{eq_diag_integro_differential_eq} (e.g, number of diagrams with two arcs decrease from 3 in \cref{eq_diag_eq_CL} to 1 in \cref{eq_diag_integro_differential_eq}) also makes the evaluation of bath influence functional $\mathcal{L}_b^c$ cheaper than $\mathcal{L}_b$.

To solve the equation \cref{eq_diag_integro_differential_eq} numerically, one 
can use Runge-Kutta methods to solve it as an ODE in the $s_{\ff}$ direction. In each time step, we truncate the series on the right side by a chosen integer $\bar{M}$ and evaluate the integrals using numerical quadratures \cite{wang2023real,sun2024simulation} or Monte Carlo methods \cite{cai2020inchworm,cai2023bold}. In this work, we use the second-order Heun's method as the time integrator and trapezoidal rule for numerical integration.
Upon time discretization, we aim to compute $G_{k_1k_2}(-n\Delta t, n\Delta t)$ for $n=0,\cdots,N$.
Based on \cref{eq_integro_differential_equation}, to obtain the value of $G_{k_1k_2}(l_1\Delta t, (l_2+1)\Delta t)$, we need the knowledge of all the values of $G_{k_1k_2}(l'\Delta t, l''\Delta t)$ with $l_1\leqslant l' \leqslant l'' \leqslant l_2$ to evaluate of the integrand. Therefore, one should implement the iteration according to a proper order. In general, one should compute those full propagators $G_{k_1k_2}(l'\Delta t, l''\Delta t)$ with smaller value of $l'' - l'$ first. For example, we can compute all full propagators in the table below from top to bottom and left to right: 
\begin{equation}\label{tab:G}
\renewcommand\arraystretch{1.5}
\begin{matrix*}[l]
    &\G{-N}{-N} & \G{-N+1}{-N+1} & \cdots & \G{N-2}{N-2} & \G{N-1}{N-1} & \G{N}{N} \\
    &\G{-N}{-N+1} & \G{-N+1}{-N+2} & \cdots & \G{N-2}{N-1} & \G{N-1}{N} \\
    &\G{-N}{-N+2} & \G{-N+1}{-N+3} & \cdots & \G{N-2}{N} \\
    &\vdots &\vdots & \iddots \\
    &\G{-N}{N-1} & \G{-N+1}{N} \\
    &\G{-N}{N}
\end{matrix*}
\end{equation}
where $\G{j_1}{j_2}$ is the numerical results of $G_{k_1k_2}(j_1\Delta t, j_2\Delta t)$.
The values in the first row are given by the initial value that $G_{k_1k_2}(s',s')=1$. Once all full propagators in the table are obtained, we compute the values of bath influence functional are computed exactly and get the values of $\tilde{Q}_{k_1k_2}(s_j)$ from the trajectories in \cref{eq_dynamics_discrete}, and the value of reduce density operator on grid points can be ensured by \cref{eq_rho_varrho}. We remark that $\tilde{Q}_{k_1k_2}(t)$ has a discontinuity at $t = 0$. During the calculation, we split all integral domains at $t = 0$ and the values of $\tilde{Q}_{k_1k_2}(0)$ takes either $Q_{k_1}(0)$ or $Q_{k_2}(0)$ based on different subdomains.

Before ending the discussion on inchworm method, we would like to introduce the following symmetric property of the coefficient $\varrho_{k_1k_2}(t)$ which can reduce the overall computational cost by half:
\begin{proposition}
    Given any $k_1,k_2 = 1,\cdots,K$, we have  \begin{equation}
\label{eq_varrho_conjugate_symmetry}
     \varrho_{k_1 k_2}(t)
    =  \varrho_{k_2 k_1}^*(t).
\end{equation}
\end{proposition}
\begin{proof}
    Based on the definition \cref{eq_varrho} of $\varrho_{k_1 k_2}(t)$, we apply change of variable $\boldsymbol{s}' = (s'_1,s'_2,\cdots,s'_m) = (- s_{m},-s_{m-1},\cdots,-s_1)$ to get 
    \begin{displaymath}
        \varrho_{k_1 k_2}(t)
    =  \sum_{m=0}^\infty 
    \int_{-t\leqslant \boldsymbol{s}'\leqslant t} 
    \prod_{j=1}^m 
    \left(-\ii \sgn ( - s_j')
        \tilde{Q}_{k_1 k_2}( - s_j')
    \right)
    \mathcal{L}_b(-s'_m,-s'_{m-1},\cdots,-s'_1)
    \dd \boldsymbol{s}'.
    \end{displaymath}
By the definition of correlation function \cref{bath_function_B}, one can easily check that $B(-\tau_2,-\tau_1) = B^*(\tau_1,\tau_2)$, which immediately yields that $\mathcal{L}_b(-\tau_m,-\tau_{m-1},\cdots,-\tau_1) = \mathcal{L}_b^*(\tau_1,\cdots,\tau_{m-1},\tau_m)$ by the definition of bath influence functional \cref{Lb}. Thus, we further have 
\begin{displaymath}
     \varrho_{k_1 k_2}(t)
    = \sum_{m=0}^\infty 
    \int_{-t\leqslant \boldsymbol{s}'\leqslant t} 
    \prod_{j=1}^m 
    \left(-\ii \sgn (-  s_j')
        \tilde{Q}_{k_2 k_1}( s_j')
    \right)
    \mathcal{L}_b^*(s'_1,s'_{2},\cdots,s'_m)
    \dd \boldsymbol{s}'  .
\end{displaymath}
Finally, we can change $\sgn(-s'_j)$ to $\sgn(s'_j)$ in the above formula since $m$ is an even number, and then we arrive at the relation \cref{eq_varrho_conjugate_symmetry}.
\end{proof}
With this property, the summation in 
\cref{eq_rho_varrho} can therefore be divided into ``diagonal'' and ``non-diagonal'' part:
\begin{equation}
\label{eq_varrho_summation}
     {\bra{x_1}\rho_s(t)\ket{x_2} }= \sum_{k=1}^K w_k^2\varrho_{kk}(t) \psi_{k}(t,x_1)\psi_{k}^*(t,x_2)
    % + 2 \sum_{(k_1,l_1)<(k_2,l_2)}
    + 2 \sum_{k_1<k_2}
    w_{k_1}w_{k_2}\operatorname{Re}\left(\varrho_{k_1k_2}(t)\psi_{k_1}(t,x_1)\psi_{k_2}^*(t,x_2)\right)
\end{equation}
with roughly half of the computational cost of the original full summation.

When computing the position probability distribution $f(t,x)$, we take the values $G_{k_1k_2}^{-n,n}$ for $n = 0,\cdots,N$ in \cref{tab:G} and only compute the ``diagonal'' of $ {\bra{x_1}\rho_s(t)\ket{x_2} }$ as
\begin{equation}
\label{eq_probability_density_summation}
    f(t,x) = \sum_{k=1}^K w_k^2\varrho_{kk}(t) \psi_{k}(t,x)\psi_{k}^*(t,x)
    + 2 \sum_{k_1<k_2}
    w_{k_1}w_{k_2}\operatorname{Re}\left(\varrho_{k_1k_2}(t)\psi_{k_1}(t,x)\psi_{k_2}^*(t,x)\right).
\end{equation}

\subsection{Summary of the algorithm}
\label{sec_summary}
We would like to end this section by a quick summary of the method and a brief analysis of the computational cost.
In general, our method combines the frozen Gaussian approximation and the inchworm algorithm.
We first use the FGA ansatz to represent the wave function and solve the ODEs for the dynamics.
In this step, the computational cost is $\mathcal{O}(KN_{\mathrm{RK}})$, where $K$ is the number of beams in \cref{eq_initial_decomposition}, and $N_{\mathrm{RK}}$ is the number of Runge-Kutta steps in the discretization.
This step is identical to the FGA method without baths.
The second step is to apply the inchworm algorithm to all possible pairs of beams.
The order of computational cost is $\mathcal{O}(K^2 C_{\mathrm{IW}}(\bar{M}))$ where $K^2$ comes from the choices of $k_1,k_2$ in \cref{eq_integro_differential_equation} and $C_{\mathrm{IW}}(\bar{M})$ is the computational cost for solving \cref{eq_integro_differential_equation} with fixed $k_1,k_2$.
If we have $N$ grid points for time discretization, we need to evaluate $\mathcal{O}(N^2)$ full propagators for each beam pair.
For each full propagator, the computational cost is at most $\mathcal{O}(N^{\bar{M}})$ for the numerical quadrature. 
Therefore, $C_{\mathrm{IW}}$ can be estimated by $\mathcal{O}(N^{\bar{M}+2})$.
Readers may refer to \cite{cai2022fast,wang2023real} for more discussion on the computational cost of the inchworm algorithm.
The last step is to recover the probability density function by \cref{eq_probability_density_summation}.
The computational cost of this step depends on the number of points on which we would like to evaluate the function.
The cost is expected to be $\mathcal{O}(K^2 C_{\mathrm{X}})$ with $C_{\mathrm{X}}$ represents the computational cost to evaluate each term in summation for all required grid points $x$.

As a summary, the computational cost of the three parts are $\mathcal{O}(K N_{\mathrm{RK}})$, $\mathcal{O}(K^2 N^{\bar{M}+2})$ and $\mathcal{O}(K^2 C_X)$.
Since we usually use methods with the same order of accuracy for solving the FGA method and the inchworm equation, it is advisable to choose $N_{\mathrm{RK}} \sim \mathcal{O}(N)$.
In practice, we can choose $N_{\mathrm{RK}}$ to be larger than $N$ to get more accurate approximations of $Q_k(t)$.
In the third step, the value of $C_X$ comes from the spatial discretization which needs to be considered only when generating the output.
To resolve the fluctuations in the solution, $C_X$ needs to be proportional to $\mathcal{O}(1/\sqrt{\epsilon})$.
The discretization \cref{eq_grids} shows that $K \sim \mathcal{O}(1/\epsilon)$, from which one can conclude that the total computational cost of the third step is $\mathcal{O}(K^{5/2})$.
To conclude, the final computational cost can be estimated by $\mathcal{O}(K^{5/2} + K^2 N^{\bar{M}+2})$.

The entire procedure is summarized in \cref{algo}.
Of all these steps, the inchworm step is in general the most costly one for most applications.
%Summing up the cost, the computational cost scales quadratically with respect to the number of beams $K$.
%Since the discretization of $p$ and $q$ should be comparable with the small parameter $\epsilon$, our method in general requires large $K$, which brings heavy computational load due to the quadratic scale of $K$.
We will discuss an alternative discretization in \cref{sec_conclusion} which may require fewer beams.

% To be specific, the first step is to use FGA ansatz to represent the wave function and solve some ODEs for the dynamics of the parameters of FGA.
% Inchworm algorithm is then applied for each pair of beams in the ansatz.
% The last step is to combine all the beams to recover the reduced density operator. We hereby summarize the full algorithm by the pseudocode in \cref{algo}. From \cref{algo}, the computational complexity is $\mathcal{O}(K^2)$ due to the loops to solve inchworm integro-differential equation \cref{eq_integro_differential_equation}.
\begin{algorithm}[ht]
    \caption{Inchworm-FGA}
    \label{algo}
    \begin{algorithmic}[1]
        \State \textbf{Input:} $\psi_{s}^{(0)}$, $\epsilon$;
        \State Initialize $\psi_{s}^{(0)}$ according to \cref{eq_initial_decomposition};
        \For{$k=1,\dots,K$}
            \State Solve the beam dynamics according to \cref{eq_dynamics_discrete};
        \EndFor
        \For{$1\leqslant k_{1} \leqslant k_{2} \leqslant K$}
            \State Compute the full propagators for each beam according to \cref{eq_integro_differential_equation};
        \EndFor
        \State Compute $\rho_s(t,x_{1},x_{2})$ by  \cref{eq_varrho_summation} or $f(t,x)$ by \cref{eq_probability_density_summation};
    \end{algorithmic}
\end{algorithm}

% \begin{algorithm}[ht]
%     \caption{Inchworm-FGA}
%     \label{algo}
%     \begin{algorithmic}[1]
%         \State \textbf{Input} $\epsilon, \psi_0, \Delta t, \Delta t_{\mathrm{RK}}, \bar{M}, t$;
%         \State \textbf{Require} $\Delta t$ is a multiple of $\Delta t_{\mathrm{RK}}$, $t=N\Delta t$ with integer $N$;
%         \State Choose proper $p_{\min},q_{\min},\Delta p, \Delta q, N_p, N_q$;
%         \For {$p_k=p_{\min}+k\Delta p, q_l = q_{\min}+l\Delta q$ for $k=0,\dots, N_p$ and $l = 0,\dots,N_q$}
%             \State Solve the trajectory of $P_{kl},Q_{kl},S_{kl},a_{kl}$ with time step $\Delta t_{\mathrm{RK}}$ by \cref{eq_dynamics_discrete};
%         \EndFor
%         \For {$J_1$ in $\{0,\dots,(N_p+1)(N_q+1)-1\}$}
%             \State $k_1\gets [J_1/(N_q+1)]$;
%             \State $j_1\gets J_1 - k_1(N_q+1)$;
%             \For {$J_2$ in $\{0,\dots,J_1\}$}
%                 \State $k_2\gets [J_2/(N_q+1)]$;
%                 \State $j_2\gets J_2 - k_2(N_q+1)$;
%                 \State Solve \cref{eq_integro_differential_equation} for $s_\ii = j_1 \Delta t, s_\ff= j_2 \Delta t$ with $-N\leqslant j_1\leqslant j_2 \leqslant N$;
%             \EndFor
%         \EndFor
%         \State Compute $f(t,x)=\rho_s(t,x,x)$ by \cref{eq_rho_varrho};
%     \end{algorithmic}
% \end{algorithm}
% \sy{I don't think we want to write too many details such as grid selecting in a pseudo code...}\ki{I will try to modify...}
{
\section{Error analysis for reduced density operator approximated by FGA}
\label{sec:error analysis}
We use this section to carry out a numerical analysis for the error bound of the reduced density operator formed as Dyson series approximated with the FGA. To begin with, let us first introduce some useful notations. Given a FGA wave function $\psi_{\FGA}$ as defined in \cref{eq_FGA_ansatz}, we define the FGA position operator
\begin{displaymath}
    \hat{Q}\psi_{\FGA}(t,x) :=  \frac{1}{({2\pi\epsilon})^{3/2}}
    \int_{\mathbb{R}^2}
    \int_{-\infty}^\infty
    Q(t,p,q)a(t,p,q) \e^{\ii \phi(t,x,y,p,q) / \epsilon}
    \psi_0(y)
    \dd y \dd p \dd q.
\end{displaymath}
In addition, given a $(m+2)$-dimensional vector 
\begin{displaymath}
    \bs := (s_0,s_1,\cdots,s_{m+1})
\end{displaymath}
satisfying $0 = s_0 < s_1 <s_2 <\cdots < s_m  < s_{m+1}= t$, we define 
\begin{equation}\label{def:VFGA}
    \mathcal{V}_{\FGA}^{(k)}[\psi_0](\bs,x):=  \frac{1}{({2\pi\epsilon})^{3/2}}
    \int_{\mathbb{R}^2}
    \int_{-\infty}^\infty
    \left(\prod_{j=1}^{k-1}Q(s_j,p,q)\right)a(s_k,p,q) \e^{\ii \phi(s_k,x,y,p,q) / \epsilon}
    \psi_0(y)
    \dd y \dd p \dd q
\end{equation}
which effectively approximates system propagation up to time $s_k$:
\begin{equation*}
    \mathcal{V}_{\FGA}^{(k)}[\psi_0](\bs,x) \approx   \mathcal{V}^{(k)}[\psi_0](\bs,x) := G^{(s)}(s_{k-1},s_k)\hat{x} \cdots   \hat{x}G^{(s)}(s_0,s_1) \psi_0 .
\end{equation*}
In particular, $\mathcal{V}_{\FGA}^{(0)}[\psi_0](\bs,x) = \psi_0(x)$ which is the initial value of the wave function considered. Here we remark that $\mathcal{V}_{\FGA}^{(0)}[\psi_0](\bs,x)$ actually does not relies on $\bs$ according to the definition \cref{def:VFGA}. We also denote the exact propagation up to time $t$ and its frozen Gaussian approximation by 
\begin{displaymath}
   \mathcal{V}[\psi_0](\bs,x) :=  \mathcal{V}^{(m+1)}[\psi_0](\bs,x) \quad, \quad  \mathcal{V}_{\FGA}[\psi_0](\bs,x) :=  \mathcal{V}_{\FGA}^{(m+1)}[\psi_0](\bs,x)
\end{displaymath}
which correspond to $\mathcal{U}^-_{x_1}$ and its approximation \cref{U_minus_FGA} (and also $\mathcal{U}^+_{x_2}$ and the corresponding approximation). Furthermore, we define
\begin{displaymath}
    \mathcal{W}^{(k)}[\psi_0](\bs,x) = G^{(s)}(s_{m},s_{m+1})\hat{x}G^{(s)}(s_{m-1},s_m) \hat{x} \cdots   \hat{x}G^{(s)}(s_{k},s_{k+1}) \hat{Q} \mathcal{V}_{\FGA}^{(k)}[\psi_0](\bs,x),
\end{displaymath}
for $k=1,\cdots,m$, which can be viewed as an approximate system influence functional obtained by first evolving $\psi_0$ under FGA dynamics up to time $s_k$, and then evolving under exact quantum dynamics until time $t$. In addition, we define
\begin{displaymath}
    \mathcal{W}^{(0)}[\psi_0](\bs,x) = \mathcal{V}[\psi_0](\bs,x) \quad , \quad \mathcal{W}^{(m+1)}[\psi_0](\bs,x) = \mathcal{V}_{\FGA}[\psi_0](\bs,x).
\end{displaymath}
Finally, we assume the two-point correlation \cref{bath_function_B} is bounded by  $$|B(\cdot,\cdot)|\leqslant C_B,$$ which is verifed numerically in Figure \ref{fig_tpc}. By definition \cref{Lb}, bath influence functional is then bounded by 
\begin{equation}\label{Lb_bound}
    |\mathcal{L}_b(s_1,\cdots,s_m)| \leqslant (m-1)!! C_B^{m/2}.
\end{equation}

With the notations above, we first introduce the error estimation for system influence functional approximated by the procedures \cref{semigroup_FGA} -- \cref{U_minus_FGA}, which is given by the following lemma:
\begin{lemma}\label{lemma:influence_functional_FGA}
Let $\psi(t,x)$ be a wave function with initial value $\psi_0(x)$ and its FGA $\psi_{\FGA}$ given by \cref{eq_FGA_ansatz}. Assume 
     \begin{equation}\label{proof assumption}
        \mathrm{supp}(\mathcal{W}^{(k)}[\psi_0](\bs,\cdot)  ) \subset [-R,R] \text{~for all~}  0\leqslant \bs \leqslant t \text{~and~} k = 1,\cdots,m+1
    \end{equation}     
    % \begin{equation}\label{proof assumption}
    %     \mathrm{supp}(\psi(\tau,\cdot)) \ ,\ \mathrm{supp}(\psi_{\FGA}(\tau,\cdot)) \subset [-L,L] \text{~for all~} \tau \in [0,t]
    % \end{equation}
where $R$ is some positive number. Without loss of generality, we assume $R\geqslant 1$. Then it holds that 
\begin{equation}\label{V_VFGA}
    \| \mathcal{V}[\psi_0](\bs,\cdot) -  \mathcal{V}_{\FGA}[\psi_0](\bs,\cdot)  \|_{L^{2}} \leqslant  C_3(m) \epsilon 
\end{equation}
where $$C_3(m) =  \frac{1}{1-\alpha^{-1}} \left(C_2 + \sqrt{\frac{2}{3}}R^{3/2} C_1\right)\alpha^{m-1} \text{~and~} \alpha = \sqrt{2/3}R^{3/2}.$$
\end{lemma}

\begin{proof}
We first decompose  
\begin{equation}
    \mathcal{V}_{\FGA}[\psi_0](\bs,x) - \mathcal{V}[\psi_0](\bs,x)  = \sum^m_{k=0} \mathcal{W}^{(k+1)}[\psi_0](\bs,x) -  \mathcal{W}^{(k)}[\psi_0](\bs,x).
\end{equation}
Each term in the sum takes the form
\begin{equation*}
    \begin{split}
   & \mathcal{W}^{(k+1)}[\psi_0](\bs,x) -  \mathcal{W}^{(k)}[\psi_0](\bs,x)\\
    =  & \   G^{(s)}(s_{m},t)\hat{x}G^{(s)}(s_{m-1},s_m) \hat{x} \cdots   \hat{x}G^{(s)}(s_{k+1},s_{k+2}) \times \\
   &\times  \left[ \hat{Q} \mathcal{V}_{\FGA}^{(k+1)}[\psi_0](\bs,x)  - \hat{x}G^{(s)}(s_{k},s_{k+1}) \hat{Q} \mathcal{V}_{\FGA}^{(k)}[\psi_0](\bs,x)  \right].
    \end{split}
\end{equation*}
By Proposition \ref{prop_scale_coefficient}, we can obtain $ \mathcal{V}_{\FGA}^{(k+1)}$ by performing FGA dynamics according to ODEs \cref{eq_dynamics} for time $s_{k+1} - s_k$ using $\hat{Q} \mathcal{V}_{\FGA}^{(k)}$ (which is in FGA form) as the initial value. Furthermore, by Proposition \ref{prop_FGA} and Proposition \ref{prop_position_operator_beam_center} we estimate 
\begin{equation*}
    \begin{split}
   & \|\hat{Q} \mathcal{V}_{\FGA}^{(k+1)}[\psi_0](\bs,\cdot)  - \hat{x}G^{(s)}(s_{k},s_{k+1}) \hat{Q} \mathcal{V}_{\FGA}^{(k)}[\psi_0](\bs,\cdot) \|_{L^2}  \\
   = & \ \|\hat{Q} \mathcal{V}_{\FGA}^{(k+1)}[\psi_0](\bs,\cdot)  - \hat{x} \mathcal{V}_{\FGA}^{(k+1)}[\psi_0](\bs,\cdot) + \hat{x} \mathcal{V}_{\FGA}^{(k+1)}[\psi_0](\bs,\cdot) -  \hat{x}G^{(s)}(s_{k},s_{k+1}) \hat{Q} \mathcal{V}_{\FGA}^{(k)}[\psi_0](\bs,\cdot)  \|_{L^2} \\
   \leqslant & \ \|\hat{Q} \mathcal{V}_{\FGA}^{(k+1)}[\psi_0](\bs,\cdot)  - \hat{x} \mathcal{V}_{\FGA}^{(k+1)}[\psi_0](\bs,\cdot) \|_{L^2} + \\
   &+ \|  \hat{x}\|_{L^2} \cdot \| \mathcal{V}_{\FGA}^{(k+1)}[\psi_0](\bs,\cdot) - G^{(s)}(s_{k},s_{k+1}) \hat{Q} \mathcal{V}_{\FGA}^{(k)}[\psi_0](\bs,\cdot)  \|_{L^2} \\
   \leqslant & \  C_2\epsilon + \sqrt{\frac{2}{3}}R^{3/2} \cdot C_1\epsilon 
    \end{split}
\end{equation*}
where we have used the assumption that $x \in [-R,R]$. In addition, since $G^{(s)}(\cdot,\cdot)$ is unitary, overall we have 
\begin{equation}\label{V_VFGA_each}
    \begin{split}
    &\| \mathcal{W}^{(k+1)}[\psi_0](\bs,\cdot) -  \mathcal{W}^{(k)}[\psi_0](\bs,\cdot) \|_{L^2} \\ \leqslant \ 
    &  \|\hat{x}\|^{m-k-1}_{L^2} \cdot  \|\hat{Q} \mathcal{V}_{\FGA}^{(k+1)}[\psi_0](\bs,x)  - \hat{x}G^{(s)}(s_{k},s_{k+1}) \hat{Q} \mathcal{V}_{\FGA}^{(k)}[\psi_0](\bs,x) \|_{L^2}\\
    \leqslant \ & \left( \sqrt{\frac{2}{3}}R^{3/2}\right)^{m-k-1} (C_2 + \sqrt{\frac{2}{3}}R^{3/2} C_1)\epsilon.
        \end{split}
\end{equation}
The bound \cref{V_VFGA} is an immediate consequence by summing over all \cref{V_VFGA_each}.
\end{proof}
The assumption \cref{proof assumption} can be achieved by assuming proper compactly supported potential functions $V$ which can confine the particle within $[-R,R]$ throughout the entire simulations. The potential functions we consider in numerical experiments in Section \ref{sec_numerical_examples} all satisfy this assumption. This lemma shows that the approximation error of system influence functional with FGA is $\mathcal{O}(\epsilon)$ with a prefactor scaling exponentially with respect to the order of the Dyson series. At first glance, this seems to be a large error if we expand Dyson series to high order. Fortunately, since the integrals in Dyson series are defined on simplices $$\{(s_1,s_2,\cdots,s_m)| -t < s_1 < s_2<\cdots<s_m<t \}$$ whose volumes are $\frac{(2t)^m}{m!!}$, the approximation error for each term in Dyson series can be well controlled by the fast double factorial decay with respect to $m$. The following theorem provides an upper bound for the overall approximation error for the reduced density operator with FGA:
% Consequently, the overall approximation error for the reduced density operator with FGA grows with simulation time at double exponential rate as stated in Theorem \ref{thm:rho_fga_err} below, which is the same rate as the variance of bare dQMC \cite[Section 5]{cai2020inchworm}. 

\begin{theorem}\label{thm:rho_fga_err}
Given reduced density operator $\rho_s(t)$ in Dyson series formulation as \cref{eq_dyson}. Let 
\begin{equation*}
    \tilde{\rho}_s(t) 
    = \sum_{m=0}^{\infty}
    \int_{-t\leqslant \boldsymbol{s} \leqslant t}
    \prod_{j=1}^m 
    \left(-\ii \sgn (s_j)\right)
    \mathcal{U}_{\FGA}[\psi_s^{(0)}](-t,\boldsymbol{s},t)
    \mathcal{L}_b(\boldsymbol{s})
    \dd \boldsymbol{s}
\end{equation*}
where $\mathcal{U}_{\FGA}$ is the system influence functional approximated by FGA such that 
\begin{displaymath}
   \langle x_1 |  \mathcal{U}_{\FGA}[\psi_s^{(0)}](-t,\boldsymbol{s},t) | x_2 \rangle =  \mathcal{V}_{\FGA}[\psi^{(0)}_s](\bs^+,x_1) \left( \mathcal{V}_{\FGA}[\psi^{(0)}_s](\bs^-,x_2)\right)^\dagger
\end{displaymath}
where $\bs^+:=(0,s_{i+1},\cdots,s_{m},t)$, $\bs^-:=(-t,s_1,\cdots,s_{i},0)$ and $s_i<0<s_{i+1}$. Under the assumption \cref{proof assumption}, we have  
\begin{equation}\label{rho_fga_err}
    \| \rho_s(t) - \tilde{\rho}_s(t) \|_{L^2} \leqslant C_4 \exp(C_5 \cdot t^2) \epsilon
\end{equation}
where 
\begin{displaymath}
    C_4 =  \frac{2}{\alpha-1} \left(C_2 + \sqrt{\frac{2}{3}}R^{3/2} C_1\right) \quad , \quad C_5 = 2\alpha^2 R^2 C_B
\end{displaymath}
\end{theorem}

\begin{proof}
    According to Proposition \ref{prop_FGA} and Lemma \ref{lemma:influence_functional_FGA}, we have 
    \begin{align*}
       &\| \mathcal{U}[\psi_s^{(0)}](-t,\boldsymbol{s},t) 
 - \mathcal{U}_{\FGA}[\psi_s^{(0)}](-t,\boldsymbol{s},t)  \|_{L^2}\\
 % = & \   \int \int  \left[ \mathcal{V}[\psi^{(0)}_s](\bs^+,x_1) \left( \mathcal{V}[\psi^{(0)}_s](\bs^-,x_2)\right)^\dagger - \mathcal{V}_{\FGA}[\psi^{(0)}_s](\bs^+,x_1) \left( \mathcal{V}_{\FGA}[\psi^{(0)}_s](\bs^-,x_2)\right)^\dagger  \right]^2 \dd x_1 \dd x_2 \\
 \leqslant & \  \left\| \mathcal{V}[\psi^{(0)}_s](\bs^+,\cdot)  - \mathcal{V}_{\FGA}[\psi^{(0)}_s](\bs^+,\cdot)  \right\|_{L^2}  \cdot \|\mathcal{V}[\psi^{(0)}_s](\bs^-,\cdot)\|_{L^2} \\
 &+  \left\|  \mathcal{V}[\psi^{(0)}_s](\bs^-,\cdot)  -  \mathcal{V}_{\FGA}[\psi^{(0)}_s](\bs^-,\cdot) \right\|_{L^2}   \cdot \|\mathcal{V}_{\FGA}[\psi^{(0)}_s](\bs^+,\cdot) \|_{L^2}\\
 \leqslant & \ C_3(m)\epsilon\cdot \left( \|\mathcal{V}[\psi^{(0)}_s](\bs^-,\cdot)\|_{L^2} +  \|\mathcal{V}_{\FGA}[\psi^{(0)}_s](\bs^+,\cdot) \|_{L^2} \right) \\
 \leqslant & \  C_3(m)\epsilon\cdot \left(  \|\mathcal{V}[\psi^{(0)}_s](\bs^-,\cdot)\|_{L^2} +  \| \mathcal{V}[\psi^{(0)}_s](\bs^+,\cdot)\|_{L^2} + \|\mathcal{V}_{\FGA}[\psi^{(0)}_s](\bs^+,\cdot)  - \mathcal{V}[\psi^{(0)}_s](\bs^+,\cdot)\|_{L^2} \right)\\
 \leqslant & \ C_3(m)\epsilon\cdot (2R^m + C_1\epsilon ) \approx 2R^mC_3(m) \epsilon. 
    \end{align*}
  For the last inequality, we have used the assumption \cref{proof assumption}. By \cref{Lb_bound}, we bound 
\begin{align*}
\| \rho_s(t) - \tilde{\rho}_s(t)\|_{L^2} \leqslant & \ \sum_{\substack{m=0\\ m \text{~is even}}}^{\infty}
    \int_{-t\leqslant \boldsymbol{s} \leqslant t}
    \| \mathcal{U}[\psi_s^{(0)}](-t,\boldsymbol{s},t) 
 - \mathcal{U}_{\FGA}[\psi_s^{(0)}](-t,\boldsymbol{s},t)  \|_{L^2}
    \cdot |\mathcal{L}_b(\bs)|
    \dd \boldsymbol{s}\\ 
    \leqslant & \ \sum_{\substack{m=0\\ m \text{~is even}}}^{\infty} \frac{(2t)^m}{m!}\cdot 2R^mC_3(m) \epsilon \cdot (m-1)!! C_B^{m/2} \\
    = & \   \frac{2}{\alpha-1} \left(C_2 + \sqrt{\frac{2}{3}}R^{3/2} C_1\right)\epsilon\cdot  \sum^{\infty}_{k=0} \frac{(2\alpha^2R^2 C_B t^2)^k}{k!},
\end{align*}
which gives the estimation \cref{rho_fga_err} by noting the sum above is the Taylor expansion of an exponential function.  

\end{proof}
}

% Let 
% \begin{align}
%     \mathcal{V}(\bs)\psi :=  \ & G^{(s)}(s_m,t)\hat{x}G^{(s)}(s_{m-1},s_m)\hat{x} \cdots  \hat{x}G^{(s)}(0,s_1) \psi_0,\\
%     \mathcal{V}_{\FGA}(\bs)\psi :=  \ & G^{(s)}_{\FGA}(s_m,t)\hat{Q}G^{(s)}_{\FGA}(s_{m-1},s_m)\hat{Q} \cdots  \hat{Q}G^{(s)}_{\FGA}(0,s_1) \psi_0,
% \end{align}
% where $\psi_0(x) = \psi(0,x)$. $\hat{Q}$ and $G^{(s)}_{\FGA}(\cdot,\cdot)$ are operators acting on functions in FGA form
% \begin{align}
%     &\hat{Q}\psi_{\FGA}(t,x) = \  \frac{1}{({2\pi\epsilon})^{3/2}}
%     \int_{\mathbb{R}^2}
%     \int_{-\infty}^\infty
%     Q(t,p,q)a(t,p,q) \e^{\ii \phi(t,x,y,p,q) / \epsilon}
%     \psi_0(y)
%     \dd y \dd p \dd q,\\
%     &G^{(s)}_{\FGA}(s_{k},s_{k+1})\psi_{\FGA}(\Si,x) = \ \psi_{\FGA}(\Sf,x) 
% \end{align}

\section{Numerical results}
\label{sec_numerical_examples}
In this section, we carry out some numerical experiments for different initial wave functions and different potentials.
We choose the Ohmic spectral density \cite{chakravarty1984dynamics,makri1999linear} in our numerical simulations.
For Ohmic spectral density, the frequencies $\omega_j$ and the coupling intensities $c_j$ are given by
\begin{equation}
    \omega_j = -\omega_c \log \left(
        1-\frac{j}{L}(1-\exp(-\omega_{\max}/\omega_c))
    \right)
\end{equation}
\begin{equation}\label{cj}
    c_j = \epsilon\omega_j \sqrt{\frac{\xi\omega_c}{L}(1-\exp(-\omega_{\max}/\omega_c))}
\end{equation}
for $j = 1,\cdots,L$.
% This corresponds to the Ohmic spectral density \cite{chakravarty1984dynamics,makri1999linear} in the literature.
{In this work, we mainly study the cases where coupling intensity between system and bath is generally not too large, so that convergence of Dyson series is not too slow and can be approximated well by a relatively smaller truncation $\bar{M}$. This requires  the two-point correlation function $\tilde{B}$ defined in \cref{bath_function_B} is $\mathcal{O}(1)$, and thus in the parameter setting, we choose the coupling intensity $c_j$ defined in \cref{cj} to be $\mathcal{O}(\epsilon)$.
% This introduces a perturbation of order $\mathcal{O}(\epsilon^2)$ in the potential function $V(x)$.
% However, in practice, $\epsilon$ is a small parameter so that the perturbation is almost negligible.
}

\subsection{Validity Tests}
\begin{figure}[t]
    \subfloat[$t=0$]{\includegraphics[width= 0.45\textwidth]{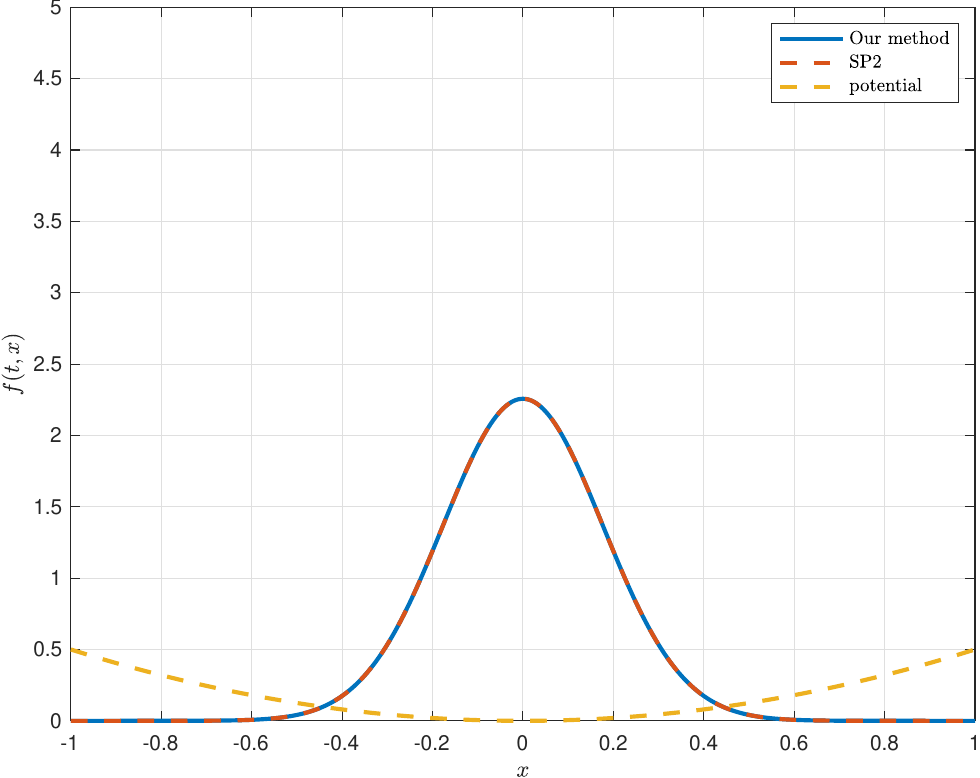}}
    \quad
    \subfloat[$t=1$]{\includegraphics[width= 0.45\textwidth]{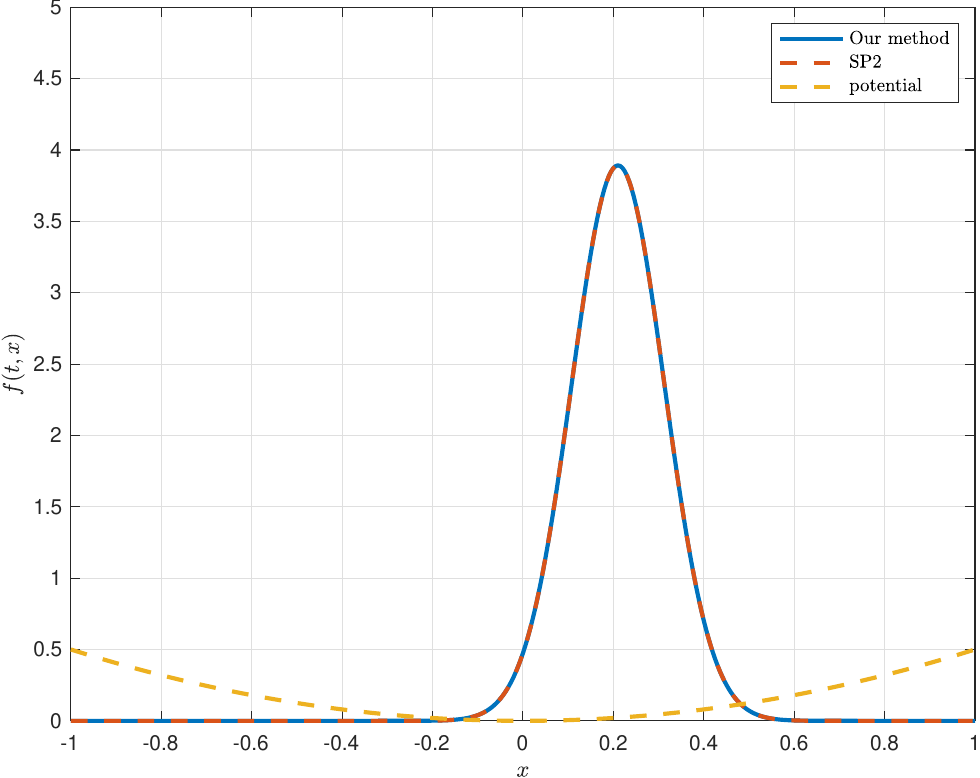}} \\
    \subfloat[$t=2$]{\includegraphics[width= 0.45\textwidth]{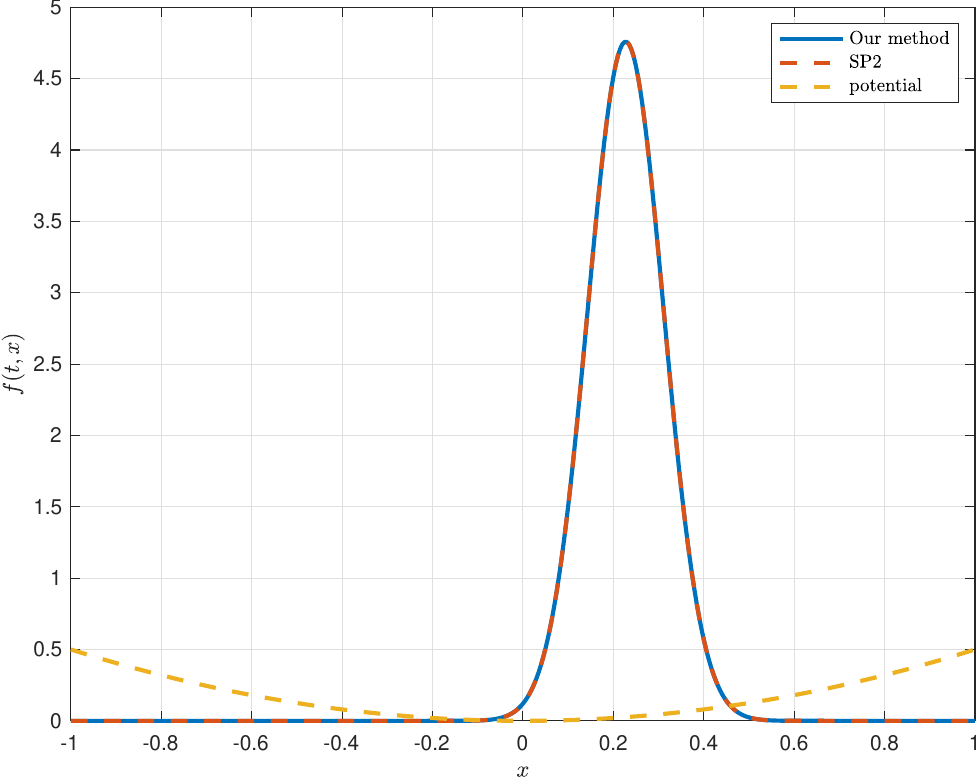}}
    \quad
    \subfloat[$t=3$]{\includegraphics[width= 0.45\textwidth]{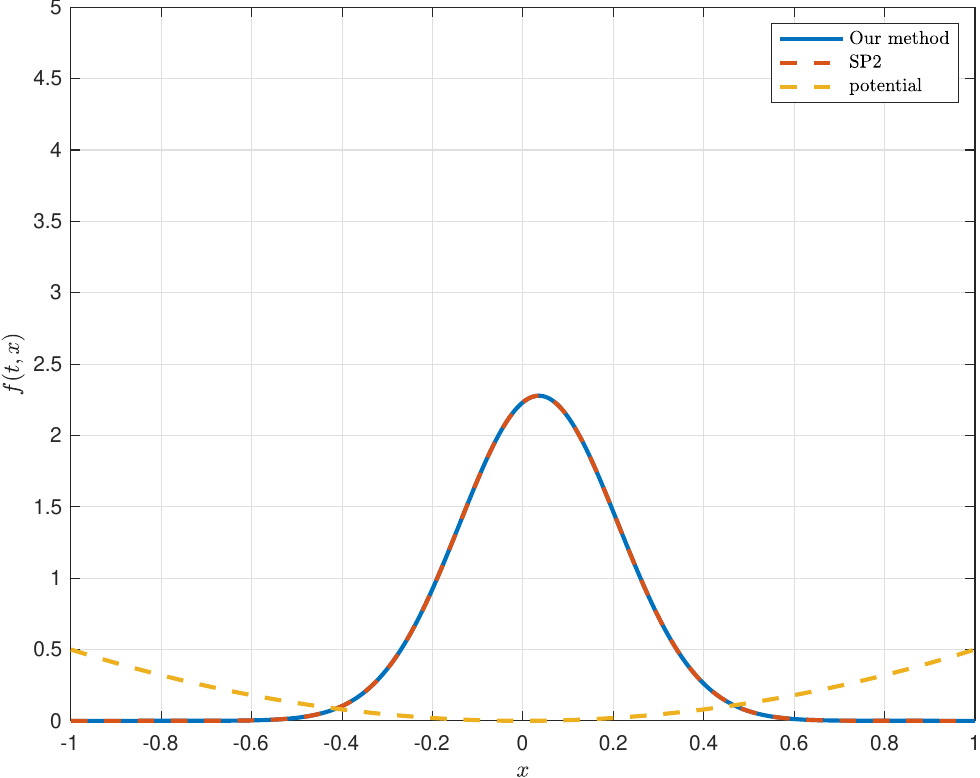}}
    \caption{Comparison of our method with SP2.}
    \label{fig_compare_with_SP2}
\end{figure}

To validate our method, we compare our numerical results with results from conventional methods based on numerical discretization as reference.
In particular, we compare the result of our method with the Strang splitting spectral method (SP2) for Schr\"odinger equation \cite{bao2002time,bao2003numerical}. Since the mesh discretization based SP2 will suffer from the curse of dimensionality from the large degrees of freedom of the harmonic oscillators, we only compare the result of our method with SP2 for zero system-bath coupling.
In our method, this can be achieved by $\xi=0$, and thus our method is essentially identical to the FGA method.
Since the system and bath evolve separately, a reference solution can be obtained by using SP2 to solve a closed system. In this example, we set the potential as
\begin{equation}
    V(x) = \frac{1}{2}x^2
\end{equation}
and the initial wave function as 
% \zc{Plot the initial value and the potential.}
\begin{equation}
    \psi_s^{(0)}(x) = \frac{1}{(4\pi\epsilon)^{1/4}}
    \exp\left(\ii \frac{x}{4\epsilon} - \frac{x^2}{8\epsilon}\right)
\end{equation}
with $\epsilon=\frac{1}{64}$.

In this experiment, the wave function has an initial momentum centered at $p = 1$, and the center of the wave function is at $x = 0$.
To cover the dynamics of most Gaussian wave packets, we set the range of $p$ as $[-1,3]$ and the range of $q$ as $[-2,2]$.
For $\epsilon = 1/64$, the standard deviation of the frozen Gaussian is $1/8$, and here we choose $\Delta p$ and $\Delta q$ to be both $\frac{1}{32}$, so that the oscillation in the wave function can be fully resolved.
The total number of beams $K$ is 16641.
It is shown in \cref{fig_compare_with_SP2} that the initial probability density can be well approximated by the FGA ansatz, and good agreement with the reference SP2 solution is maintained up to $t = 3$.
%The results in \cref{fig_compare_with_SP2} show consistency of our method and the reference SP2 method. 
This experiment shows that the FGA method can provide accurate solutions for these parameters, and similar numerical settings will be used in our examples with system-bath coupling.

{
While it is not practical to directly use SP2 to solve the dynamics of the Caldeira-Leggett model due to the large number of bath degrees of freedom, we remark that SP2 can still be applied to compute the system influence functional in \cref{eq_U_half} and subsequently calculate the reduced density matrix using the Dyson series formulation. However, it is not straightforward to extend this approach to the inchworm method. The inchworm method utilizes partial sums of the Dyson series, where the evolution of the system and the influence of the bath are mixed up in full propagators to accelerate the computation. These full propagators cannot be evaluated by solving closed Schr\"odinger equations, making the application of SP2 in this context nontrivial.

Although FGA works well in this validity tests, it might fail when $\epsilon$ is large. As stated in \cref{prop_FGA}, FGA has numerical accuracy of order $\mathcal{O}(\epsilon)$. For large $\epsilon$ such as $\epsilon=1$, FGA use wide beams to approximate the wave function. If the initial wave function has high frequency oscillates, FGA cannot even capture the initial states, let along carrying out numerical simulations.
We take the following wave function
\begin{equation}
    \psi_s^{(0)}(x) = \frac{1}{(\pi/16)^{1/4}}
    \exp\left(
        16\ii x - 8 x^2 
    \right)
\end{equation}
and use different $\epsilon$ in FGA to approximate this wave function. The results are shown in \cref{fig_different_epsilon}.
\begin{figure}
    \centering
    \includegraphics[width=0.45\linewidth]{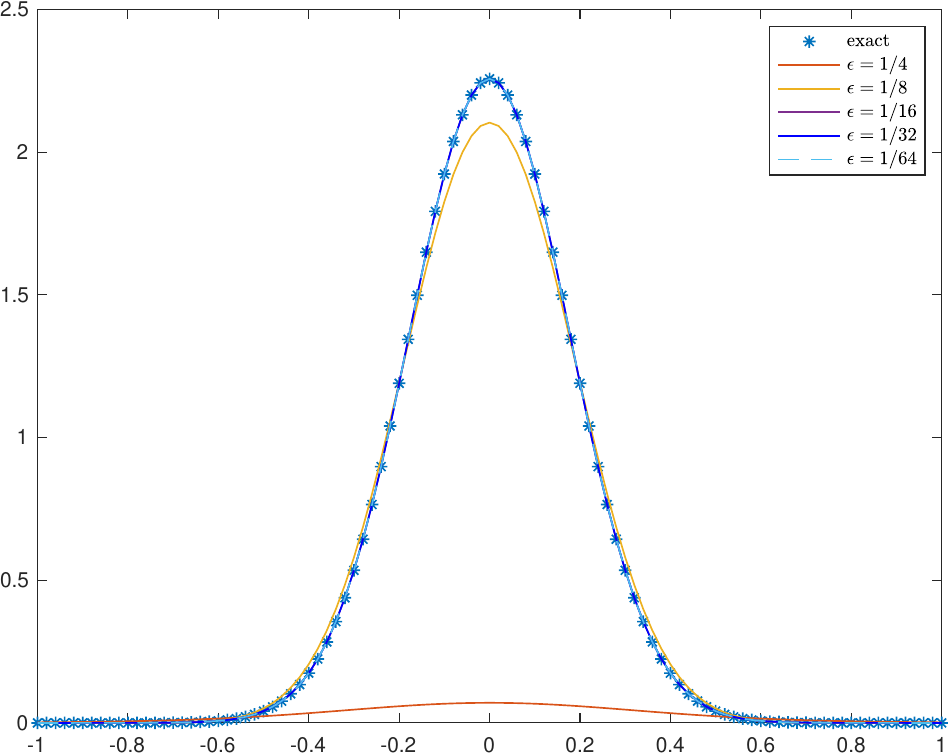}
    \caption{Wave function approximation by FGA with different $\epsilon$.}
    \label{fig_different_epsilon}
\end{figure}
From \cref{fig_different_epsilon}, we can observe that when we choose $\epsilon=1/8$ or larger, the approximation based on FGA is inaccurate. One possible way is directly evaluate the Dyson series without inchworm method, where we can use traditional numerical schemes such as SP2 to compute the system propagation. This can be considered as a future extension.
}

\subsection{Harmonic oscillator}
In this subsection, we consider the system bath coupling for different coupling intensities.
The potential function $V(x)$ and the initial wave function $\psi_s^{(0)}(x)$ are set the same as those in the previous validity test.
% \begin{equation}
%     V(x) = \frac{1}{2}x^2.
% \end{equation}
% That means the system of interest is also a quantum harmonic oscillator.
% The initial wave function of the particle is given by
% \begin{equation}
%     \psi_s^{(0)}(x) = \frac{1}{(4\pi\epsilon)^{1/4}}\exp\left(\ii\frac{x}{4\epsilon}-\frac{x^2}{8\epsilon}\right)
% \end{equation}
% with $\epsilon=\frac{1}{64}$.
Other parameters for the system-bath coupling are given by
\begin{equation}
    L = 400, \quad
    \omega_{\max} = 10,\quad
    \omega_c = 2.5,\quad
    \beta = 5.
\end{equation}
Numerical experiments are carried out for various coupling intensity $\xi \leqslant 1.6$. The amplitude of two-point correlation function $\tilde{B}(\Delta \tau)$ when $\xi=1.6$ is plotted in \cref{fig_tpc}. Since the initial condition is the same as the validity test, we also choose the same range of $p,q$ and same $\Delta p,\Delta q$. Recall that the number of beams $K$ is 16641. According to our estimate of the time complexity $\mathcal{O}(K^{5/2} + K^2 N^{\bar{M}+2})$, the computational cost is quite significant despite a one-dimensional problem.
\begin{figure}[t]
    \centering
    \includegraphics[width=0.45\linewidth]{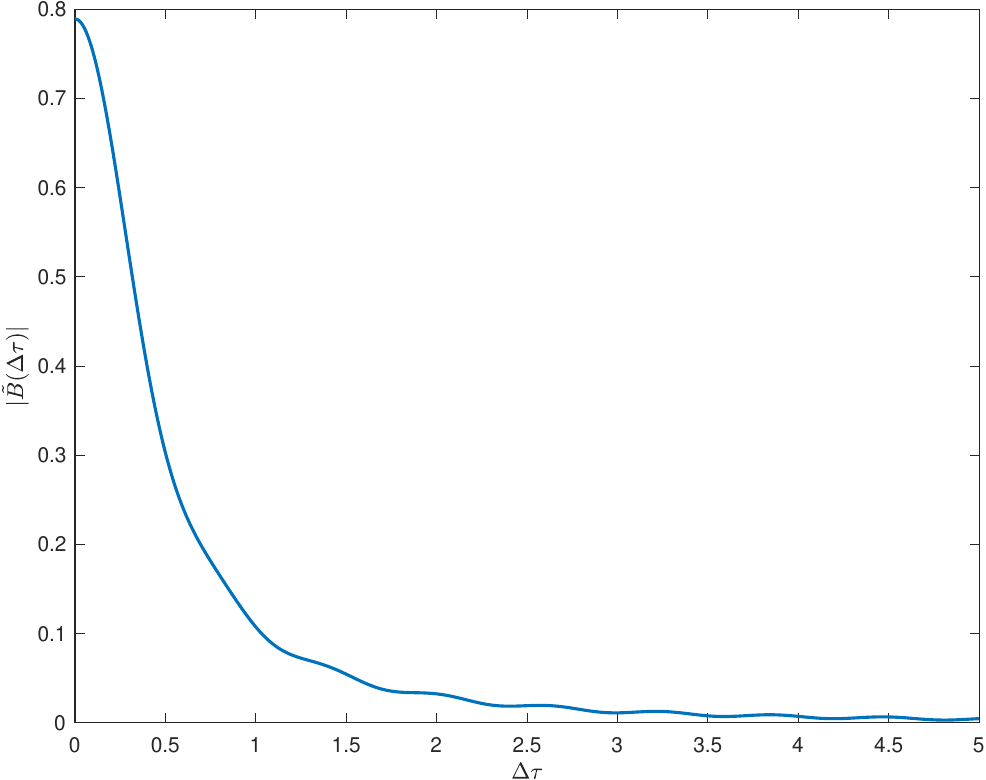}
    \caption{Amplitude of two-point correlation function.}
    \label{fig_tpc}
\end{figure}

In our simulations, we need to specify the truncation $\bar{M}$ for the inchworm expansion. As we have discussed in previous sections, the inchworm expansion generally converges very fast and thus we can choose a relatively small $\bar{M}$. Below we present our strategy for choosing the truncation: on the right-hand side of \cref{eq_integro_differential_equation}, the magnitude of 
$m$-th term in the inchworm expansion can be estimated by
\begin{equation} \label{eq:est}
\begin{split}
& \left| \int_{s_\ii \leqslant \boldsymbol{s}\leqslant s_\ff} 
        \prod_{j=1}^{m+1} 
        \left(-\ii \sgn(s_j) \tilde{Q}_{k_1k_2}(s_j)
            G_{k_1k_2}(s_{j-1},s_j)
        \right)
        \mathcal{L}_b^c([\boldsymbol{s},s_\ff])
   \dd \boldsymbol{s} \right| \\
   \leqslant & \   |\tilde{Q}_{\max} G_{\max}|^{m+1}  
      \int_{s_\ii \leqslant \boldsymbol{s}\leqslant s_\ff} \sum_{\mathfrak{q} \in \mathscr{Q}^c(\boldsymbol{s})}   
        \prod_{(\tau_1,\tau_2) \in \mathfrak{q}} | B(\tau_1, \tau_2) |
    \dd \boldsymbol{s}
\end{split}
\end{equation}
where $\tilde{Q}_{\max}$ and $G_{\max}$ are the maximum values of $\tilde{Q}_{k_1k_2}$ and $G_{k_1k_2}$, respectively.
To get a practical value of $\bar{M}$, we examine the decay of the upper bound \cref{eq:est} numerically. In the test example we considered in the previous subsection, we notice that the distribution of beams are confined in the region $[-0.5,0.5]$. Therefore, we set $\tilde{Q}_{\max} = 0.5$. In addition, we set $G_{\max} = 1$ from the initial condition and consider simulations up to $t = 5$. Under these parameter settings, we obtain the numerical evaluation of the upper bound \cref{eq:est} as 0.236671 for $m = 1$, and 0.0362826 for $m = 3$. Since the contribution from $m   = 3$ term in the inchworm expansion is minor compared with $m=1$, we therefore choose the truncation $\bar{M}=1$ in this numerical experiment. 

\begin{figure}[ht]
    \subfloat[$t=0$]{\includegraphics[width= 0.32\textwidth]{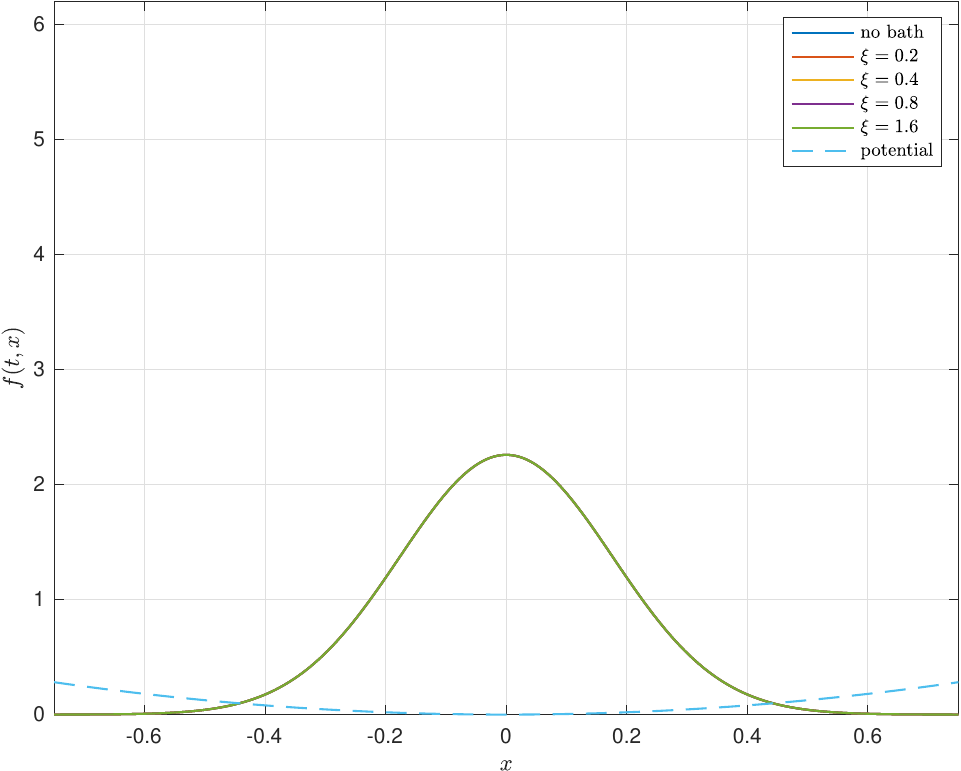}}
    \quad
    \subfloat[$t=1$]{\includegraphics[width= 0.32\textwidth]{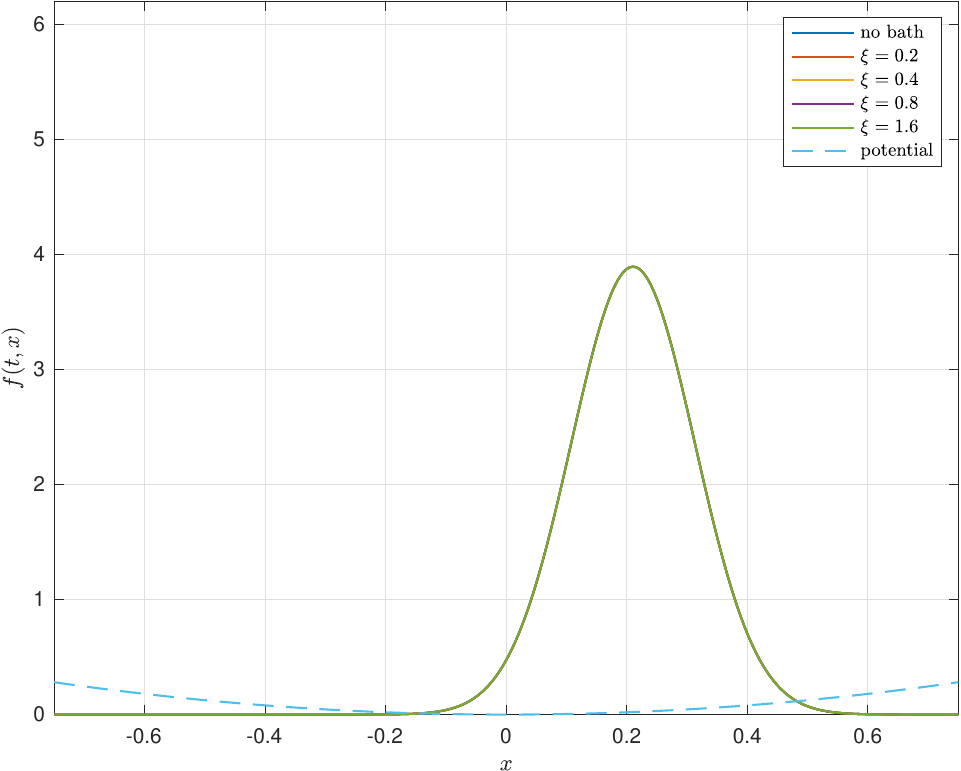}} \quad
    \subfloat[$t=2$]{\includegraphics[width= 0.32\textwidth]{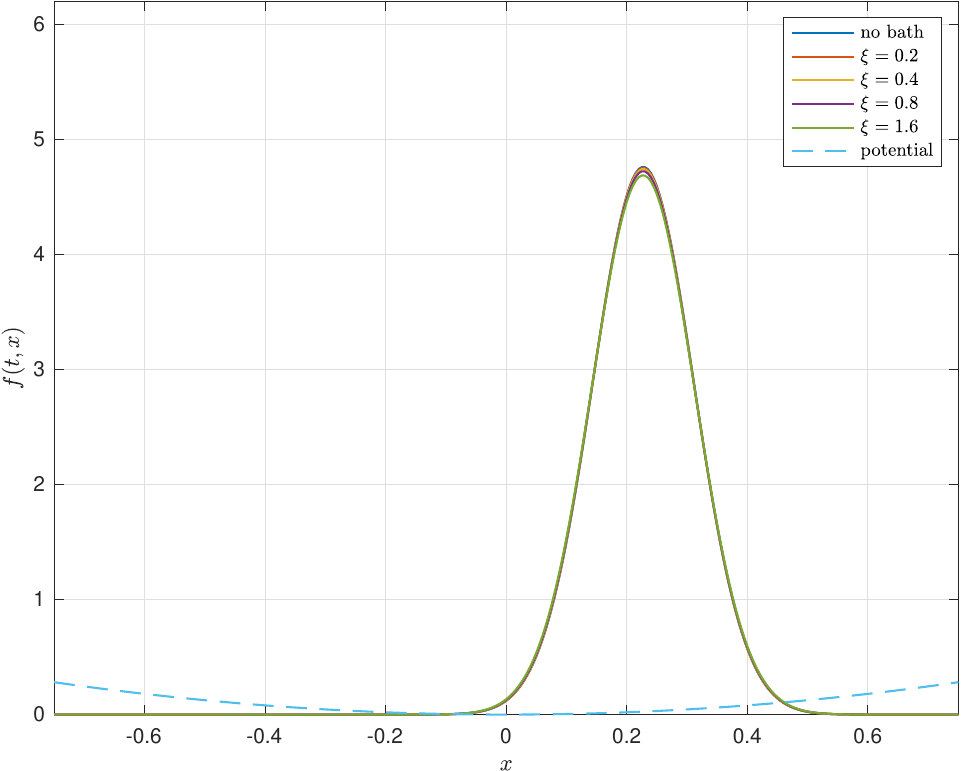}} \\
    \subfloat[$t=3$]{\includegraphics[width= 0.32\textwidth]{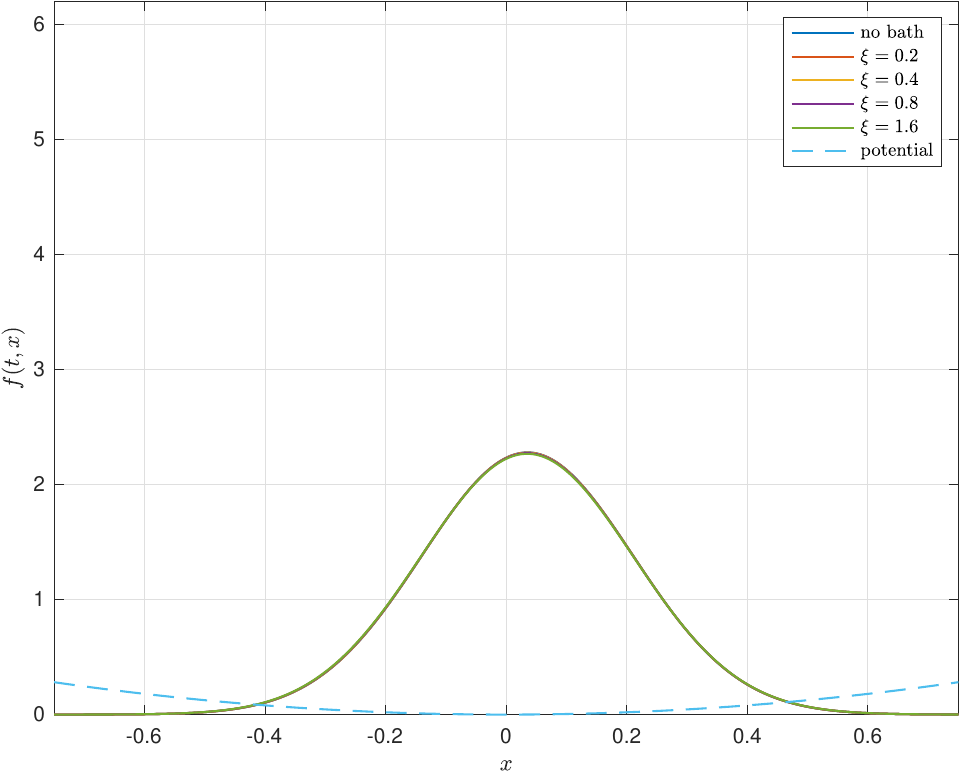}} \quad
    \subfloat[$t=4$]{\includegraphics[width= 0.32\textwidth]{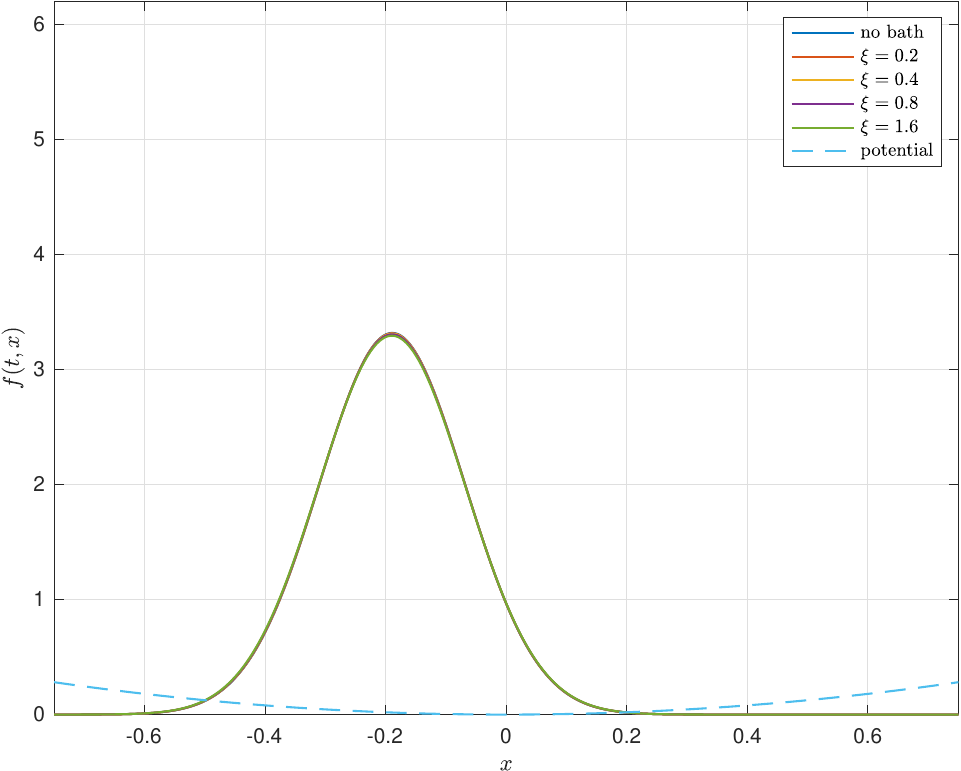}} 
    \quad
    \subfloat[$t=5$]{\includegraphics[width= 0.32\textwidth]{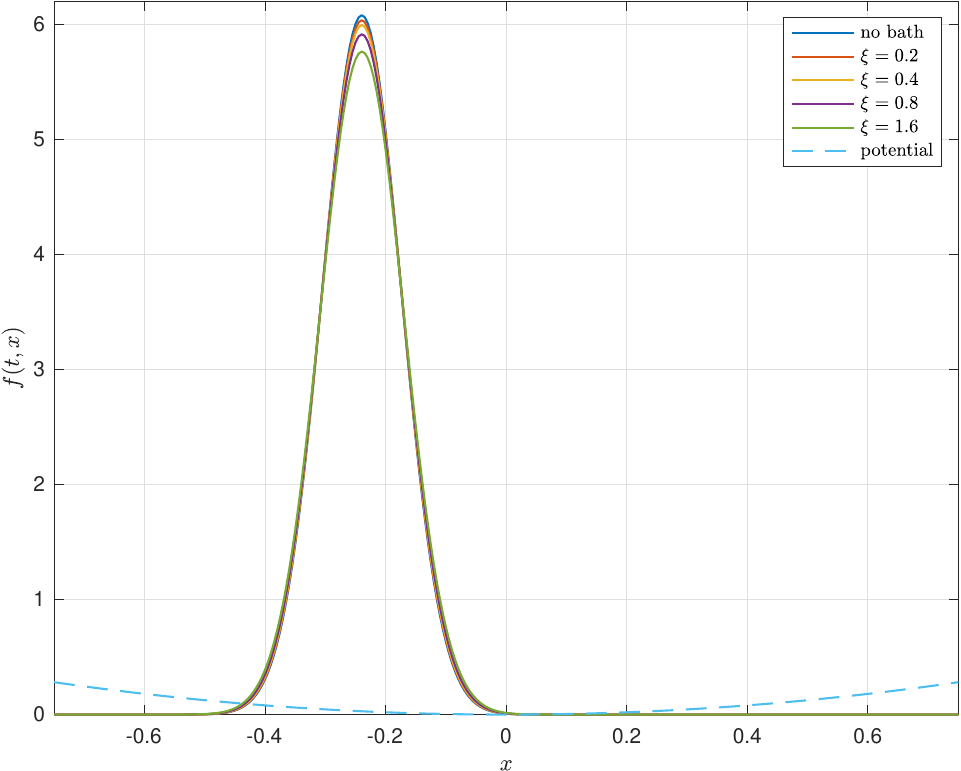}}
    \caption{Change of position probability density $f(t,x)$ of harmonic oscillator for different coupling intensities.}
    \label{fig_example_1}
\end{figure}
In other words, only the first diagram on the right side is considered when we solve the integro-differential equation \cref{eq_integro_differential_equation}. 
In previous works of the inchworm method \cite{yang2021inclusion}, it has been shown that choosing $\bar{M} = 1$ in the inchworm method can match the accuracy of the summation of Dyson series \cref{eq:G_Dyson} up to $m = 6$.
The numerical results for the simulations are given in \cref{fig_example_1}.

In this experiment, there is an initial momentum to the positive side for the particle, which can also be observed from the numerical result that the probability distribution moves to the right initially.
Due to the effect of the potential, it moves back to the left later.
From these numerical experiments, choosing different $\xi$ does not change the position probability distribution significantly.
When the coupling intensity $\xi$ increases, the peak in the position probability distribution of the system particle is lower.

\subsection{Two-peak initial wave function}
In this subsection, we change the initial wave function of the system particle to be a two-peak function:
\begin{equation}
    \psi_s^{(0)}(x) = C\left(\exp\left(
        -\frac{(x-1/2)^2}{4\epsilon}
    \right)
    + \frac{4}{5}\exp\left(
        -\frac{(x+1/2)^2}{4\epsilon}
    \right)\right)
\end{equation}
with $\epsilon = \frac{1}{64}$ and $C$ being the normalization constant:
\begin{equation}
    C = 5 (41+40\e^{-\frac{1}{8\epsilon}})^{-1/2}(2\pi\epsilon)^{-1/4}
    % C = 5\left((41+40\e^{-\frac{1}{8\epsilon}})\sqrt{2\pi\epsilon}\right)^{-1/2}
    % C = \frac{5}{\sqrt{(41+40\e^{-\frac{1}{8\epsilon}})\sqrt{2\pi\epsilon}}}
\end{equation}
Initially, the system wave function has two peaks centered at $\frac{1}{2}$ and $-\frac{1}{2}$ respectively without any initial momentum.
The peak at $\frac{1}{2}$ is set to be higher than the other one.
Similar to the previous example, we simply assume that the potential of the particle is given by a quadratic function
\begin{equation}
    V(x) = \frac{1}{2}x^2.
\end{equation}
This potential will force both peaks to move towards each other at the beginning.
Other parameters are the same as the previous example
\begin{equation}
    L = 400, \quad
    \omega_{\max} = 10,\quad
    \omega_c = 2.5,\quad
    \beta = 5.
\end{equation}
The range of $p$ and $q$ are chosen to be $-2$ to $2$ and the increments are $\Delta p = \Delta q = \frac{1}{32}$.
The total number of beams $K$ is 16641.
We simply choose $\bar{M}=1$ for the \cref{eq_integro_differential_equation}.
The numerical results are given in \cref{fig_example_2}.
\begin{figure}[t]
    \subfloat[$t=0$]{\includegraphics[width= 0.48\textwidth]{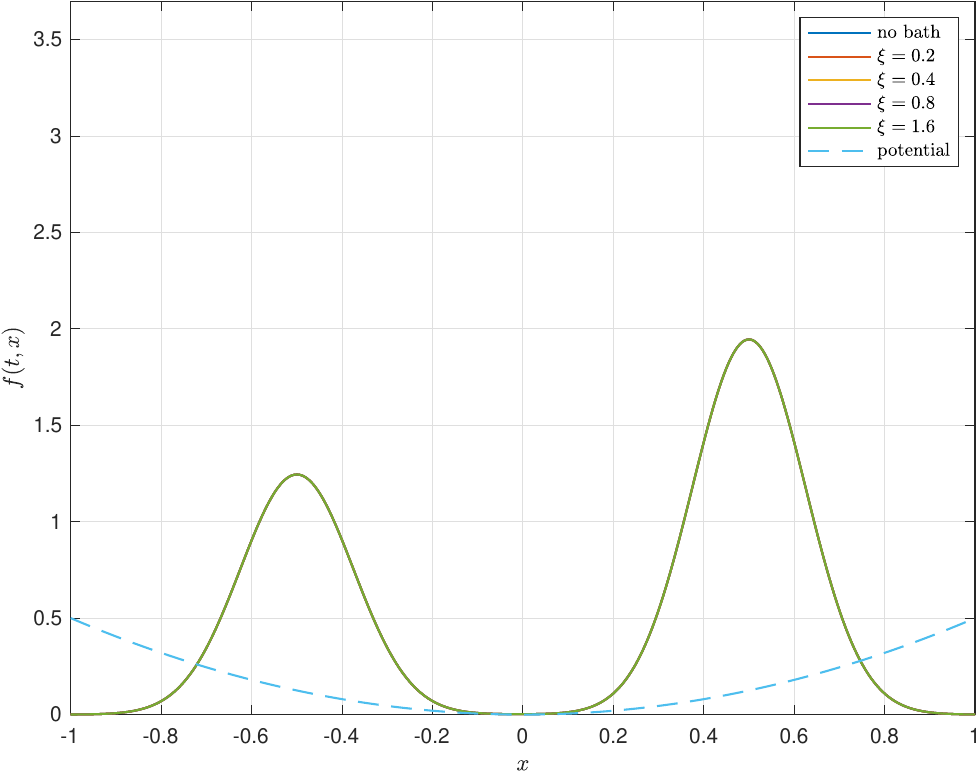}}
    \quad
    \subfloat[$t=5$]{\includegraphics[width= 0.48\textwidth]{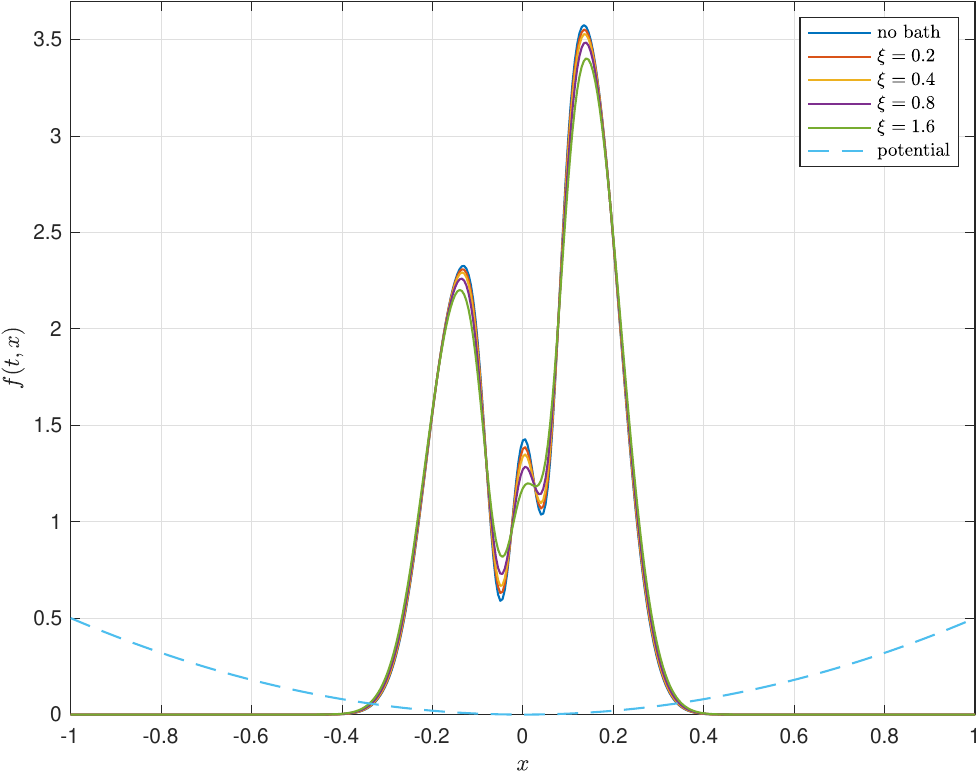}}
    \caption{Change of position probability density $f(t,x)$ for different coupling intensities $\xi$ for two-peak initial.}
    \label{fig_example_2}
\end{figure}
In the whole simulation period, the two peaks repeatedly move toward, cross and move away.
Interesting pattern appears when they meet each other, as capture 
at $t=5$.
We observe a small peak in between the main peaks.
The middle small peaks become flatter when the coupling intensity increases.
% It can be expected that a similar pattern appears some time between $t=1$ and $t=2$.
This is the result of quantum decoherence.
It is consistent with the physical fact that a particle behaves more ``classical{ly}'' when the coupling intensity increases.
Quantum decoherence is the theory of how a quantum system is converted to a ``classical'' system.
The existence of bath weakens the interference between waves.
When the interaction between the particle and bath is stronger, the interference becomes weaker.
That is why we see smaller interference waves when the coupling intensity increases in \cref{fig_example_2}.

\subsection{Double well potential}
In this subsection, we change the potential of the particle to be a double well potential
\begin{equation}
    V(x) = -x^2 + 2x^4.
\end{equation}
The initial wave function is set to be the same as the previous two-peak wave function example. The double well potential is a widely used example to study quantum tunneling effect \cite{NIETO1985doublewell,SONG2008doublewell}, meaning that in quantum mechanics a particle can pass through a high potential energy barrier which is not passable in classical mechanics due to the low energy of the particle. In our example, the double well potential we consider has two symmetric wells at $\pm{\frac{1}{2}}$ with the barrier height $\frac{1}{8}$. Other parameters are also the same as the previous example
\begin{equation}
    L = 400, \quad
    \omega_{\max} = 10,\quad
    \omega_c = 2.5,\quad
    \beta = 5.
\end{equation}
The ranges of $p$ and $q$ are $-2$ to $2$ and the increments are $\Delta p = \Delta q = \frac{1}{32}$.
Therefore, the total number of beams $K$ is 16641.
In this example, we would like to set a large coupling intensity $\xi$ and visualize the quantum decoherence.
Generally, larger $\xi$ means that we might need more terms in the infinite sum \cref{eq_integro_differential_equation} to accurately simulate the dynamics.
In order to check whether $\bar{M} = 1$ is sufficient, we first set $\xi=12.8$ and compare the results for $\bar{M}=1$ and $\bar{M}=3$.
The results are given in \cref{fig_example_34_fix_xi_M1}.
\begin{figure}[t]
    \subfloat[$t=0$]{\includegraphics[width= 0.48\textwidth]{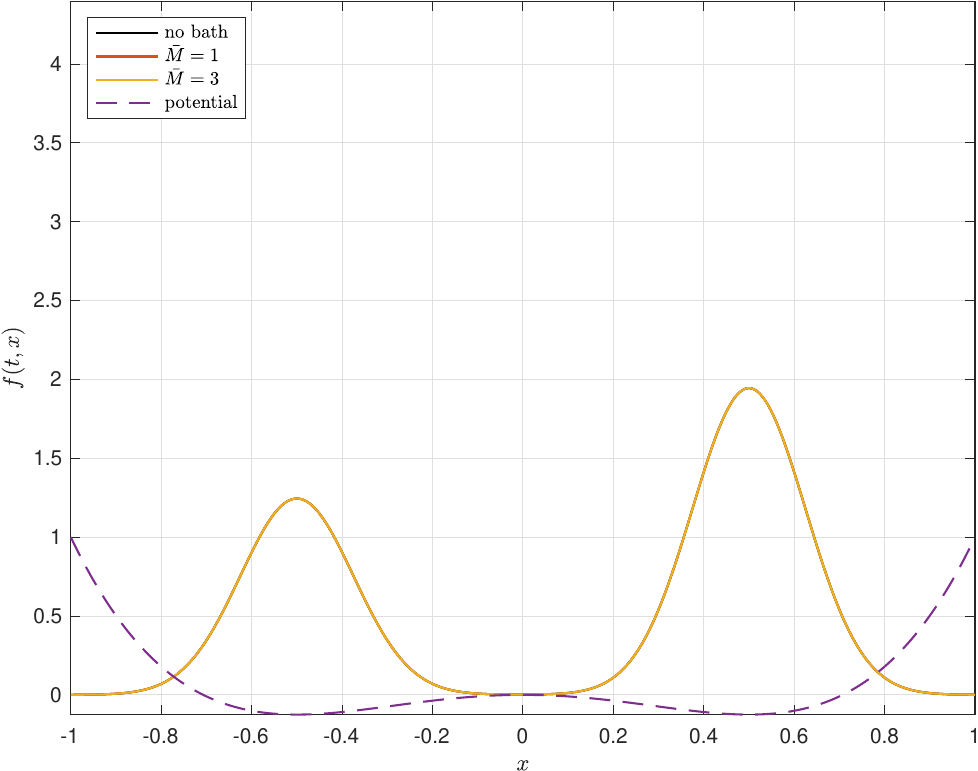}}
    \quad
    \subfloat[$t=1$]{\includegraphics[width= 0.48\textwidth]{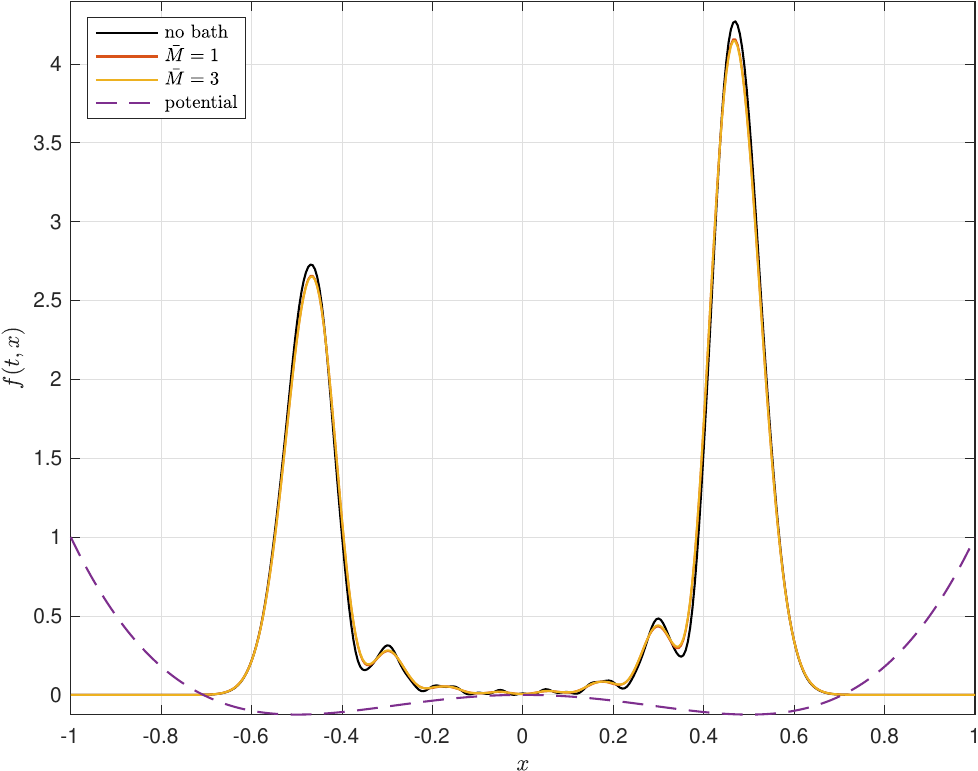}} \\
    \subfloat[$t=2$]{\includegraphics[width= 0.48\textwidth]{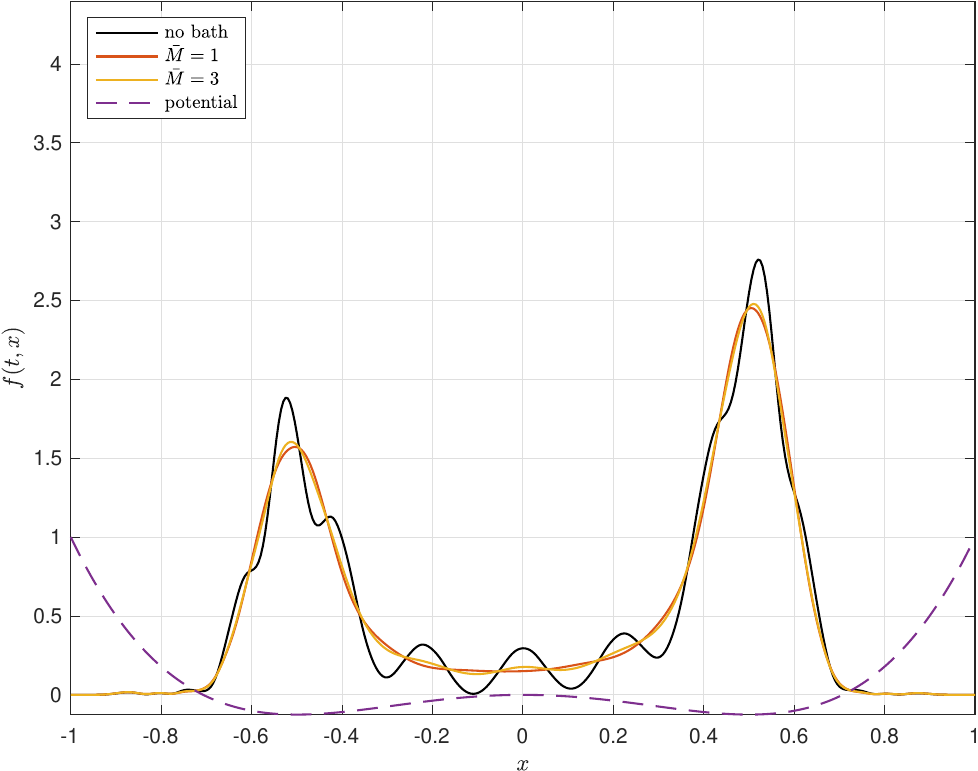}}
    \quad
    \subfloat[$t=3$]{\includegraphics[width= 0.48\textwidth]{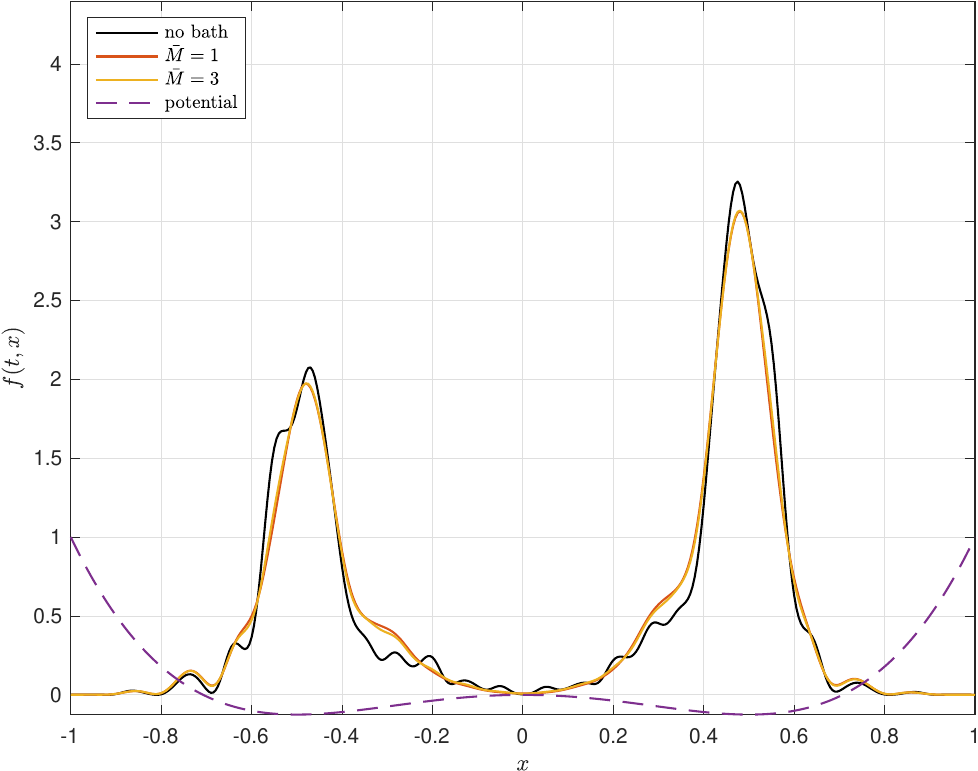}}
    \caption{Time evolution of position probability density $f(t,x)$ for different $\bar{M}$.}
    \label{fig_example_34_fix_xi_M1}
\end{figure}
In \cref{fig_example_34_fix_xi_M1}, the curves for $\bar{M}=1$ and $\bar{M}=3$ show small deviations. This result suggests that the truncation for $\bar{M}=1$ is sufficient in \cref{eq_integro_differential_equation} for $\xi \leqslant 12.8$.
Note that the computational cost for $\bar{M} = 3$ is significantly higher then $\bar{M} = 1$, since when $\bar{M} = 3$, a three-dimensional integral, instead of the one-dimensional integral in the case of $\bar{M} = 1$, needs to be computed to evaluate the right-hand side of \cref{eq_integro_differential_equation}.
To save time for our experiments, below we fix $\bar{M}=1$ and test for different coupling intensities $\xi$ no larger than $12.8$.
The results are given by \cref{fig_example_34_M1}.
\begin{figure}[t]
    \subfloat[$t=0$]{\includegraphics[width= 0.48\textwidth]{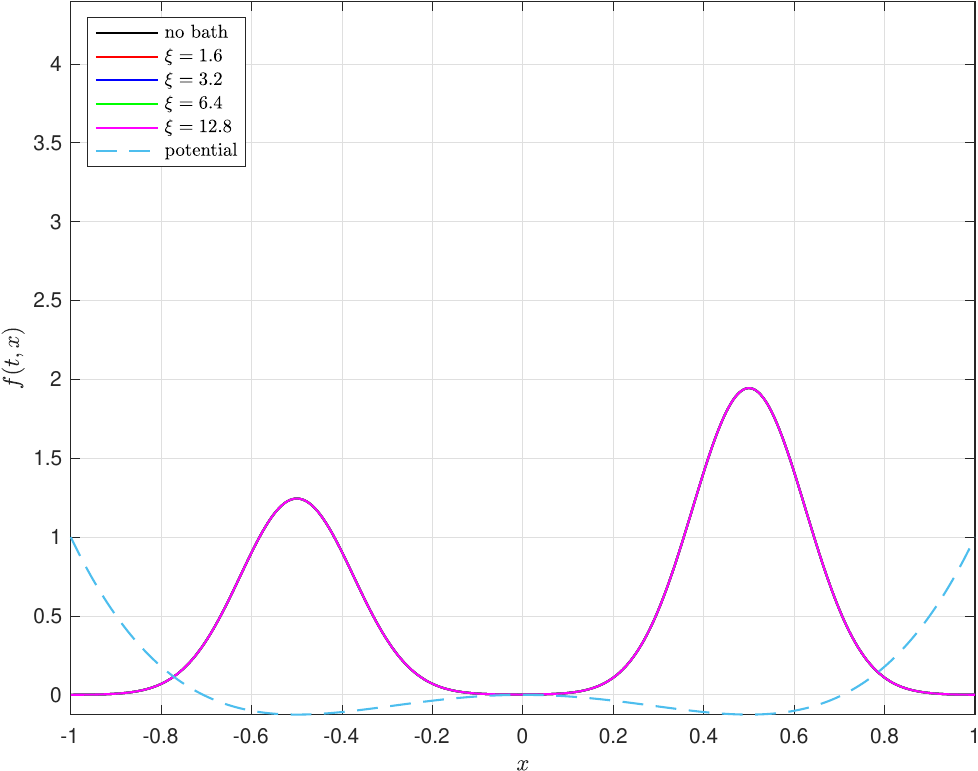}}
    \quad
    \subfloat[$t=1$]{\includegraphics[width= 0.48\textwidth]{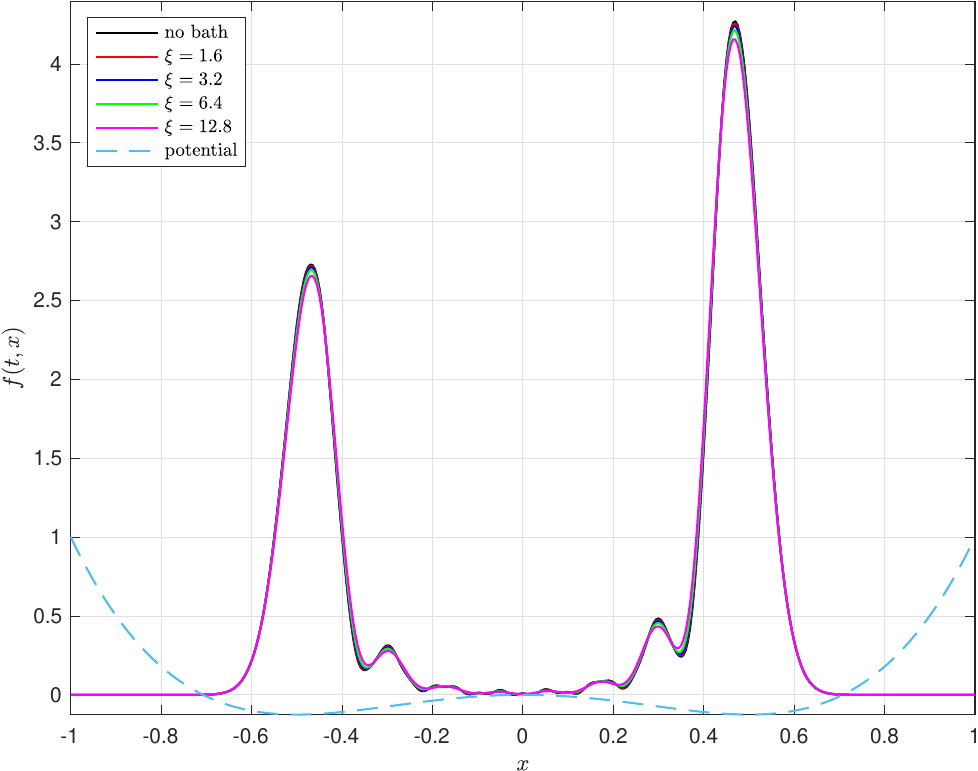}} \\
    \subfloat[$t=2$]{\includegraphics[width= 0.48\textwidth]{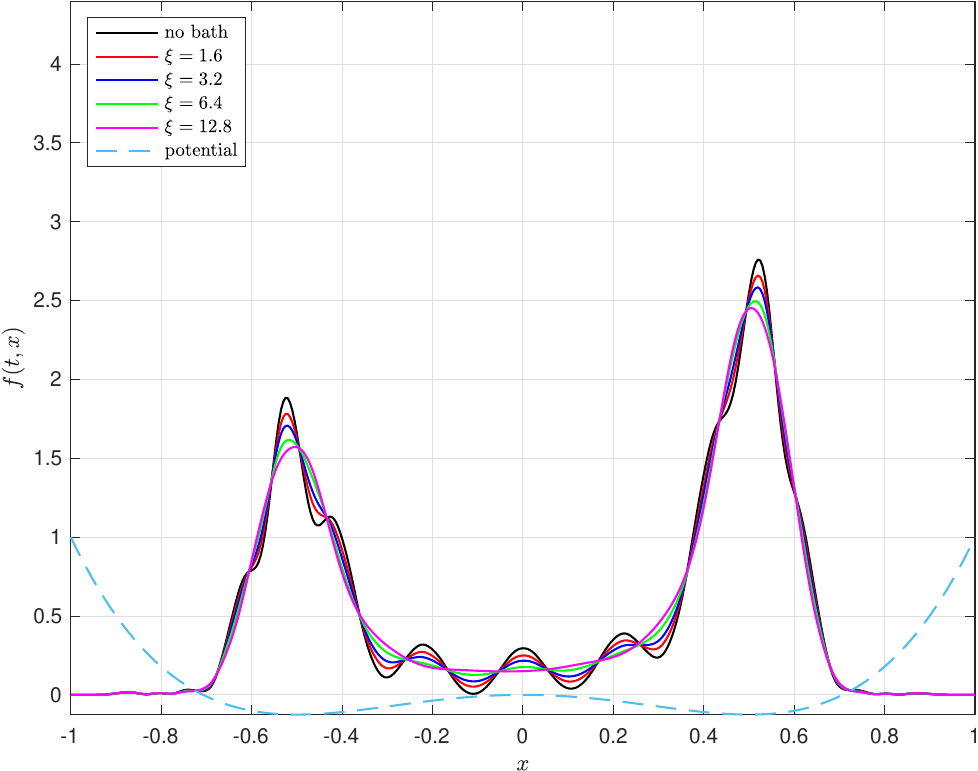}}
    \quad
    \subfloat[$t=3$]{\includegraphics[width= 0.48\textwidth]{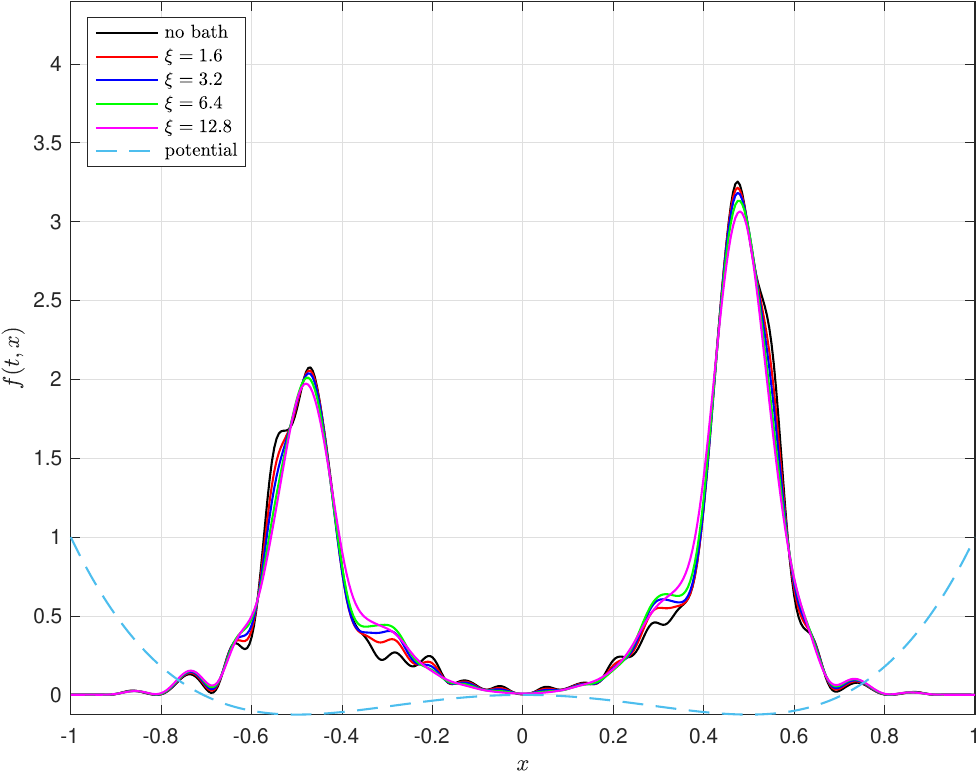}}
    \caption{Time evolution of position probability density $f(t,x)$ for different coupling intensities $\xi$ in double well potential.}
    \label{fig_example_34_M1}
\end{figure}

From this result, we again observe the effect of quantum decoherence.
An obvious difference of this example and the harmonic oscillator example is that we can observe more oscillations in the position probability density.
This result is similar to the famous double-slit experiment.
In the double-split experiment, the interference pattern, bright and dark bands, appears when electrons pass through double-slit.
In our numerical results, the position probability density also has bands structure.
The probability oscillates along the $x$ direction.
When the coupling intensity increases, the position probability density is flatter between two peaks.
When the coupling intensity increases to $\xi=12.8$, the small peaks between two main peaks almost disappear as a result of quantum decoherence.

\section{Conclusion and discussion}
\label{sec_conclusion}
We use this section to summarize the framework of our proposed method in this paper as well as remark on possible future extensions. In this paper, we propose a novel method for simulating the real-time dynamics of the reduced density operator for the Caldeira-Leggett model. Our method combines an efficient diagrammatic method for simulating open quantum systems known as the inchworm method, with the frozen Gaussian approximation technique which has shown its success in simulating closed quantum systems. 

We first formulate the reduced density operator in terms of Dyson series which allows us to express the dynamics of reduced density operator in the path integral form where each path is a combination of multiple system-associated semi-groups and position operators as perturbations. Under FGA, the wave function is approximated as a wave packet, and its dynamics under the semi-groups is described by a set of ODEs with respect to parameters of the wave packet. In addition, FGA further approximates every perturbation as a time-independent scaling factor, making the overall dynamics of reduced density operator easy to obtain. Upon FGA, we apply the inchworm method to resum the approximated path integrals to achieve faster convergence.     

% For the system-associated semi-groups, we notice that FGA provides a good approximation for the exponential operator. Additionally, in FGA, the wave function is written as the linear combination of beams, providing a simple expression to apply the position operator to the system. After writing the initial density operator as a double sum of beams, the quantity corresponding to each pair of beams bears a similar structure as the Dyson expansion and the inchworm algorithm is then applied to accelerate the convergence of the series.
% Based on such an idea, we conduct some numerical tests.

Here, we would also like to discuss some details in the methods and provide some possible ideas for future work.

\subsection{FGA and the Gaussian beam method}
In this part, we would like to remark here the similarities and differences between the FGA and the Gaussian beam (GB) method \cite{leung2007eulerian,leung2009eulerian}, and also explain the reason to choose the FGA instead of the GB method in the simulation of Caldeira-Leggett model.
Both FGA and GB use an ansatz with a linear combination of beams and differential equations for the dynamics of each beam.
There are two main differences between both methods. 
The beam widths for the Gaussian beam evolves with time while the FGA fixes the beam widths with some pre chosen parameter.
Another vital difference is that the evolution of a single beam in the Gaussian beam method satisfies the Schrödinger equation.
On the contrary, The evolution of a single beam in the FGA does not solve the equation. 
However, the integral of the beams, as a whole, solves the Schrödinger equation.

In our framework, it is also possible to use GB as the ansatz, which has a much simpler form and is also easier and cheaper for implementation.
However, there are two main reasons to choose FGA instead of GB in our framework.
The FGA achieves higher accuracy than the GB.
The accuracy of GB is $\mathcal{O}(\sqrt{\epsilon})$ while FGA has accuracy $\mathcal{O}(\epsilon)$ \cite{lu2012convergence}.
The second reason to choose FGA, also the fatal flaw of GB in our framework, is that since the beam widths of GB evolve with time, approximating the position operator $\hat{x}$ by multiplying the beam centers may therefore have large error for wide beams.
Although some reinitialization techniques are introduced for GB to keep the beam width small \cite{qian2010fast}, it would be difficult to combine reinitialization of beams and inchworm algorithm in our framework.
In the FGA, however, the beam widths are fixed.
Replacing the position operators by the beam centers will therefore yield much less error as shown in \cref{prop_position_operator_beam_center}.

\subsection{Frozen Gaussian sampling}
In this paper, we only consider a particle in one dimensional space. A natural idea to extend this work is to consider open quantum particles in two or three dimensional space. In these cases, one will encounter {the} curse of dimensionality when evaluating integrals in FGA, and thus our current framework based on numerical quadrature becomes unaffordable. The frozen Gaussian sampling (FGS) \cite{huang2023efficient} proposes a Monte Carlo based technique for the sampling of initial wave function.
With the application of Monte Carlo based sampling, it is possible to reduce the exponential number of beams to polynomial or even linear scales.

\subsection{Other open quantum system models}
In this work, we consider the Caldeira-Leggett model, where the coupling between system and bath $W_s$ is linear with position operator $\hat{x}$. This linearity is vital to the validity of our proposed method since it allows us to utilize the exponential decay of the beam magnitude to replace the position operator by the beam center with small approximation error according to \cref{prop_position_operator_beam_center}. 
Our proposed framework can also be applied to other similar open quantum system models with non-linear coupling term $W_s=F(\hat{x})$ which is at most of polynomial growth with $\hat{x}$ such as the examples listed in \cite[Section 3]{weiss2012quantum}. 
In such cases, we can derive a similar approximation as \cref{prop_position_operator_beam_center} such that we can apply our current framework by replacing the interaction operator $W_s$ by $F(Q)$ with frozen Gaussian beam center $Q$.
The non-linear coupling models will be considered in our future work. { Furthermore, the Wick's theorem plays a key role in developing inchworm method, which essentially utilizes the fact that Caldeira-Leggett model has harmonic bath which is quadratic, and eventually reduces higher-order bath correlations to two-point correlation functions. For other bath (e.g, non-Gaussian bath), Wick's theorem might not hold. Nevertheless, our current framework has the potential to extend beyond the Caldeira-Leggett model. For example, if the bath influence functional can be formulated as sum of higher-order correlations (such as four-point interactions), there is also hope to develop the corresponding inchworm method to regroup the Dyson series and reduce the number of quantum diagrams. Also, as pointed out previously, this work also provides a framework to simulate the dynamics of quantum particles in open systems combining FGA and direct summation of diagrams in Dyson series without inchworm method, which is applicable to the systems where the bath influence functional does not have any structure.}

\bibliographystyle{quantum}
\bibliography{main}

\begin{thebibliography}{10}

\bibitem{leforestier1991comparison}
Claude Leforestier, RH~Bisseling, Charly Cerjan, MD~Feit, Rich Friesner, A~Guldberg, A~Hammerich, G~Jolicard, W~Karrlein, H-D Meyer, et~al.
\newblock ``A comparison of different propagation schemes for the time dependent {S}chrödinger equation''.
\newblock \href{https://dx.doi.org/10.1016/0021-9991(91)90137-A}{J. Comput. Phys. {\bf 94}, 59--80}~(1991).

\bibitem{iitaka1994solving}
Toshiaki Iitaka.
\newblock ``Solving the time-dependent {S}chrödinger equation numerically''.
\newblock \href{https://dx.doi.org/10.1103/PhysRevE.49.4684}{Phys. Rev. E {\bf 49}, 4684}~(1994).

\bibitem{bao2002time}
Weizhu Bao, Shi Jin, and Peter~A Markowich.
\newblock ``On time-splitting spectral approximations for the {S}chrödinger equation in the semiclassical regime''.
\newblock \href{https://dx.doi.org/10.1006/jcph.2001.6956}{J. Comput. Phys. {\bf 175}, 487--524}~(2002).

\bibitem{bao2003numerical}
Weizhu Bao, Shi Jin, and Peter~A Markowich.
\newblock ``Numerical study of time-splitting spectral discretizations of nonlinear {S}chrödinger equations in the semiclassical regimes''.
\newblock \href{https://dx.doi.org/10.1137/S1064827501393253}{SIAM J. Sci. Comput. {\bf 25}, 27--64}~(2003).

\bibitem{vandijk2007accurate}
W~Van~Dijk and FM~Toyama.
\newblock ``Accurate numerical solutions of the time-dependent {S}chrödinger equation''.
\newblock \href{https://dx.doi.org/10.1103/PhysRevE.75.036707}{Phys. Rev. E {\bf 75}, 036707}~(2007).

\bibitem{heller1975time}
Eric~J Heller.
\newblock ``Time-dependent approach to semiclassical dynamics''.
\newblock \href{https://dx.doi.org/10.1063/1.430620}{J. Chem. Phys. {\bf 62}, 1544--1555}~(1975).

\bibitem{leung2007eulerian}
Shingyu Leung, Jianliang Qian, and Robert Burridge.
\newblock ``Eulerian {G}aussian beams for high-frequency wave propagation''.
\newblock \href{https://dx.doi.org/10.1190/1.2752136}{Geophysics {\bf 72}, SM61--SM76}~(2007).

\bibitem{leung2009eulerian}
Shingyu Leung and Jianliang Qian.
\newblock ``Eulerian {G}aussian beams for {S}chrödinger equations in the semi-classical regime''.
\newblock \href{https://dx.doi.org/10.1016/j.jcp.2009.01.007}{J. Comput. Phys. {\bf 228}, 2951--2977}~(2009).

\bibitem{heller1981frozen}
Eric~J Heller.
\newblock ``Frozen {G}aussians: A very simple semiclassical approximation''.
\newblock \href{https://dx.doi.org/10.1063/1.442382}{J. Chem. Phys. {\bf 75}, 2923--2931}~(1981).

\bibitem{lu2011frozen}
Jianfeng Lu and Xu~Yang.
\newblock ``Frozen {G}aussian approximation for high frequency wave propagation''.
\newblock \href{https://dx.doi.org/10.4310/cms.2011.v9.n3.a2}{Commun. Math. Sci. {\bf 9}, 663--683}~(2011).

\bibitem{lu2012convergence}
Jianfeng Lu and Xu~Yang.
\newblock ``Convergence of frozen {G}aussian approximation for high-frequency wave propagation''.
\newblock \href{https://dx.doi.org/10.1002/cpa.21384}{Commun. Pure Appl. Math. {\bf 65}, 759--789}~(2012).

\bibitem{lu2012frozen}
Jianfeng Lu and Xu~Yang.
\newblock ``Frozen {G}aussian approximation for general linear strictly hyperbolic systems: Formulation and {E}ulerian methods''.
\newblock \href{https://dx.doi.org/10.1137/10081068X}{Multiscale Model. Simul. {\bf 10}, 451--472}~(2012).

\bibitem{lu2018frozen}
Jianfeng Lu and Zhennan Zhou.
\newblock ``Frozen {G}aussian approximation with surface hopping for mixed quantum-classical dynamics: A mathematical justification of fewest switches surface hopping algorithms''.
\newblock \href{https://dx.doi.org/10.1090/mcom/3310}{Math. Comput. {\bf 87}, 2189--2232}~(2018).

\bibitem{delgadillo2018frozen}
Ricardo Delgadillo, Jianfeng Lu, and Xu~Yang.
\newblock ``Frozen {G}aussian approximation for high frequency wave propagation in periodic media''.
\newblock \href{https://dx.doi.org/10.3233/ASY-181479}{Asymptot. Anal. {\bf 110}, 113--135}~(2018).

\bibitem{herman1984semiclassical}
Michael~F Herman and Edward Kluk.
\newblock ``A semiclassical justification for the use of non-spreading wavepackets in dynamics calculations''.
\newblock \href{https://dx.doi.org/10.1016/0301-0104(84)80039-7}{Chem. Phys. {\bf 91}, 27--34}~(1984).

\bibitem{swart2009mathematical}
Torben Swart and Vidian Rousse.
\newblock ``A mathematical justification for the {H}erman-{K}luk propagator''.
\newblock \href{https://dx.doi.org/10.1007/s00220-008-0681-4}{Commun. Math. Phys. {\bf 286}, 725--750}~(2009).

\bibitem{knill1997theory}
Emanuel Knill and Raymond Laflamme.
\newblock ``Theory of quantum error-correcting codes''.
\newblock \href{https://dx.doi.org/10.1103/PhysRevA.55.900}{Phys. Rev. A {\bf 55}, 900}~(1997).

\bibitem{nielsen2002quantum}
Michael~A Nielsen and Isaac Chuang.
\newblock ``Quantum computation and quantum information''.
\newblock \href{https://dx.doi.org/10.1017/CBO9780511976667}{Cambridge university press}. ~(2002).

\bibitem{breuer2002theory}
Heinz-Peter Breuer and Francesco Petruccione.
\newblock ``The theory of open quantum systems''.
\newblock \href{https://dx.doi.org/10.1093/acprof:oso/9780199213900.001.0001}{Oxford University Press on Demand}. ~(2002).

\bibitem{grigorescu1998decoherence}
M~Grigorescu.
\newblock ``Decoherence and dissipation in quantum two-state systems''.
\newblock \href{https://dx.doi.org/10.1016/S0378-4371(98)00076-4}{Phys. A: Stat. Mech. Appl. {\bf 256}, 149--162}~(1998).

\bibitem{schlosshauer2019quantum}
Maximilian Schlosshauer.
\newblock ``Quantum decoherence''.
\newblock \href{https://dx.doi.org/10.1016/j.physrep.2019.10.001}{Phys. Rep. {\bf 831}, 1--57}~(2019).

\bibitem{nakajima1958quantum}
Sadao Nakajima.
\newblock ``On quantum theory of transport phenomena: Steady diffusion''.
\newblock \href{https://dx.doi.org/10.1143/PTP.20.948}{Prog. Theor. Phys. {\bf 20}, 948--959}~(1958).

\bibitem{zwanzig1960ensemble}
Robert Zwanzig.
\newblock ``Ensemble method in the theory of irreversibility''.
\newblock \href{https://dx.doi.org/10.1063/1.1731409}{J. Chem. Phys. {\bf 33}, 1338--1341}~(1960).

\bibitem{lindblad1976generators}
Goran Lindblad.
\newblock ``On the generators of quantum dynamical semigroups''.
\newblock \href{https://dx.doi.org/10.1007/BF01608499}{Commun. Math. Phys. {\bf 48}, 119--130}~(1976).

\bibitem{feynman1948space}
R~P Feynman.
\newblock ``Space-time approach to non-relativistic quantum mechanics''.
\newblock \href{https://dx.doi.org/10.1103/RevModPhys.20.367}{Rev. Mod. Phys. {\bf 20}, 367--387}~(1948).

\bibitem{makri1992improved}
Nancy Makri.
\newblock ``Improved {F}eynman propagators on a grid and non-adiabatic corrections within the path integral framework''.
\newblock \href{https://dx.doi.org/10.1016/0009-2614(92)85654-S}{Chem. Phys. Lett. {\bf 193}, 435--445}~(1992).

\bibitem{makri1995numerical}
Nancy Makri.
\newblock ``Numerical path integral techniques for long time dynamics of quantum dissipative systems''.
\newblock \href{https://dx.doi.org/10.1063/1.531046}{J. Math. Phys. {\bf 36}, 2430--2457}~(1995).

\bibitem{makri2014blip}
Nancy Makri.
\newblock ``Blip decomposition of the path integral: Exponential acceleration of real-time calculations on quantum dissipative systems''.
\newblock \href{https://dx.doi.org/10.1063/1.4896736}{J. Chem. Phys. {\bf 141}, 134117}~(2014).

\bibitem{wang2022differential}
Geshuo Wang and Zhenning Cai.
\newblock ``Differential equation based path integral for open quantum systems''.
\newblock \href{https://dx.doi.org/10.1137/21M1439833}{SIAM J. Sci. Comput. {\bf 44}, B771--B804}~(2022).

\bibitem{makri2024kink}
Nancy Makri.
\newblock ``Kink sum for long-memory small matrix path integral dynamics''.
\newblock \href{https://dx.doi.org/10.1021/acs.jpcb.3c08282}{J. Phys. Chem. B}~(2024).

\bibitem{makri2020smallMatrixPath}
Nancy Makri.
\newblock ``Small matrix path integral for system-bath dynamics''.
\newblock \href{https://dx.doi.org/10.1021/acs.jctc.0c00039}{J. Chem. Theory Comput. {\bf 16}, 4038--4049}~(2020).

\bibitem{makri2021smallMatrixPathIntegralExtended}
Nancy Makri.
\newblock ``Small matrix path integral with extended memory''.
\newblock \href{https://dx.doi.org/10.1021/acs.jctc.0c00987}{J. Chem. Theory Comput. {\bf 17}, 1--6}~(2021).

\bibitem{wang2024tree}
Geshuo Wang and Zhenning Cai.
\newblock ``Tree-based implementation of the small matrix path integral for system-bath dynamics''.
\newblock \href{https://dx.doi.org/10.4208/cicp.OA-2023-0329}{Commun. Comput. Phys. {\bf 36}, 389--418}~(2024).

\bibitem{cerrillo2014nonmarkovian}
Javier Cerrillo and Jianshu Cao.
\newblock ``Non-{M}arkovian dynamical maps: numerical processing of open quantum trajectories''.
\newblock \href{https://dx.doi.org/10.1103/PhysRevLett.112.110401}{Phys. Rev. Lett. {\bf 112}, 110401}~(2014).

\bibitem{dyson1949radiation}
Freeman~J Dyson.
\newblock ``The radiation theories of {T}omonaga, {S}chwinger, and {F}eynman''.
\newblock \href{https://dx.doi.org/10.1103/PhysRev.75.486}{Phys. Rev. {\bf 75}, 486}~(1949).

\bibitem{xu2018convergence}
Meng Xu, Yaming Yan, Yanying Liu, and Qiang Shi.
\newblock ``Convergence of high order memory kernels in the nakajima-zwanzig generalized master equation and rate constants: Case study of the spin-boson model''.
\newblock \href{https://dx.doi.org/10.1063/1.5022761}{J. Chem. Phys.{\bf 148}}~(2018).

\bibitem{loh1990sign}
EY~Jr Loh, JE~Gubernatis, RT~Scalettar, SR~White, DJ~Scalapino, and RL~Sugar.
\newblock ``Sign problem in the numerical simulation of many-electron systems''.
\newblock \href{https://dx.doi.org/10.1103/PhysRevB.41.9301}{Phys. Rev. B {\bf 41}, 9301}~(1990).

\bibitem{cai2023numerical}
Zhenning Cai, Jianfeng Lu, and Siyao Yang.
\newblock ``Numerical analysis for inchworm {M}onte {C}arlo method: Sign problem and error growth''.
\newblock \href{https://dx.doi.org/10.1090/mcom/3785}{Math. Comput. {\bf 92}, 1141--1209}~(2023).

\bibitem{chen2017inchwormITheory}
Hsing-Ta Chen, Guy Cohen, and David~R Reichman.
\newblock ``Inchworm {M}onte {C}arlo for exact non-adiabatic dynamics. i. theory and algorithms''.
\newblock \href{https://dx.doi.org/10.1063/1.4974328}{J. Chem. Phys. {\bf 146}, 054105}~(2017).

\bibitem{chen2017inchwormIIBenchmarks}
Hsing-Ta Chen, Guy Cohen, and David~R Reichman.
\newblock ``Inchworm {M}onte {C}arlo for exact non-adiabatic dynamics. ii. benchmarks and comparison with established methods''.
\newblock \href{https://dx.doi.org/10.1063/1.4974329}{J. Chem. Phys. {\bf 146}, 054106}~(2017).

\bibitem{Antipov2017currents}
Andrey~E Antipov, Qiaoyuan Dong, Joseph Kleinhenz, Guy Cohen, and Emanuel Gull.
\newblock ``Currents and green’s functions of impurities out of equilibrium: Results from inchworm quantum {M}onte {C}arlo''.
\newblock \href{https://dx.doi.org/10.1103/PhysRevB.95.085144}{Phys. Rev. B {\bf 95}, 085144}~(2017).

\bibitem{cai2020inchworm}
Zhenning Cai, Jianfeng Lu, and Siyao Yang.
\newblock ``Inchworm {M}onte {C}arlo method for open quantum systems''.
\newblock \href{https://dx.doi.org/10.1002/cpa.21888}{Commun. Pure Appl. Math. {\bf 73}, 2430--2472}~(2020).

\bibitem{erpenbeck2023quantum}
André Erpenbeck, Emanuel Gull, and Guy Cohen.
\newblock ``Quantum {M}onte {C}arlo method in the steady state''.
\newblock \href{https://dx.doi.org/10.1103/PhysRevLett.130.186301}{Phys. Rev. Lett. {\bf 130}, 186301}~(2023).

\bibitem{prokof2007bold}
Nikolay Prokof’ev and Boris Svistunov.
\newblock ``Bold diagrammatic {M}onte {C}arlo technique: When the sign problem is welcome''.
\newblock \href{https://dx.doi.org/10.1103/PhysRevLett.99.250201}{Phys. Rev. Lett. {\bf 99}, 250201}~(2007).

\bibitem{prokof2008bold}
NV~Prokof’ev and BV~Svistunov.
\newblock ``Bold diagrammatic {M}onte {C}arlo: A generic sign-problem tolerant technique for polaron models and possibly interacting many-body problems''.
\newblock \href{https://dx.doi.org/10.1103/PhysRevB.77.125101}{Phys. Rev. B {\bf 77}, 125101}~(2008).

\bibitem{yang2021inclusion}
Siyao Yang, Zhenning Cai, and Jianfeng Lu.
\newblock ``Inclusion–exclusion principle for open quantum systems with bosonic bath''.
\newblock \href{https://dx.doi.org/10.1088/1367-2630/ac02e1}{New J. Phys. {\bf 23}, 063049}~(2021).

\bibitem{cai2022fast}
Zhenning Cai, Jianfeng Lu, and Siyao Yang.
\newblock ``Fast algorithms of bath calculations in simulations of quantum system-bath dynamics''.
\newblock \href{https://dx.doi.org/10.1016/j.cpc.2022.108417}{Comput. Phys. Commun. {\bf 277}, 108417}~(2022).

\bibitem{cai2023bold}
Zhenning Cai, Geshuo Wang, and Siyao Yang.
\newblock ``The bold-thin-bold diagrammatic {M}onte {C}arlo method for open quantum systems''.
\newblock \href{https://dx.doi.org/10.1137/22M1499297}{SIAM J. Sci. Comput. {\bf 45}, A1812--A1843}~(2023).

\bibitem{chakravarty1984dynamics}
Sudip Chakravarty and Anthony~J Leggett.
\newblock ``Dynamics of the two-state system with {O}hmic dissipation''.
\newblock \href{https://dx.doi.org/10.1103/PhysRevLett.52.5}{Phys. Rev. Lett. {\bf 52}, 5}~(1984).

\bibitem{thorwart2004dynamics}
Michael Thorwart, E~Paladino, and Milena Grifoni.
\newblock ``Dynamics of the spin-boson model with a structured environment''.
\newblock \href{https://dx.doi.org/10.1016/j.chemphys.2003.10.007}{Chem. Phys. {\bf 296}, 333--344}~(2004).

\bibitem{xu2023performance}
Meng Xu and Joachim Ankerhold.
\newblock ``About the performance of perturbative treatments of the spin-boson dynamics within the hierarchical equations of motion approach''.
\newblock \href{https://dx.doi.org/10.1140/epjs/s11734-023-01000-6}{Eur. Phys. J. Spec. Top. {\bf 232}, 3209--3217}~(2023).

\bibitem{caldeira1981influence}
Amir~O Caldeira and Anthony~J Leggett.
\newblock ``Influence of dissipation on quantum tunneling in macroscopic systems''.
\newblock \href{https://dx.doi.org/10.1103/PhysRevLett.46.211}{Phys. Rev. Lett. {\bf 46}, 211}~(1981).

\bibitem{caldeira1983path}
Amir~O Caldeira and Anthony~J Leggett.
\newblock ``Path integral approach to quantum {B}rownian motion''.
\newblock \href{https://dx.doi.org/10.1016/0378-4371(83)90013-4}{Phys. A {\bf 121}, 587--616}~(1983).

\bibitem{caldeira1985influence}
Amir~O Caldeira and Anthony~J Leggett.
\newblock ``Influence of damping on quantum interference: An exactly soluble model''.
\newblock \href{https://dx.doi.org/10.1103/PhysRevA.31.1059}{Phys. Rev. A {\bf 31}, 1059}~(1985).

\bibitem{makri2018modular}
Nancy Makri.
\newblock ``Modular path integral methodology for real-time quantum dynamics''.
\newblock \href{https://dx.doi.org/10.1063/1.5058223}{J. Chem. Phys. {\bf 149}, 214108}~(2018).

\bibitem{wang2023real}
Geshuo Wang and Zhenning Cai.
\newblock ``Real-time simulation of open quantum spin chains with the inchworm method''.
\newblock \href{https://dx.doi.org/10.1021/acs.jctc.3c00751}{J. Chem. Theory Comput. {\bf 19}, 8523--8540}~(2023).

\bibitem{sun2024simulation}
Yixiao Sun, Geshuo Wang, and Zhenning Cai.
\newblock ``Simulation of spin chains with off-diagonal coupling using the inchworm method''.
\newblock \href{https://dx.doi.org/10.1021/acs.jctc.4c00864}{J. Chem. Theory Comput. {\bf 20}, 9269--9758}~(2024).

\bibitem{yan2021efficient}
Yaming Yan, Meng Xu, Tianchu Li, and Qiang Shi.
\newblock ``Efficient propagation of the hierarchical equations of motion using the tucker and hierarchical tucker tensors''.
\newblock \href{https://dx.doi.org/10.1063/5.0050720}{J. Chem. Phys.{\bf 154}}~(2021).

\bibitem{erpenbeck2023tensor}
A~Erpenbeck, W-T Lin, T~Blommel, L~Zhang, S~Iskakov, L~Bernheimer, Y~Núñez-Fernández, G~Cohen, O~Parcollet, X~Waintal, and E~Gull.
\newblock ``Tensor train continuous time solver for quantum impurity models''.
\newblock \href{https://dx.doi.org/10.1103/PhysRevB.107.245135}{Phys. Rev. B {\bf 107}, 245135}~(2023).

\bibitem{ferialdi2017dissipation}
L~Ferialdi.
\newblock ``Dissipation in the caldeira-leggett model''.
\newblock \href{https://dx.doi.org/10.1103/PhysRevA.95.052109}{Phys. Rev. A {\bf 95}, 052109}~(2017).

\bibitem{fisher1985dissipative}
Matthew~PA Fisher and Alan~T Dorsey.
\newblock ``Dissipative quantum tunneling in a biased double-well system at finite temperatures''.
\newblock \href{https://dx.doi.org/10.1103/PhysRevLett.54.1609}{Phys. Rev. Lett. {\bf 54}, 1609}~(1985).

\bibitem{leggett1987dynamics}
Anthony~J Leggett, Sudip Chakravarty, Alan~T Dorsey, Matthew~PA Fisher, Anupam Garg, and Wilhelm Zwerger.
\newblock ``Dynamics of the dissipative two-state system''.
\newblock \href{https://dx.doi.org/10.1103/RevModPhys.59.1}{Rev. Mod. Phys. {\bf 59}, 1}~(1987).

\bibitem{hanggi1986escape}
Peter Hänggi.
\newblock ``Escape from a metastable state''.
\newblock \href{https://dx.doi.org/10.1007/BF01010843}{J. Stat. Phys. {\bf 42}, 105--148}~(1986).

\bibitem{thompson1999influence}
Keiran Thompson and Nancy Makri.
\newblock ``Influence functionals with semiclassical propagators in combined forward–backward time''.
\newblock \href{https://dx.doi.org/10.1063/1.478011}{J. Chem. Phys. {\bf 110}, 1343--1353}~(1999).

\bibitem{wick1950evaluation}
Gian-Carlo Wick.
\newblock ``The evaluation of the collision matrix''.
\newblock \href{https://dx.doi.org/10.1103/PhysRev.80.268}{Phys. Rev. {\bf 80}, 268}~(1950).

\bibitem{dauria2011special}
Riccardo D'Auria and Mario Trigiante.
\newblock ``From special relativity to {F}eynman diagrams''.
\newblock \href{https://dx.doi.org/10.1007/978-3-319-22014-7}{Springer}. ~(2011).

\bibitem{prokof1998polaron}
Nikolai~V Prokof’ev and Boris~V Svistunov.
\newblock ``Polaron problem by diagrammatic quantum {M}onte {C}arlo''.
\newblock \href{https://dx.doi.org/10.1103/PhysRevLett.81.2514}{Phys. Rev. Lett. {\bf 81}, 2514}~(1998).

\bibitem{werner2009diagrammatic}
Philipp Werner, Takashi Oka, and Andrew~J Millis.
\newblock ``Diagrammatic {M}onte {C}arlo simulation of nonequilibrium systems''.
\newblock \href{https://dx.doi.org/10.1103/PhysRevB.79.035320}{Phys. Rev. B {\bf 79}, 035320}~(2009).

\bibitem{gull2010bold}
Emanuel Gull, David~R Reichman, and Andrew~J Millis.
\newblock ``Bold-line diagrammatic {M}onte {C}arlo method: General formulation and application to expansion around the noncrossing approximation''.
\newblock \href{https://dx.doi.org/doi.org/10.1103/PhysRevB.82.075109}{Phys. Rev. B {\bf 82}, 075109}~(2010).

\bibitem{kulagin2013boldDiagrammaticMonteCarloMethod}
SA~Kulagin, N~Prokof’ev, OA~Starykh, B~Svistunov, and CN~Varney.
\newblock ``Bold diagrammatic {M}onte {C}arlo method applied to fermionized frustrated spins''.
\newblock \href{https://dx.doi.org/10.1103/PhysRevLett.110.070601}{Phys. Rev. Lett. {\bf 110}, 070601}~(2013).

\bibitem{kulagin2013boldMonteCarloTechnique}
SA~Kulagin, Nikolay Prokof’ev, OA~Starykh, Boris Svistunov, and Christopher~N Varney.
\newblock ``Bold diagrammatic {M}onte {C}arlo technique for frustrated spin systems''.
\newblock \href{https://dx.doi.org/10.1103/PhysRevB.87.024407}{Phys. Rev. B {\bf 87}, 024407}~(2013).

\bibitem{cohen2015taming}
Guy Cohen, Emanuel Gull, David~R Reichman, and Andrew~J Millis.
\newblock ``Taming the dynamical sign problem in real-time evolution of quantum many-body problems''.
\newblock \href{https://dx.doi.org/10.1103/PhysRevLett.115.266802}{Phys. Rev. Lett. {\bf 115}, 266802}~(2015).

\bibitem{li2019bold}
Yingzhou Li and Jianfeng Lu.
\newblock ``Bold diagrammatic {M}onte {C}arlo in the lens of stochastic iterative methods''.
\newblock \href{https://dx.doi.org/10.1093/imatrm/tnz001}{Trans. Math. Appl. {\bf 3}, tnz001}~(2019).

\bibitem{makri1999linear}
Nancy Makri.
\newblock ``The linear response approximation and its lowest order corrections: An influence functional approach''.
\newblock \href{https://dx.doi.org/10.1021/jp9847540}{J. Phys. Chem. B {\bf 103}, 2823--2829}~(1999).

\bibitem{NIETO1985doublewell}
Michael~Martin Nieto, Vincent~P Gutschick, Carl~M Bender, Fred Cooper, and D~Strottman.
\newblock ``Resonances in quantum mechanical tunneling''.
\newblock \href{https://dx.doi.org/10.1016/0370-2693(85)90292-8}{Phys. Lett. B {\bf 163}, 336--342}~(1985).

\bibitem{SONG2008doublewell}
Dae-Yup Song.
\newblock ``Tunneling and energy splitting in an asymmetric double-well potential''.
\newblock \href{https://dx.doi.org/10.1016/j.aop.2008.09.004}{Ann. Phys. {\bf 323}, 2991--2999}~(2008).

\bibitem{qian2010fast}
Jianliang Qian and Lexing Ying.
\newblock ``Fast {G}aussian wavepacket transforms and {G}aussian beams for the {S}chrödinger equation''.
\newblock \href{https://dx.doi.org/10.1016/j.jcp.2010.06.043}{J. Comput. Phys. {\bf 229}, 7848--7873}~(2010).

\bibitem{huang2023efficient}
Zhen Huang, Limin Xu, and Zhennan Zhou.
\newblock ``Efficient frozen {G}aussian sampling algorithms for nonadiabatic quantum dynamics at metal surfaces''.
\newblock \href{https://dx.doi.org/10.1016/j.jcp.2022.111771}{J. Comput. Phys. {\bf 474}, 111771}~(2023).

\bibitem{weiss2012quantum}
Ulrich Weiss.
\newblock ``Quantum dissipative systems''.
\newblock \href{https://dx.doi.org/10.1142/12402}{World scientific, River Edge, NJ}. ~(2012).

\end{thebibliography}

\end{document}